\documentclass{article}
\usepackage{graphicx} 

\title{Contracting with discretionary bonuses}
\author{Guillermo Alonso Alvarez\footnote{Department of Mathematics, University of Michigan. guialv@umich.edu.}, Ibrahim Ekren\footnote{Department of Mathematics, University of Michigan. iekren@umich.edu. I. Ekren is partially supported by the NSF grant DMS-2406240}, Liwei Huang\footnote{Department of Mathematics, University of Michigan. huanglw@umich.edu.}}
\date{\today}
\usepackage{dsfont}

\usepackage{verbatim}

\usepackage{amssymb,amsfonts,amsmath,amsthm}
\usepackage{float}
\usepackage{xcolor}
\usepackage{bbm}
\usepackage{enumitem}
\usepackage{bbm}
\usepackage{url}

\usepackage{caption}
\usepackage{subcaption}

\usepackage{titlesec}

\titleformat{\paragraph}
{\normalfont\normalsize\bfseries}{\theparagraph}{1em}{}
\titlespacing*{\paragraph}
{0pt}{3.25ex plus 1ex minus .2ex}{1.5ex plus .2ex}

\DeclareMathOperator*{\argmax}{arg\,max}

\newtheorem{theorem}{Theorem}

\newtheorem{lemma}{Lemma}
\newtheorem{assumption}{Assumption}

\newtheorem{remark}{Remark}

\newtheorem{Algorithm}{Algorithm}

\newcommand\EE {\mathbb E}
\newcommand\FF {\mathbb F}
\newcommand\NN {\mathbb N}
\newcommand\RR {\mathbb R}
\newcommand\PP {\mathbb P}

\def\e{\epsilon}

\def\qed{\hskip6pt\vrule height6pt width5pt depth1pt}

\begin{document}
\maketitle
\begin{abstract}
We study a continuous time contracting model in which a principal hires a risk averse agent to manage a project over a finite horizon and provides sequential payments whose timing is endogenously determined. The resulting nonzero-sum interaction between the principal and the agent is reformulated as a mixed control and stopping problem. Using numerical simulations, we investigate how factors such as the relative impatience of the parties and the number of bonus payments influence the principal’s value and the structure of the optimal bonus payment scheme. A notable finding is that, in some contractual environments, the principal optimally offers a sign-on bonus to front-load incentives.
\end{abstract}

\section{Introduction}
Principal–agent theory formalizes how a decision-maker (the principal, she) designs incentives to motivate another party (the agent, he) to act in her best interest when the agent’s actions cannot be directly observed. Such information asymmetry creates a moral hazard, requiring the principal to craft a contract that aligns the agent’s private incentives with her own objectives. From a mathematical standpoint, this interaction can be modeled as a bi-level optimization or a Stackelberg game, where the principal anticipates the agent’s optimal response to any given contract. This framework has found wide-ranging applications—from managerial compensation and executive pay to insurance, regulation, and financial contracting—making it a central tool for studying delegation and incentive alignment under uncertainty.

The continuous-time formulation of the principal–agent problem originates from the work of Holmstrom and Milgrom in \cite{holmstrom1987aggregation}, who show that under exponential utility and when the agent’s effort affects only the drift of output, the optimal contract is linear. Their framework revealed the analytical tractability of continuous-time models and inspired multiple extensions, including those by \cite{schattler1993first}, \cite{muller2000asymptotic}, and \cite{hellwig2002discrete}. A major advance was later achieved by \cite{sannikov2008continuous}, who recast the problem as a stochastic control model with continuous incentives and a terminal lump-sum payment, thereby enabling the use of dynamic programming methods. This approach was subsequently unified and extended by \cite{cvitanic2018dynamic} through a BSDE formulation, establishing a general connection between continuous-time contracting and stochastic control theory. We highlight \cite{possamai2024there}, which provides a detailed discussion and significant extension of Sannikov’s model.

Building on this line of research, subsequent studies have connected both static and dynamic principal–agent frameworks to the design of bonus payments in optimal contracting. Early extensions of Holmström and Milgrom’s model—such as \cite{van2010pay}, \cite{anderson2018agency}, and \cite{xu2016golden}—rationalize signing or front-loaded bonuses as mechanisms to satisfy participation constraints and attract high-ability agents. Later developments following Sannikov \cite{sannikov2008continuous} reinterpret the terminal lump-sum component as a form of deferred or termination bonus, encompassing the golden-parachute structure in \cite{possamai2024there}, the deferred-bonus account in \cite{edmans2012dynamic}, and the analysis of short-term and deferred banker compensation in \\
\cite{jokivuolle2014bankers} and \cite{jokivuolle2015bonus}. Related works—including \cite{hoffmann2021only}, \cite{shan2019incentives}, and \cite{maestri2014efficiency}—show that variations in the timing and structure of bonuses naturally emerge as optimal responses to intertemporal incentive and participation constraints. Nevertheless, most of these studies present bonus payment schemes that are either static (without explicit time dynamics) or exogenously characterized through equilibrium or boundary conditions, without providing explicit analytical formulas for the evolution of bonuses over time.

This paper develops a unifying framework for analyzing bonus payments in continuous time within a finite-horizon contracting environment, endogenizing both the level and timing of bonus schemes. Our model provides a general explanation of how bonus structures evolve in practice as functions of the model’s inputs, offering a more interpretable solution compared to previous formulations. In particular, it captures the emergence of the “golden hello”—a signing bonus optimally offered by the principal—and highlights the role of the relative impatience of the agent and the principal in shaping both the timing and size of payments.

Our paper makes the following contributions. We study a sequential contracting problem in which the contract consists of multiple lump-sum payments scheduled at flexible times chosen by the principal, in a setting that builds on \cite{sannikov2008continuous} and \cite{possamai2024there}. To begin with, we study the first-best benchmark, where the agent directly follows the principal's instructions. We rigorously establish, following the approach of \cite{possamai2024there}, that in a finite contracting horizon, degeneracy never occurs; that is, the principal's value in the first-best problem cannot reach its maximal bound. Next, we study the second-best problem, where the principal cannot directly control the agent's actions, we extend the methods of \cite{cvitanic2018dynamic}  and \cite{alvarez2025sequential}. Specifically, we represent the agent's value through a recursive system of backward stochastic differential equations (BSDEs), thus reducing the principal's bilevel optimization problem to a single stochastic control problem that combines mixed control and stopping with a finite number of impulses. Instead of directly tackling this single-level problem, we apply an iterated optimal stopping technique (see \cite{oksendal2007applied}, \cite{baccarin2013portfolio}, \cite{asri2019finite}) to recursively characterize the principal’s value function when the principal can issue up to 
$N$ bonuses at her discretion. Using the weak dynamic programming principle (WDPP) (\cite{bouchard2011weak} and \cite{bouchard2012weak}) for the mixed control-and-stopping problem arising in this recursive framework, we provide a rigorous analysis showing that the principal’s value function admits a viscosity characterization in terms of a recursive system of HJB variational inequalities. Finally, using the comparison principle, we prove the well-posedness of the principal’s value function.

Finally, we apply a penalized approximation together with a policy iteration scheme introduced in \cite{reisinger2018penalty} to numerically approximate the recursive system of HJB variational inequalities. This allows us to investigate how the contracting environment influences the principal’s value and to identify the regions where bonus payments occur (intervention regions). We then separate the discussion into the first-best and second-best cases. In the first-best case, by comparing the principal's value between both settings, we confirm that the principal's value in the first-best case always dominates that in the second-best case and further show that the principal’s value is ultimately a decreasing function in the second-best setting.
Turning to the structure of the contract, we find that the principal benefits from increasing the number of payments. Although degeneracy does not occur for any fixed number of bonus payments, it can emerge as the number of payments becomes arbitrarily large. In addition, the principal benefits from the patience of the agent.

Next, we turn to the second-best setting. Depending on the contracting environment, three distinct regimes arise: (i) the principal is significantly more impatient than the agent; (ii) the principal is more impatient, but not excessively so; and (iii) the agent is more impatient than the principal. These regimes are characterized by the parameters $\delta$, the ratio of the agent’s to the principal’s discount rates, and $\gamma$, the agent’s coefficient of risk aversion.

First, we compare the principal’s value under a two-period deterministic payment schedule (sequential case) and a single-bonus payment scheme. The principal’s value under the bonus scheme consistently dominates that of the deterministic case, supporting both our numerical approach and the underlying economic intuition. From the two-period deterministic schedule, we observe that when the principal is highly impatient, delaying the first payment is optimal, whereas when the agent is more impatient, paying the bonus as early as possible becomes optimal. Interestingly, when the principal is impatient but not excessively so, the outcome depends on the reservation utility: for low reservation utilities, delayed payments are preferable, while for high reservation utilities, early deterministic payments yield a higher value.

Secondly, we examine the principal’s problem in the second-best case under the single-bonus scenario. 
\begin{enumerate}[label=\roman*)]
    \item When the principal is substantially more impatient than the agent, she is unwilling to offer a sign-on bonus; in this regime, a “golden hello” never arises. Bonus payments occur only when the agent’s continuation utility lies within an intermediate range, jointly driven by the informational rent of the agent and the principal’s high discount rate.
    \item When the principal is more impatient than the agent, but not excessively so, she is willing to offer a sign-on bonus to agents with relatively high reservation utilities. As the contract approaches maturity, it becomes increasingly important for the principal to implement such bonus payments. Moreover, for a fixed maturity date, the bonus payment scheme is monotone with respect to the agent’s utility.
    \item When the agent is more impatient than the principal, the principal is willing to offer a sign-on bonus to agents with relatively low reservation utilities. Nonetheless, there always exists a small region in which a “golden hello” is not guaranteed: due to limited liability, the principal withholds the bonus when the agent’s initial utility is extremely low, to prevent the agent from leaving the contract. Furthermore, because of the agent’s high discount rate, the principal must provide a larger initial bonus to achieve the same incentive effect.
\end{enumerate}
Finally, we examine the principal’s problem in the second-best setting under the multiple-bonus scenario. Two general conclusions emerge. First, as the number of bonus payments increases, the principal prefers to allocate a larger share to the later bonuses, reflecting the greater flexibility afforded by sequential contracting. Second, the “golden hello” gradually disappears as the number of bonus payments increases. 

The rest of the paper is organized as follows. Section \ref{setup} introduces the model and formulates the objectives of both the agent and the principal. Section \ref{firstbest} is devoted to the first-best setting, where we establish upper and lower bounds for the principal’s value in a general framework. In Section \ref{secondbest}, we turn to the second-best problem, deriving the HJB variational inequalities satisfied by the principal’s value function. Finally, Section \ref{numerics} presents the numerical analysis based on these variational inequalities, which validate the economic intuition and illustrate how the ratio of the agent’s to the principal’s discount rates influences the principal’s value and the structure of the bonus payment schemes.



\section{The principal-agent model}\label{setup}

This section outlines the basic setting of the contract and the agent's effort, partially building on the continuous-time contracting models introduced in \cite{sannikov2008continuous} and \cite{possamai2024there}. Let $\left(\Omega, \mathcal{F}, \mathbb{P}^0\right)$ be a probability space supporting a one-dimensional $\mathbb{P}^0$-Brownian motion $W^0$. For fixed parameters $\sigma>0$, and $X_0 \in \mathbb{R}$, the output process is given by
$$
X_t=X_0+\sigma W_t^0 \; , \; t \geq 0.
$$
We denote by $\mathbb{F}$ the $\mathbb{P}^0$-augmentation of the natural filtration generated by $X$ and satisfying the usual conditions. To model the agent's effort, consider a family of probability measures $\left\{\mathbb{P}^\alpha\right\}_{\alpha \in \mathcal{A}}$ under which the output process evolves as
$$
\mathrm{d} X_t=\alpha_t \mathrm{~d} t+\sigma \mathrm{d} W_t^\alpha,
$$where $W^\alpha$ is a $\mathbb{P}^\alpha$-Brownian motion. This change of measure is justified by Girsanov's theorem. Let $\mathcal{A}$ be the collection of all $\mathbb{F}$-predictable processes $\alpha$ taking values in a compact subset $A:=[0,\bar{a}]$, where $\bar{a}>0$ denotes the agent's maximal effort. For each $\alpha \in \mathcal{A}$, we define an equivalent measure $\mathbb{P}^\alpha$ such that
$$
W_t^\alpha:=W_t^0-\int_0^t \frac{\alpha_s}{\sigma} \mathrm{~d} s.
$$
Any $\alpha \in \mathcal{A}$ is referred to as an effort process and represents an action taken to alter the distribution of the output process from $\mathbb{P}^0$ to $\mathbb{P}^\alpha$. 

\medskip
Let $N$ be the fixed maximum number of bonus payments the principal can assign to the agent. We define a contract by a triple $\mathbf{C}:=\left(\bar{\tau}_N, \bar{\xi}_N, \xi_T\right)$, where $\bar{\xi}_N:=\left(\xi_1, \xi_2, \ldots, \xi_N\right)$ denotes a sequence of $N$ bonus payments that occur at increasing times determined by the principal, and $\xi_T$ denotes a final payment made at the end of the contract period. The timing vector $\bar{\tau}_N=\left(\tau_1, \ldots, \tau_N\right)$ belongs to $\mathcal{T}_{[0, T)}$, the set of all $\mathbb{F}$-stopping times taking values in $[0,T)$. Each $\xi_i $, for $i\in \{1,\ldots,N\}$, is a nonnegative, $\mathcal{F}_{\tau_i}$-measurable random variable, while $\xi_T$ is a nonnegative, $\mathcal{F}_T$-measurable random variable. The principal's objective is to design an optimal contract that induces the agent to act in accordance with her interests, while the agent simultaneously maximizes his own expected utility under the contract's terms. 

\medskip
Next, we introduce the agent’s preferences and behavioral parameters:
\begin{itemize}
    \item The agent's utility function: $u(\xi):=\xi^{1 / \gamma}$, with the inverse function $u^{-1}(y)=y^\gamma$ for some $\gamma>1 $.
    \item The agent's cost function: $h:[0, \infty) \rightarrow[0, \infty)$ is increasing, strictly convex, continuously differentiable, and  $ h(0)=0$.
    \item The agent's discount rate $r>0$ is a fixed constant.
\end{itemize}

\medskip
Given a contract $\mathbf{C}:=(\bar{\tau}_N, \bar{\xi}_N,\xi_T) $ and $\alpha \in \mathcal{A}$, the utility obtained by the agent is defined by
\begin{equation}\label{Jagent}
  J^{\mathrm{A}}(\mathbf{C}, \alpha):=\mathbb{E}^{\mathbb{P}^\alpha}\left[ \sum_{i = 1}^{N } \mathrm{e}^{-r \tau_{i} } u(\xi_{i})  + e^{-rT}u(\xi_{T}) -\int_0^{T} r \mathrm{e}^{-r t} h\left(\alpha_t\right)\mathrm{d} t\right],  
\end{equation}
where the agent receives a sequence of $N$ bonus payments at stopping times $\bar{\tau}_N$, along with a terminal payment $\xi_T$ at the end of the contracting horizon $T$. The discount rate $r>0$ reflects the agent's degree of impatience. At the inception of the contract, the agent exerts an effort $\alpha\in \mathcal{A}$, which is bounded and takes values in $A=[0, \bar{a}]$.

\medskip
We now introduce the principal's objective 
\begin{equation}\label{Jprincipal}
    J^{\mathrm{P}}(\mathbf{C}, \alpha)=\mathbb{E}^{\mathbb{P}^\alpha}\left[ -\sum_{i=1}^{N} \mathrm{e}^{-\rho \tau_i} (1+k)\xi_i -  e^{ - \rho T}\xi_{T} + \int_0^{T} \rho \mathrm{e}^{-\rho t}\left(\alpha_t \right) \mathrm{d} t\right].
\end{equation}
The functional $J^{\mathrm{P}}$ is well defined on $\{-\infty\} \cup \mathbb{R}$, due to the boundedness of $A$ and $\bar{\tau}_N$, together with the nonnegativity of $\bar{\xi}_N$ and $\xi_T$. Here, $k>0$ denotes a fixed exogenous transaction cost rate, and $\rho>0$ is the principal's discount rate.

\medskip
To ensure that both the agent's and the principal's utility criteria are finite, we define the subset of admissible contracts $\mathbb{C}$, consisting of all $\mathbf{C} = (\bar{\tau}_N, \bar{\xi}_N,\xi_T)$ such that:
\begin{equation}\label{contracts::integrability}
    \max_{i\in [N]}\sup _{\alpha \in \mathcal{A}} \left[\mathbb{E}^{\mathbb{P}^\alpha}\left[ |\xi_i|^\gamma\right] +  \mathbb{E}^{\mathbb{P}^\alpha}\left[ |\xi_T|^\gamma\right]\right]<\infty.
\end{equation}

\section{The first-best problem}\label{firstbest}
In this section, we introduce the first-best scenario, where the principal has full bargaining power and therefore determines both the contract terms and the agent’s effort, subject to the agent's participation restriction. First, we state the principal's problem in the first-best scenario, write it in a tractable form, and specify the corresponding optimization problem. Next, we state the main result, which characterizes the principal's value. 

In this case, given a reservation utility level $u(R)$ (with $R \geq 0$) chosen by the agent, the principal's problem in the First-Best setting is
\begin{equation}\label{principalpfb}
    V^{\mathrm{FB},\mathrm{P}}:= v^{FB,P}(u(R)) =\sup \left\{J^{\mathrm{P}}(\mathbf{C}, \alpha): \mathbf{C} \in \mathbb{C}, \alpha \in \mathcal{A}\right., \left.J^{\mathrm{A}}(\mathbf{C}, \alpha) \geq u(R)\right\},
\end{equation}
which is the optimal value that the principal could obtain in the first-best case.

\medskip
Equivalently, the optimal value of the principal under the first-best case is
\begin{equation*}\label{principalpfbs}
\begin{split}
  V^{\mathrm{FB},\mathrm{P}}&= \inf_{\lambda \leq 0} \sup_{\mathcal{C} \in \mathbb{C},\alpha \in \mathcal{A}} \{\lambda u(R) + J^{\mathrm{P}}(\mathbf{C}, \alpha) - \lambda J^{\mathrm{A}}(\mathbf{C}, \alpha) \}  \\ 
  &=\inf_{\lambda \leq 0} \left\{\lambda u(R)  + \sup_{(\alpha,\bar{\tau}_N, \bar{\xi}_N, \xi_T)} \mathbb{E}^{\mathbb{P}^{\alpha}} \Bigg[  - \sum_{i=1}^{N} 
    \mathrm{e}^{-\rho \tau_i}  \left(  u(\xi_i)\lambda \mathrm{e}^{(\rho-r)\tau_i} + (1 + k)\xi_i  \right) \right. \\ 
    &\quad\quad\quad\quad\quad\quad\quad- \mathrm{e}^{-\rho T} \left(u(\xi_T)\lambda \mathrm{e}^{(\rho-r)T} +  \xi_T  \right)   \left.    +   \int_{0}^{T} \rho \mathrm{e}^{-\rho t}   \left(\alpha_t + h(\alpha_t) \delta \lambda \mathrm{e}^{(\rho-r)t} \right) \mathrm{d}t \Bigg] \right\},
\end{split}
\end{equation*}
where $\lambda \leq 0$ denotes the Lagrange multiplier associated with the participation constraint, and $\delta:=\frac{r}{\rho}$ represents the relative discounting between the principal and the agent. The second equality follows from the standard Karush-Kuhn-Tucker method.

To tackle this problem, we employ pointwise maximization within the expectation, which allows us to derive the upper bound and lower bound of the principal's optimal value.

Having outlined the solution approach via the pointwise maximization technique, we are now prepared to state the main result for the first-best benchmark. The proof, which builds upon the preceding arguments together with several technical lemmas, is deferred to Appendix \ref{appendix::FB}.

\medskip
First, in the finite-horizon setting, the inherent time-dependent (term-structure) nature of the problem prevents the derivation of an explicit optimal contract. Instead, we establish upper and lower bounds for the principal's value under the first-best benchmark, which depend on the relative discounting parameter $\delta$.

Second, we show that, for any value of $\delta$, the first-best problem remains non-degenerate over a finite horizon. In particular, it is impossible to construct a sequence of contracts that induces the agent to exert maximal effort continuously throughout the entire period.

In the following theorem, we shall express upper bounds and lower bounds in terms of the convex dual of $-h^{-1}(\cdot)$ and the concave conjugate of $F(y):= -u^{-1}(y)$
\begin{equation}\label{defG*F*}
    G^{*}(p) := \sup_{a \in A} \left(a + h(a) p \right) ,\, p \leq 0   ,\quad F^{*,k}(p) := \inf_{y \geq 0 }(yp - (1 + k) F(y)) ,\, p \leq 0.
\end{equation}

\begin{theorem}(First-Best Problem)\label{thm::FB}
\begin{enumerate} 
    \item If $\delta\gamma > 1$ , then  $\underline{v}^{FB,P}(u(R)) \leq V^{\mathrm{FB},\mathrm{P}} \leq \bar{v}^{FB,P}(u(R))$, where
    \begin{footnotesize}
        \begin{flalign*}
      &\underline{v}^{FB,P}(u(R)) 
      := \inf_{\lambda \leq 0} \left\{\lambda u(R)  + \sum_{i=1}^{N}
   \left\{- \mathrm{e}^{-\rho \tau_i} \left(  F^{*,k}(\lambda \mathrm{e}^{(\rho - r)\tau_i}) \right) \right\} -  \mathrm{e}^{-\rho T} F^{*,0}(\lambda \mathrm{e}^{(\rho - r)T}) \right.\\  &\quad\quad\quad\quad\quad\quad\quad\quad\quad\quad\quad\left.+ \int_{0}^{T}\rho \mathrm{e}^{-\rho t}G^{*}(\delta  \lambda \mathrm{e}^{(\rho-r)t} ) dt \right\}, \\
    &\overline{v}^{FB,P}(u(R)) 
            := \inf_{\lambda \leq 0} \left\{\lambda u(R) - N 
    F^{*,k}(\lambda) -  \mathrm{e}^{-\rho T} F^{*,0}(\lambda \mathrm{e}^{(\rho - r)T}) + \int_{0}^{T}\rho \mathrm{e}^{-\rho t}G^{*}(\delta  \lambda \mathrm{e}^{(\rho-r)t} ) dt \right\},   
   \end{flalign*}
   \end{footnotesize}
   and $\tau_{i} = (i-1)\zeta \; , \; 0 < \zeta < \frac{T}{N}$.
    \item If $\delta\gamma < 1$, then  $ \underline{v}^{FB,P}(u(R)) \leq V^{\mathrm{FB},\mathrm{P}} \leq \bar{v}^{FB,P}(u(R))$ where
     \begin{footnotesize}
        \begin{flalign*}
      &\underline{v}^{FB,P}(u(R)) 
      := \inf_{\lambda \leq 0} \left\{\lambda u(R)  + \sum_{i=1}^{N}
   \left\{- \mathrm{e}^{-\rho \tau_i} \left(  F^{*,k}(\lambda \mathrm{e}^{(\rho - r)\tau_i}) \right) \right\} -  \mathrm{e}^{-\rho T} F^{*,0}(\lambda \mathrm{e}^{(\rho - r)T}) \right.\\
   &\quad\quad\quad\quad\quad\quad\quad\quad\quad\quad\quad\left.+ \int_{0}^{T}\rho \mathrm{e}^{-\rho t}G^{*}(\delta  \lambda \mathrm{e}^{(\rho-r)t} ) dt \right\}, \\
            &\overline{v}^{FB,P}(u(R)) 
            :=  \inf_{\lambda \leq 0} \left\{\lambda u(R) -  N\mathrm{e}^{-\rho T} F^{*,k}(\lambda \mathrm{e}^{(\rho - r)T}) - \mathrm{e}^{-\rho T} F^{*,0}(\lambda \mathrm{e}^{(\rho - r)T}) + \int_{0}^{T}\rho \mathrm{e}^{-\rho t}G^{*}(\delta  \lambda \mathrm{e}^{(\rho-r)t} ) dt \right\},
\end{flalign*}
\end{footnotesize}
and $\tau_{i} = T - i\zeta \; , \; 0 < \zeta < \frac{T}{N}$.
    \item If $\delta\gamma = 1$, then 
    \begin{footnotesize}    
    $$V^{\mathrm{FB},\mathrm{P}} =  \inf_{\lambda \leq 0} \left\{\lambda u(R) -  (N+1)\mathrm{e}^{-\rho T} F^{*}(\lambda \mathrm{e}^{(\rho - r)T}) + \int_{0}^{T}\rho \mathrm{e}^{-\rho t}G^{*}(\delta  \lambda \mathrm{e}^{(\rho-r)t} ) dt \right\} . $$
    \end{footnotesize}

    \item If $R = 0$, then  $0 < V^{\mathrm{FB},\mathrm{P}} < \bar{a}(1 - e^{-\rho T})$.
\end{enumerate}
\end{theorem}
\begin{remark}\hfill 
\begin{enumerate}[label=\roman*]
    \item The existence of an explicit representation for the principal’s value under the first-best case, as a function of the reservation utility, crucially depends on the composite parameter $\delta\gamma$. We note that in the benchmark case $\delta\gamma = 1$, a closed-form expression can be derived. This configuration implies that the principal’s value under the first-best case is invariant to the timing of bonus payments when the agent’s relative impatience ratio $\delta$ coincides with the reciprocal of their risk-aversion coefficient $1/\gamma$.
    \vspace{-0.2cm}
    \item When $\delta\gamma \neq 1$, an explicit solution is no longer available. In this regime, the principal’s value under the first-best case can only be characterized through analytical upper and lower bounds, reflecting its nonlinear dependence on the reservation utility.
    \vspace{-0.2cm}
    \item The sign of $1 - \delta\gamma$ determines the optimal timing of payments. When $\delta\gamma < 1$, the agent is relatively more patient than the principal, and it is optimal for the principal to defer bonus payments to exploit discounting advantages. In contrast, when $\delta\gamma > 1$, the agent is more impatient, and the principal benefits from advancing payments to mitigate the agent’s stronger time preference.
\end{enumerate}
\end{remark}

\section{The second-best problem}\label{secondbest}
In this section, we introduce the second-best problem, where the agent’s effort is unobservable and thus, not contractible. In this setting, the principal strategically designs the contract to maximize her expected utility, while anticipating the agent’s optimal effort response. The contract must satisfy the agent’s participation constraint, ensuring that his expected utility is at least as high as his reservation utility. 
\subsection{The principal's problem}\label{sb::principalp}
In this case, we define the principal's value under the second-best case when the maximum number of bonus payments is equal to $N$ as 
\begin{equation}\label{principalpsb}
    V^{\mathrm{P},\mathrm{N}}:=\sup_{\mathbf{C} \in \mathfrak{C}_R} \sup _{\alpha \in \mathcal{A}^*(\mathbf{C})} J^{\mathrm{P}}(\mathbf{C}, \alpha),
\end{equation}
where $\mathfrak{C}_R$ denotes all contracts that satisfy the agent's participation constraint
$$
\mathfrak{C}_R:=\left\{\mathbf{C} \in \mathbb{C}: V^{\mathrm{A}}(\mathbf{C}) \geq u(R)\right\}, \text { with } V^{\mathrm{A}}(\mathbf{C}):=\sup _{\alpha \in \mathcal{A}} J^{\mathrm{A}}(\mathbf{C}, \alpha).
$$
In addition, we denote by $\mathcal{A}^{\star}(\mathbf{C})$ the collection of all the agent's optimal effort responses to a given contract $\mathbf{C}$ 
$$
\mathcal{A}^{\star}(\mathbf{C}):=\left\{\alpha \in \mathcal{A}: V^{\mathrm{A}}(\mathbf{C})=J^{\mathrm{A}}(\mathbf{C}, \alpha)\right\} .
$$
To address the second-best contracting problem, we follow the approach outlined in \cite{alvarez2025sequential}, allowing us to write the principal’s problem (\ref{principalpsb}) as a standard stochastic control problem. The terminal utility $\eta_T := u(\xi_T)$  promised by the principal takes the form
\begin{equation}
    \eta_T:=Y_T^{0,Y_0, Z,\bar{\tau}_N,\bar{\eta}_N} = Y_0 - \int_0^{T} H(Y_t^{0,Y_0, Z,\bar{\tau}_N,\bar{\eta}_N},Z_t) \mathrm{d} t+\int_0^Tr Z_t \sigma \mathrm{~d} X_t -  \sum_{m=1}^{N} \eta_m \mathbbm{1}_{\left[\tau_m, T\right]}, 
\end{equation}
where $\bar{\eta}_N$ denotes the utility form of the sequential bonus payments 
$$\bar{\eta}_N := \left(\eta_1,\cdots,\eta_N) = (u(\xi_1),\cdots,u(\xi_N)\right),$$
 and
\begin{align}
h(a,y,z) &:=  r z a -r y- rh(a), \quad (a,y,z) \in A\times [0,\infty)\times \RR, \nonumber \\
H(y,z) &:= \sup_{a\in A}h(a,y,z). \label{hamiltonian.1}
\end{align}
We know that $Y^{Y_0,Z,\bar{\tau}_N,\bar{\eta}_N}$, which represents the agent's utility process and $Z$, which denotes the compensation process, satisfy the following integrability constraints
\begin{equation}\label{condition::interability::YZ}
    \sup _{\alpha \in \mathcal{A}} \mathbb{E}^{\mathbb{P}^\alpha}\left[\sup _{0 \leq t \leq T}\left(\left|Y_t^{Y_0, Z,\bar{\tau}_N, \bar{\eta}_N}\right|\right)^p\right]<\infty, \text { and } \sup _{\alpha \in \mathcal{A}} \mathbb{E}^{\mathbb{P}^\alpha}\left[\left(\int_0^{T}\left(\left|Z_t\right|\right)^2 \mathrm{~d} t\right)^{\frac{p}{2}}\right]<\infty,
\end{equation}
for some $p > 2 \vee \gamma$. \\
Applying a similar technique to \cite[Theorem 1]{alvarez2025sequential}, we obtain that the second-best principal's value solves the following mixed control-and-stopping problem
\begin{equation}\label{principalvreduction.1}
\begin{split}
    V^{\mathrm{P},\mathrm{N}}  &= \sup_{Y_0 \geq u(R)} V^{p,N}(Y_0),\\ 
    V^{p,N}(Y_0) &= \sup_{(Z,\bar{\tau}_N,\bar{\eta}_N) \in \mathcal{U}^{N}(Y_0), \hat{a} \in \hat{\mathcal{A}}(Z)} J^{p,N}(0,Y_0;(Z,\bar{\tau}_N,\bar{\eta}_N),\hat{a} ),\\
   J^{p,N}(0,Y_0;(Z,\bar{\tau}_N,\bar{\eta}_N) , \hat{a} ) &=   \mathbb{E}^{\mathbb{P}^{\hat{a}}}  \left[  -\sum_{i=1}^{N} \mathrm{e}^{-\rho \tau_i} (1+k)u^{-1}(\eta_i)   - \mathrm{e}^{- \rho T}u^{-1}(Y^{0,Y_0, Z,\bar{\tau}_N,\bar{\eta}_N}_{T})  \right.\\
   &\left.\quad\quad\quad\quad\quad+ \int_0^{T} \rho \mathrm{e}^{-\rho s }\left(\hat{a}_s\right) \mathrm{d} s  \right],
\end{split}
\end{equation}
where $\mathcal{U}^{N}\left(Y_0\right)$ denotes the collection of all admissible triples $(Z,\bar{\tau}_N,\bar{\eta}_N)$ such that the process $Y^{0,Y_0,Z,\bar{\tau}_N,\bar{\eta}_N}$ satisfies the integrability conditions \eqref{contracts::integrability} and \eqref{condition::interability::YZ}, with the limited-liability constraint $Y^{0,Y_0,Z,\bar{\tau}_N,\bar{\eta}_N} \geq 0 $ $dt\otimes\PP^{\hat{a}}-\mathrm{a.e.}$, which is further associated with the following state dynamics
\begin{equation*}
    Y_t= Y_0 + \int_0^{t} r ( Y_s^{0,Y_0, Z,\bar{\tau}_N,\bar{\eta}_N} +  h(\hat{a}_s) ) \mathrm{d} s+ \int_0^{t} r Z_s \sigma \mathrm{~d} W^{\hat{a}}_s -  \sum_{m=1}^{N} \eta_m \mathbbm{1}_{\left[\tau_m, T\right]}(t) .
\end{equation*}
The set $\hat{\mathcal{A}}(z)$ represents the collection of admissible agent's optimal responses
$$\hat{\mathcal{A}}(z) :=\partial h^{\star}(z)=\left\{a \in A: h^{\star}(z)=z a-h(a)\right\}.$$

\begin{remark}
    The limited-liability condition characterizes the situation in which the agent’s continuation utility reaches zero, i.e., when $Y_t = 0$, the agent ceases to provide effort. This boundary is absorbing, as described in \cite{possamai2024there}, and similar boundary conditions also appear in \cite{federico2008pension}. From a control-theoretic perspective, this formulation can be interpreted as a stochastic control problem with discretionary exit \cite{fleming2006controlled}, or equivalently, a stochastic exit-time control problem.
\end{remark}
Invoking \cite[Theorem $4.3$]{ElKarouiTan2024}, we justify that the previous weak formulation stochastic control problem is equivalent to a stochastic control problem in the strong formulation where the stochastic basis is fixed.
Combined with the previous observation, we reduce the principal's problem (\ref{principalvreduction.1}) to the following stochastic control problem
\begin{equation}\label{principalvreduction.2}
\begin{split}
 V^{\mathrm{P},\mathrm{N}}  &= \sup_{Y_0 \geq u(R)} V^{p,N}(Y_0),\\ 
    V^{p,N}(t,y) &= \sup_{(Z,\bar{\tau}_N,\bar{\eta}_N) \in \mathcal{U}^{N}(t,y), \hat{a} \in \hat{\mathcal{A}}(Z)} J^{p,N}(t,y;(Z,\bar{\tau}_N,\bar{\eta}_N),\hat{a} ),\\
   J^{p,N}(t,y;(Z,\bar{\tau}_N,\bar{\eta}_N), \hat{a} ) &=   \mathbb{E} \left[  -\sum_{i=1}^{N} \mathrm{e}^{-\rho (\tau_i-t)} (1+k)u^{-1}(\eta_i)   - \mathrm{e}^{- \rho (T-t)}u^{-1}(Y^{t,y, Z,\bar{\tau}_N,\bar{\eta}_N}_{T})  \right.\\
   &\left.\quad\quad\quad\quad\quad+ \int_t^{T} \rho \mathrm{e}^{-\rho (s-t) }\left(\hat{a}_s \right) \mathrm{d} s  \right], 
\end{split}
\end{equation}
for all $(t,y) \in [0,T] \times [0,\infty)$, and for $\hat{t}> t$, we have
\begin{equation}\label{dynamicsYB}
    \begin{split}
        Y_t^{t,y, Z,\bar{\tau}_N,\bar{\eta}_N}&= y, \\ 
        \mathrm{~d} Y_s^{t,y, Z,\bar{\tau}_N,\bar{\eta}_N}  &= r\left(Y_s^{t,y, Z,\bar{\tau}_N,\bar{\eta}_N}+h\left(\hat{a}_s\right)\right) \mathrm{d} s + r Z_s \sigma \mathrm{~d} B_s \;, \; s  \in [t, T], \\
        Y^{t,y, Z,\bar{\tau}_N,\bar{\eta}_N}_{\tau_i}  &= Y^{t,y, Z,\bar{\tau}_N,\bar{\eta}_N}_{\tau_i-} - \eta_i, \quad i \in \{1,\cdots,N\}.
    \end{split}
\end{equation}
Similar to \cite{belak2017general}, to make (\ref{dynamicsYB}) well defined at $\tau_1 = t$, we set $Y_{t-}^{t,y, Z,\bar{\tau}_N,\bar{\eta}_N} := y$. $B$ is a Brownian-motion defined on fixed probability space $(\Omega,\FF,\PP)$, and $\FF$ is the augmented filtration generated by $B$. The control $Z$ is an $\FF$-predictable process,$\bar{\tau}_{N} = (\tau_1,\cdots,\tau_N)$ is a sequence of preditable stopping times $\tau_{i+1} \geq \tau_i$ and $\tau_i \in \mathcal{T}_{[t,T)}$, for $i \in \{1,\cdots,N\}$, and $\bar{\eta}_N:= (\eta_1,\cdots,\eta_N) $,  for every $i \in \{1,\ldots,N\}$, $\eta_i$ is a $\mathcal{F}_{\tau_i}$-measurable random variable satisfying the integrability condition
\begin{equation}\label{contracts::integrability::W}
     \mathbb{E}\left[ |\eta_i|  \right]<\infty, \quad i \in \{1,\ldots,N\},
\end{equation}
where $\mathbb{E}\left[\cdot\right]$ denotes the expectation taken under $\PP$. In addition, $\mathcal{U}^{N}(t,y)$ is the collection of triplets $(Z,\bar{\tau}_N,\bar{\eta}_N)$ so that the process $Y^{t,y, Z,\bar{\tau}_N,\bar{\eta}_N}$ satisfies 
\begin{equation}\label{condition::interability::YZ::reduction2}
     \mathbb{E}\left[\sup_{ t \leq s\leq T}\left|Y_s^{t,y, Z,\bar{\tau}_N, \bar{\eta}_N}\right|^p\right]<\infty, \text { and }  \mathbb{E}\left[\left(\int_0^{T}\left|Z_t\right|^2 \mathrm{~d} t\right)^{\frac{p}{2}}\right]<\infty, 
\end{equation}
 for some $p > 2 \vee \gamma$, and the limited liability constraint:  $Y_s^{t,y, Z,\bar{\tau}_N,\bar{\eta}_N} \geq 0,\mathbb{W}-\mathrm{a.s}$, for all $s\geq t$.

\subsection{Main result for the second-best problem}\label{sb::main}
In this section, we characterize the principal's value function in the second-best setting. While the existing literature, such as \cite{sannikov2008continuous,possamai2024there,possamai2025golden}, considers infinite-horizon problems, our focus is on short-term contracting environments with a fixed termination date and finitely many discretionary payments. The finite time horizon motivates the use of parabolic variational inequalities rather than ordinary differential equations to describe the principal’s value function. The finite number of bonuses leads naturally to an iterative optimal stopping approach, following methodologies developed in \cite{baccarin2013portfolio,asri2019finite,oksendal2007applied}. Through this iterative framework, we establish the well-posedness of the principal’s value function in the viscosity sense. Moreover, by exploiting the partial concavity of the principal’s value with respect to the Agent’s value, we characterize the optimal bonus payment scheme via static optimization techniques.

The following assumption underlies our main results on the second-best contracting problem, which in turn provides a key step toward establishing the well-posedness of viscosity solutions to the associated HJB variational inequalities.

\begin{assumption}\label{assm::enoughgammah(a)}\hfill
\vspace{-0.2cm}
\begin{enumerate}
    \item Agent's risk aversion parameter $\gamma > 2$.
    \vspace{-0.2cm}
    \item Cost effort function $h(\cdot) \in \mathcal{C}^{3}$: $h'(\cdot)$ is a convex function, $h'(0) = \beta > 0$, and $h''(0) = \alpha > 0$.  
\end{enumerate}
\end{assumption}

To state our main result, we define the differential operator $\mathcal{G}$,
 the nonlocal operator $\mathcal{M}$, and the function $g^n$ as follows
    \begin{align}
            \mathcal{G} v &:= \rho v - \sup_{z \geq 0 ,\hat{a} \in \hat{\mathcal{A}}(z)}\{rh(\hat{a})v_y + \rho \hat{a} + \frac{1}{2}v_{yy}r^{2}\sigma^{2}z^{2}  \} - rv_y y,  \label{operatorG}
\\
\label{operatorM}
     \mathcal{M} V(t, y)&:=\sup _{0 \leq \eta \leq y}[V(t, y - \eta)-(1 + k)\eta^{\gamma}],  \\\label{terminalg}
      g^{n}(y) &:= \mathcal{M}^{n}F(y) =   - \left(\frac{(1+k)^{1 /(\gamma-1)}}{n+(1+k)^{1 /(\gamma-1)}}\right)^{\gamma-1} y^{\gamma}\; , \;  n \in \{0,\ldots,N\} .
    \end{align}

\begin{theorem}\label{hjbvi:pp:in}
Let Assumption \ref{assm::enoughgammah(a)} hold. The principal's value function when the principal has $N \geq 1$ bonus payments satisfies the following properties
\begin{enumerate}
    \item $V^{p,N}=v^{p,N}$, where $(v^{p,n})_{n=0}^N$ is the unique continuous viscosity solution to the following system of constrained HJB variational inequalities on the domain $ [0,T) \times [0,\infty)$ 
    \begin{equation}\label{vnpdevi}
    \begin{split}
        &\min\{ -v^{p,n}_t + \mathcal{G}v^{p,n}  , v^{p,n} - \mathcal{M} v^{p,n-1}, - v^{p,n}_{yy}\} = 0, \; n \in \{1,\ldots,N\},\\
        &\min\{-v^{p,0}_t+\mathcal{G}v^{p,0},-v_{yy}^{p,0}\}=0,
            \end{split}
    \end{equation}
    and terminal-boundary condition
\begin{equation}\label{terminalvnp}
    \begin{split}
        &v^{p,n}(T-,y) = g^n(y), \quad (y,n)  \in [0,\infty)\times\{0,\ldots,N\},  \\
        &v^{p,n}(t,0) = 0, \quad (t,n) \in [0,T]\times\{0,\ldots,N\}  .
              \end{split}
    \end{equation}
    \item \label{partialconcave}  $V^{p,n}(t,\cdot)$ is concave, for any given $(t,n) \in [0,T)\times \{0,\ldots,N\}$.
    \item \label{nofiring}  The $n$-th optimal bonus payment satisfies the following feedback characterization 
    \begin{equation}\label{etastari}
    \begin{split}
       \eta^{n,*}(t,y) &= \argmax_{ 0 \leq \eta \leq y }  -  (1 + k)\eta^{\gamma} +V^{p,n-1}(t, y - \eta) \\
        &= \inf \{\eta \in [0,y] : \partial^{+}_{y}V^{p,n-1}(t,y - \eta) > -\gamma (1 + k)\eta^{\gamma - 1} \},
        \end{split}
        \end{equation}
           in the intervention region $ \{(t,y) \in [0,T)\times(0,\infty), V^{p,n}(t,y) = \mathcal{M}V^{p,n-1}(t,y)  \}$, for all , $n \in \{1,\ldots,N\} $. We denote $\partial_y^{-} V$ and $\partial_y^{+} V$ the left and right derivative of $V$ with respect to $y$ respectively, which exist everywhere by the concavity of $V(t,\cdot)$, for all $0 \leq t\leq T$. Moreover,  $$\eta^{n,*}(t,y) < y,  \text{  } \forall n \in \{ 
 1 , \cdots ,  N  \}.$$
\end{enumerate}
\end{theorem}
\newpage
\begin{remark}\hfill
    \begin{enumerate}[label=\roman*)]
    \item In Assumption \ref{assm::enoughgammah(a)}, the condition $\gamma>2$ ensures that the agent’s degree of risk aversion is sufficiently high. Since the effort set is $A:=[0,\bar{a}]$, the assumption \ref{assm::enoughgammah(a)} implies that $h$ is strictly convex on this set which tantamounts to requiring the curvature of the marginal cost of effort to be large enough. From a mathematical standpoint, these assumptions are instrumental in constructing smooth super- and subsolutions, which in turn are used to establish the well-posedness of the corresponding variational inequality as shown in Appendix \ref{appendix::SB}. From an economic perspective, a high level of risk aversion, together with a rapidly increasing marginal burden of effort, facilitates a stable and economically consistent characterization of the principal’s value function.
    \item Theorem \ref{partialconcave}.2 provides an important intermediate result that enables us to establish Theorem \ref{partialconcave}.3. In particular, the condition $\eta^{n,*}(t,y) < y $, for all $n \in \{1,\ldots, N\}$, corresponds to the no-firing condition, indicating that it is never optimal for the principal to terminate the contract prematurely by actively dismissing the agent.
    \end{enumerate}
\end{remark}

\section{Numerical results}\label{numerics}
In this section, we use numerical methods to extract economic insights from the solutions derived in the previous sections. We fix the agent’s utility function $u(y) := y^{1/3}$, the cost of effort  $h(a) := \frac{1}{2}a^{2} + a$, the action set $A := [0,4]$, and the transaction ratio $k = 0.1$. We approximate the solution to the iterated HJB inequalities using the numerical scheme presented in \cite{reisinger2018penalty}. After verifying the assumptions in \cite[Assumption 1, Remark 1]{reisinger2018penalty}, we apply the convergence result in \cite[Theorem 4.7]{reisinger2018penalty}, which guarantees that the numerical solution converges to the unique viscosity solution of the recursive system of HJB variational inequalities \eqref{vnpdevi}. By Theorem \ref{hjbvi:pp:in}, this corresponds to the principal’s value function.

\medskip
Given $\epsilon>0$, we propose the following penalized problem for deriving the principal's value $V^{p,n}, n \in \{1,\cdots,N\}$ sequentially
\begin{small}
\begin{equation}\label{Vpkue}
    \begin{split}
        0 & =G^\e\left(t,y, u^{\e}, D u^{\e}, D^2 u^{\e}\right) \\
        & = \begin{cases} - u_t^\e- \underset{z \in [0,\infty),\hat{a} \in \partial h^{\star}(z)}{\sup}\{A^{z,\hat{a}} u^\e+&f^{z,\hat{a}}(u^{\e}) \}-\frac{1}{\e}\left(\mathcal{M}V^{p,n-1}-u^\e\right)^{+} ,(t,y) \in [0,T) \times [0,\infty),\\
        u^\e(T, y) = g^{n}(y) ,\quad y \in [0,\infty) ,
        &u^\e(t, 0) = 0  ,\quad t\in [0,T],
        \end{cases}
    \end{split}
\end{equation}    
\end{small}
where $A^{z,\hat{a}} u^{\e}(t, y)= r(y + h(\hat{a}))u^{\e}_y  + \frac{1}{2}u^{\e}_{yy}r^{2}\sigma^{2}z^{2}$, $f^{\hat{a}}\left(u^{\e}\right) = \rho(\hat{a} -u^{\e} )$, and $u^{\e}$ is a penalized approximation to $V^{p,n}$, $n \in \{1 ,\cdots, N\}$.

Let $U_i^m$ denote the discrete approximation of the solution to \eqref{Vpkue} at the node $(t_m, y_i)$. We denote by $\Delta_{+} U_i^m$ (resp. $\Delta_{-} U_i^m$) the one-step forward (resp. backward) difference of $U$ along the $y$–coordinate, and by $\Delta U_i^m=\left(\Delta_{+} U_i^m+\Delta_{-} U_i^m\right)$ the central difference of $U$ at the node $(t_m, y_i)$.



Next, we propose an implicit numerical scheme adapted to our setting. For simplicity, we consider a uniform spatial grid $\left\{y_i\right\}_i=h \mathbb{Z^{+}}$ on $\mathbb{R}^{+}$ and a time partition $\left\{t_m\right\}_{m=0}^M$ on $[0, T]$ with $\max _m\left|t_{m+1}-t_m\right|=\Delta t$. The implicit formulation is suitable for handling large time steps, especially when the penalty parameter $\frac{1}{\epsilon}$ is large. The scheme is therefore defined as follows:
$U_i^T=g^{n}\left(y_i\right)$, for all $i \in  \mathbb{Z^{+}}$, and for any given $m\in \{0, \ldots, M-1\}$
\begin{equation}\label{GhU}
    \begin{split}
    0 & =G_h\left(t_{m}, y_i, U_i^m,\left\{U_a^{b+1}\right\}_{(a, b) \neq(i, m)}\right) \\
& =\inf_{z \in [0,\infty), \hat{a} \in \hat{\mathcal{A}}(z)}\left(\frac{U_i^m-U_i^{m+1}}{\Delta t}-A^{z,\hat{a}} U_i^m-\bar{f}^{\hat{a}}\left( U_i^m\right)\right), \quad i \in \mathbb{Z}^d
    \end{split}
\end{equation}
where $A_{h}^{z,\hat{a}} U_i^m =r(y + h(\hat{a})) \frac{\Delta U_i^m}{2 h}  + \frac{1}{2} r^{2}\sigma^{2}z^{2} \frac{\Delta_{+} U_i^m-\Delta_{-} U_i^m}{h}$, $\bar{f}^{\hat{a}}\left( U_i^m\right) = \rho(\hat{a} - U_i^m ) + \frac{1}{\e}\left(V^{p,n-1}- U_i^m\right)^{+} $.\\
Then we focus on the implementation of (\ref{GhU}) through a policy iteration method. Firstly, we truncate the discrete equation (\ref{GhU}) by localizing it onto a chosen bounded computational domain, specifying the behaviour of the solution outside it. Hence, we consider the following finite-dimensional problem: for any given $u^{m+1} \in \mathbb{R}^I$, we aim to find $u \in \mathbb{R}^I$ such that for each $i \in \mathcal{I}:=\{1, \ldots, I\}$,
$$
0=\mathcal{G}_h^{m}[u]_i= \inf_{z \in [0,\infty), \hat{a} \in \partial h^{*}(z)} \left(\frac{u_i^m-u_i^{m+1}}{\Delta t} - A^{z,\hat{a}} u_i^m-\bar{f}^{\hat{a}}\left( u_i^m\right)\right).
$$
Let $u^{m+1}$ be a given solution at the previous discrete time point, $\mathbf{z} \in \mathbb{R^+}^I$, and $u \in \mathbb{R}^I$. We introduce the gradient of $\mathcal{G}_h^{m}[u]$ as follows
\begin{equation}\label{HessianG}
    \mathcal{L}^{\mathbf{z}}[u]_i:= \inf_{\hat{a} \in \partial h^{*}(z)} \left(I-\Delta t A^{z_i,\hat{a}}\right)_i-\Delta t\left(-\rho \cdot I+\zeta^{+}[u]_i\right), \quad i=1, \ldots, I,
\end{equation}
where $\mathcal{L}^{\mathbf{z}}[u]_i$ denotes by the $i$-th row of matrix $\mathcal{L}^z[u]$, $\zeta^{+}[u]_i$ is the digonal matrix: $\zeta^{+}[u] = \{\zeta_{ij}[u]\}$ with  $\zeta_{i i}[u]=- \frac{1}{\e} 1_{\left\{\mathcal{M}V^{p,n-1}\left(t_m, y_i\right)-u_i>0\right\}}$, for each $i \in \mathcal{I}$.

Our policy iteration algorithm is an adaptation to (\cite{reisinger2018penalty}, Algorithm 1).
\vspace{-0.2cm}
\begin{Algorithm}
Set $u^{(0)}=u^{m+1}$. Given $u^{(k)}, k \geq 0$, the next iterate $u^{(k+1)}$ is computed as follows:
\begin{enumerate}
    \item Policy improvement step: Compute $\mathbf{z}^{(k)}=\left\{z_i^{k}\right\}_{i=1}^I$ such that for each $i \in \mathcal{I}$, $$
    \mathbf{z}^{k} \in \underset{\mathbf{z} \in \mathbb{R^+}^I}{\arg \min } \mathcal{G}\left[u^{(k)}\right]_i.
    $$ 
    \vspace{-0.6cm}
    \item Policy evaluation step: Compute $u^{(k+1)} \in \mathbb{R}^I$ by solving
    $$    \mathcal{G}\left[u^{(k)}\right]+\mathcal{L}^{(k+1)}\left[u^{(k)}\right]\left(u^{(k+1)}-u^{(k)}\right)=0,
    $$
    where $\mathcal{L}^{(k+1)}\left[u^{(k)}\right]$ is the matrix (\ref{HessianG}) evaluated at the control $z^{(k+1)}$ and the iterate $u^{(k)}$.
    \end{enumerate}
\end{Algorithm}
After introducing the policy iteration algorithm, we proceed to its implementation. We consider a bounded truncation region $y \in [0, 2]$, with the number of spatial grid points set to $I = 40$. The time horizon is fixed at $T = 2$, with the number of time steps set to $M = 100$. The penalization parameter is chosen as $\epsilon = 0.001$.
\subsection{Comparison between first-best and second-best}
\vspace{-0.3cm}
\begin{figure}[H]
     \centering
     \begin{subfigure}[b]{0.4\textwidth}
         \centering
         \includegraphics[width=\textwidth]{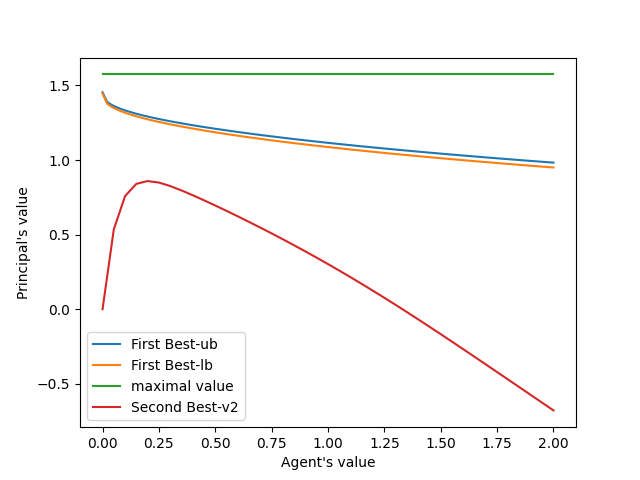}
         \caption{Principal's value ($\delta \gamma < 1,\delta  = 0.36$).}
         \label{fig:fbsbcmp036}
     \end{subfigure}
     \quad
     \begin{subfigure}[b]{0.4\textwidth}
         \centering
         \includegraphics[width=\textwidth]{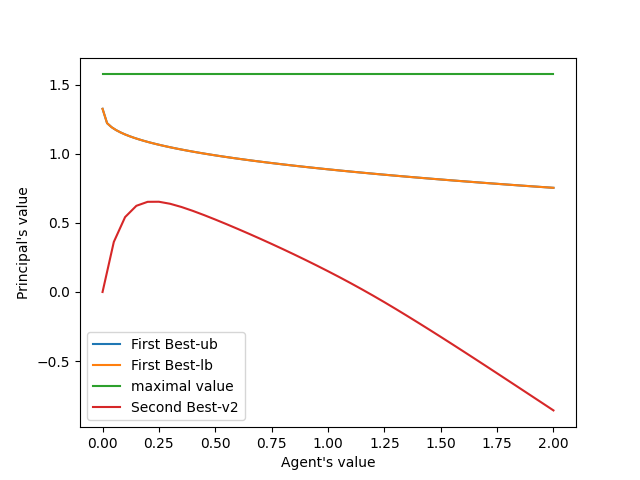}
         \caption{Principal's value ($\delta \gamma > 1,\delta < 1,\delta  = 0.4$)}
         \label{fig:fbsbcmp040}
     \end{subfigure}
     \quad
     \begin{subfigure}[b]{0.4\textwidth}
         \centering
         \includegraphics[width=\textwidth]{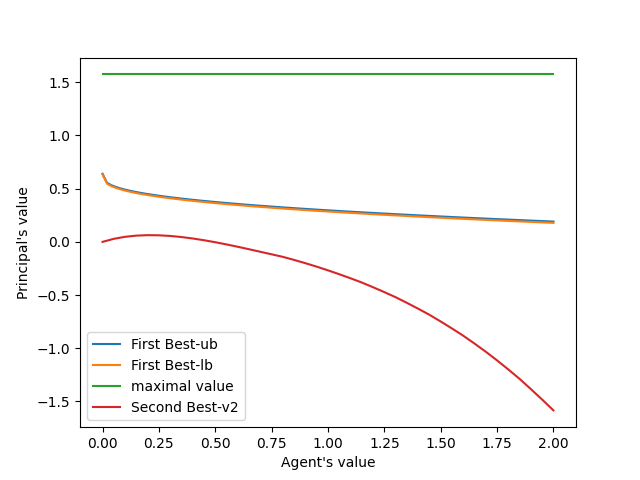}
         \caption{Principal's value ($\delta > 1,\delta  = 1.08$).}
         \label{fig:fbsbcmp108}
     \end{subfigure}
        \caption{Comparison of the Principal’s value under the first-best and second-best cases.}
        \label{fig:fbsbcmpdiffdeltas}
\end{figure}

The subfigures in Figure \ref{fig:fbsbcmpdiffdeltas} compare the maximal value attainable by the principal, the upper and lower bound estimators of the first-best value (“First Best-ub”, “First Best-lb”), and the principal’s value function $V^{p,2}(0, y)$ (“Second Best-v2”) at the beginning of the contract. Specifically, these results correspond to the first-best scenario in which the principal offers two bonus payments with flexible timings and one fixed terminal payment, with the associated value denoted by $V^{\mathrm{FB},\mathrm{P}}$.
\begin{enumerate}
    \item  Subfigure (\ref{fig:fbsbcmp036}) illustrates the case in which the principal is significantly more impatient than the agent, i.e., $\delta \gamma < 1$.
    \item Subfigure (\ref{fig:fbsbcmp040}) addresses the situation where the principal is more impatient than the agent, but not excessively so, i.e., $\gamma \delta > 1$ and $\delta < 1$.
    \item Subfigure (\ref{fig:fbsbcmp108}) depicts the case in which the agent is more impatient than the principal, i.e., $\delta > 1$.
\end{enumerate}
From the three subfigures in Figure \ref{fig:fbsbcmpdiffdeltas}, several observations follow. First, they partially confirm Theorem \ref{thm::FB}.4, namely that, for a fixed number of bonus payments—even when their timings are flexible—the first-best problem does not degenerate to the maximal value $\bar{a}(1 - \mathrm{e}^{-\rho T})$. Second, in all three cases, the principal’s first-best value serves as an upper bound for the second-best value $V^{p,2}(0, y)$, implying that $V^{p,2}(0, y)$ is ultimately decreasing in $y$. Finally, a comparison across the subfigures shows that, under the first-best setting, the principal benefits from greater levels of agent patience: the principal’s value decreases in $\delta>0$. 

\begin{figure}[H]
     \centering
     \begin{subfigure}[b]{0.4\textwidth}
         \centering
         \includegraphics[width=\textwidth]{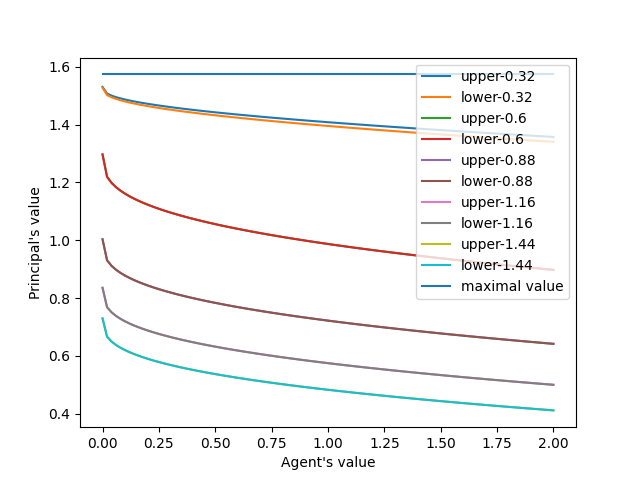}
         \caption{Principal's value under different impatience ratios $\delta$.}
         \label{fig:fbdiffdeltasnum4}
     \end{subfigure}
     \quad
     \begin{subfigure}[b]{0.4\textwidth}
         \centering
         \includegraphics[width=\textwidth]{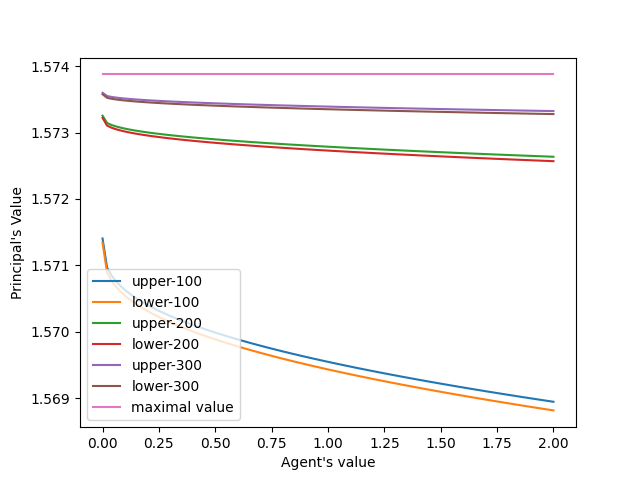}
         \caption{Principal's value under different payment schemes.}
         \label{fig:fbdiffnumcmp104}
     \end{subfigure}
        \caption{Principal's value with different ratio $\delta$ and different payment scheme.}
        \label{fig:fbdiffdeltaslimnum}
\end{figure}

Figure \ref{fig:fbdiffdeltaslimnum} illustrates the impact of $\delta$ and the number of bonus payments on the principal’s first-best value. Subfigure \ref{fig:fbdiffdeltasnum4} reports the upper and lower bound estimates of $V^{\mathrm{FB},\mathrm{P}}$ when the principal offers $N=4$ bonus payments under different values of $\delta$. For example, “upper–0.32” (“lower–0.32”) denotes the corresponding bounds when $\delta = 0.32$, confirming that the principal benefits from greater agent patience. Subfigure \ref{fig:fbdiffnumcmp104} compares the principal’s value across different numbers of bonus payments under the same $\delta = 1.04 > 1$, showing that increasing the number of payments improves the principal’s value. Taken together, these results indicate that the principal consistently benefits from both a smaller $\delta$ and a larger number of bonus payments, with the value converging toward the maximal level $\bar{a}(1 - \mathrm{e}^{-\rho T})$ as the number of payments increases.
\subsection{Discretionary vs. scheduled  payments}

Figure \ref{fig:sb_pv_uneven_gammadelta<1} compares the principal’s value at the beginning of the contracting period in two scenarios: one in which the principal is allowed to choose a flexible timing for a single bonus payment, and another in which the principal must follow a fixed schedule with two payments. For example, the label “$v1\text{–}0.24$” refers to the case in which the principal makes two payments—one at time $0.24$ and the other at the end of the contracting period, $T = 2$. In contrast, $v^{p,1}$ denotes the principal’s value at the beginning of the contract when a single bonus payment is made at a flexible time of the principal’s choosing. The results clearly show that the principal’s value under the flexible bonus-payment scheme consistently dominates that of any fixed two-payment scheme.

\begin{figure}[H]
     \centering
     \begin{subfigure}[b]{0.47\textwidth}
         \centering
         \includegraphics[width=\textwidth]{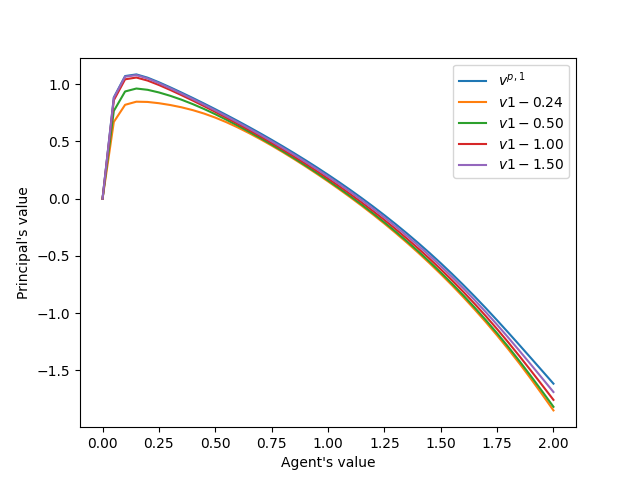}
         \caption{Principal's value ($\delta \gamma < 1,\delta  = 0.2$).}
         \label{fig:sb_pv_uneven-delta-0.2}
     \end{subfigure}
     \quad
     \begin{subfigure}[b]{0.47\textwidth}
         \centering
         \includegraphics[width=\textwidth]{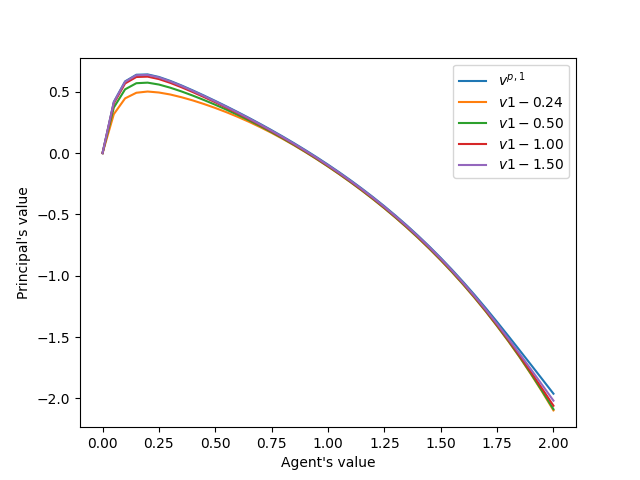}
         \caption{Principal's value ($\delta \gamma < 1,\delta  = 0.32$).}
         \label{fig:sb_pv_uneven-delta-0.32}
     \end{subfigure}
        \caption{Scheduled vs. discretionary bonuses when the principal is excessively more impatient than the agent.}
        \label{fig:sb_pv_uneven_gammadelta<1}
\end{figure}

Moreover, when $\delta \gamma < 1$, meaning that the agent is relatively patient, the principal’s value under a fixed-payment scheme increases as the first payment date is delayed. Specifically, a fixed scheme with a later first payment yields a higher principal value than one with an earlier first payment. This suggests that the principal benefits more from fixed-timing schemes when the first bonus payment is scheduled closer to the end of the contracting period.

\begin{figure}[H]
     \centering
     \begin{subfigure}[b]{0.47\textwidth}
         \centering
         \includegraphics[width=\textwidth]{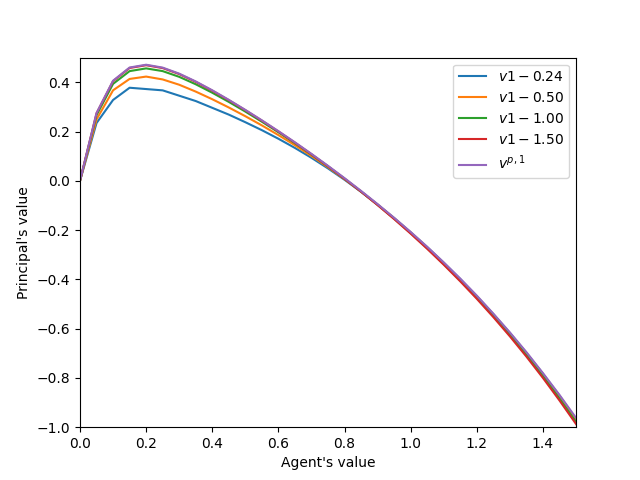}
         \caption{Principal's value ($\delta \gamma > 1,\delta < 1,\delta  = 0.4$).}
         \label{fig:sb_pv_uneven-delta-0.4}
     \end{subfigure}
     \quad
     \begin{subfigure}[b]{0.47\textwidth}
         \centering
         \includegraphics[width=\textwidth]{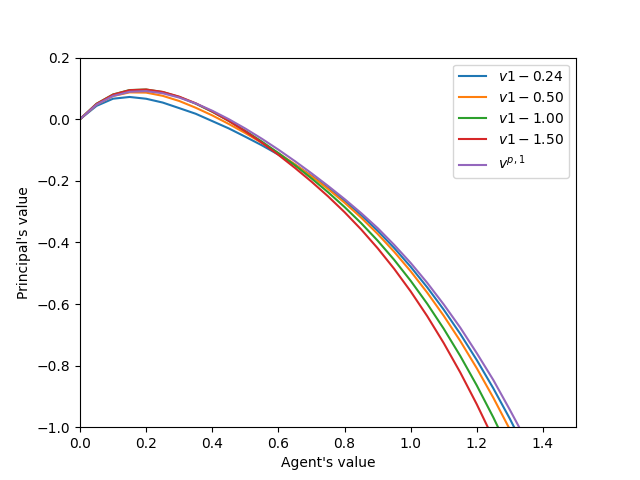}
         \caption{Principal's value ($\delta \gamma > 1,\delta < 1,\delta  = 0.8$).}
         \label{fig:sb_pv_uneven-delta-0.8}
     \end{subfigure}
        \caption{Scheduled vs. discretionary bonuses when the principal is more impatient than the agent, but not excessively.}
        \label{fig:sb_pv_uneven_gammadelta>1&delta<1}
\end{figure}

Next, we consider the case in which the principal is more impatient than the agent, but not excessively so ($\gamma \delta > 1$, $\delta < 1$). Figure \ref{fig:sb_pv_uneven_gammadelta>1&delta<1} compares the deterministic two-payment scheme (fixed timings) with the one-bonus scheme (flexible timing). First, we observe that $v^{p,1}$ strictly dominates the principal’s value at the beginning of the contract, regardless of the reservation-utility level, confirming the robustness of our numerical results. Second, although the model setup differs from that of \cite{alvarez2025sequential}, we obtain a similar numerical pattern: under the deterministic two-payment scheme, when the reservation utility is low, the principal benefits from delaying payments as much as possible, whereas when the reservation utility is high, the principal benefits from advancing the payments.

\begin{figure}[H]
     \centering
     \begin{subfigure}[b]{0.47\textwidth}
         \centering
         \includegraphics[width=\textwidth]{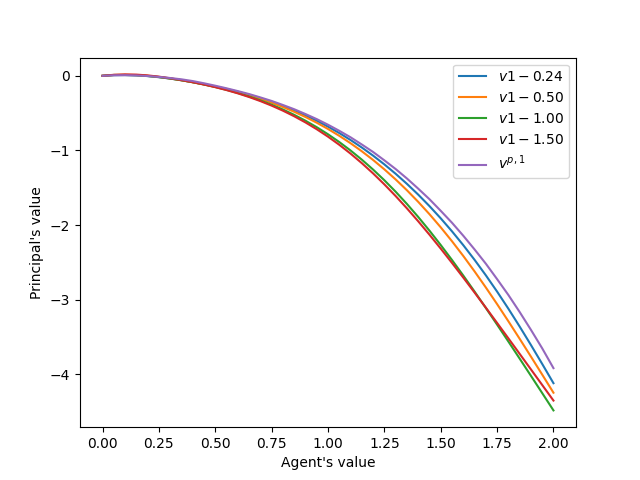}
         \caption{Principal's value ($\delta  > 1,\delta  = 1.08$).}
         \label{fig:sb_pv_uneven-delta-1.08}
     \end{subfigure}
     \quad
     \begin{subfigure}[b]{0.47\textwidth}
         \centering
         \includegraphics[width=\textwidth]{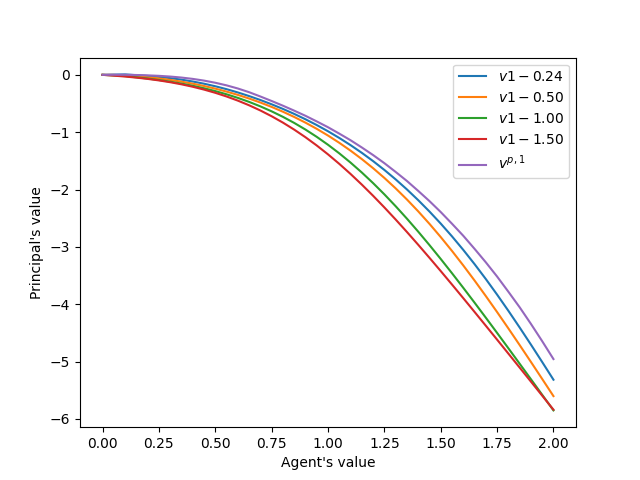}
         \caption{Principal's value ($\delta > 1,\delta  = 1.4$).}
         \label{fig:sb_pv_uneven-delta-1.4}
     \end{subfigure}
        \caption{Scheduled vs. discretionary bonuses when the agent is more impatient than the principal.}
        \label{fig:sb_pv_uneven_gammadelta>1}
\end{figure}

Similar to Figure \ref{fig:sb_pv_uneven_gammadelta<1}, Figure \ref{fig:sb_pv_uneven_gammadelta>1} illustrates the case in which the agent is more impatient than the principal, i.e., $\delta > 1$. As in Figure \ref{fig:sb_pv_uneven_gammadelta<1}, we observe that the principal’s value function—when the principal has the flexibility to choose the timing of the bonus payment at the beginning of the contract—dominates the value functions under fixed-timing bonus schemes.

In contrast to Figure \ref{fig:sb_pv_uneven_gammadelta<1}, however, the principal now benefits from scheduling the bonus payment earlier. As the timing of the bonus payment shifts closer to the start of the contract, the principal’s value increases and converges to the level achieved when the principal has full flexibility in choosing the payment time.

\subsection{The influence of the relative patience}
\begin{figure}[H]
     \centering
     \begin{subfigure}[b]{0.44\textwidth}
         \centering
         \includegraphics[width=\textwidth]{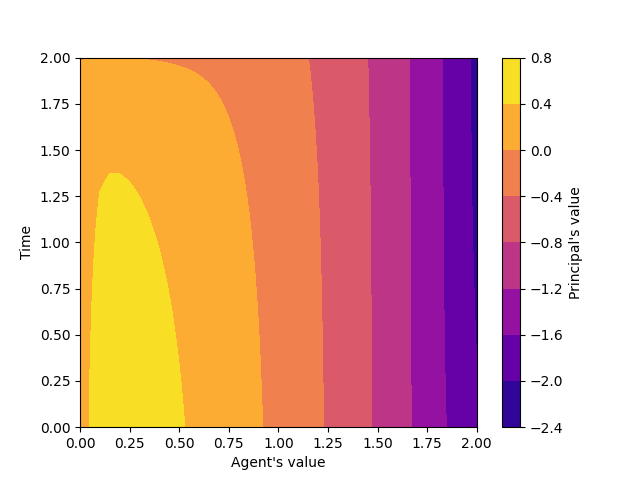}
         \caption{Principal's value $V^{p,1}$.}
         \label{fig:principalvi1v008025}
     \end{subfigure}
     \begin{subfigure}[b]{0.44\textwidth}
         \centering
         \includegraphics[width=\textwidth]{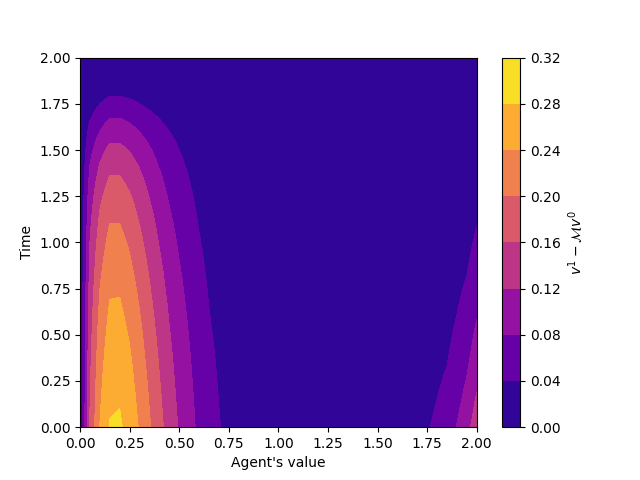}
         \caption{ $V^{p,1} - \mathcal{M}V^{p,0}$.}
         \label{fig:principalvi1diff008025}
     \end{subfigure}
      \begin{subfigure}[b]{0.44\textwidth}
         \centering
         \includegraphics[width=\textwidth]{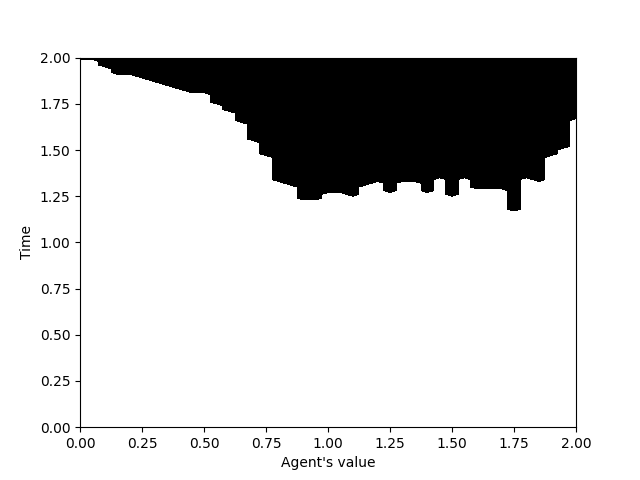}
         \caption{Bonus payment region.}
         \label{fig:interventionvi1008025}
     \end{subfigure}
     \begin{subfigure}[b]{0.44\textwidth}
         \centering
         \includegraphics[width=\textwidth]{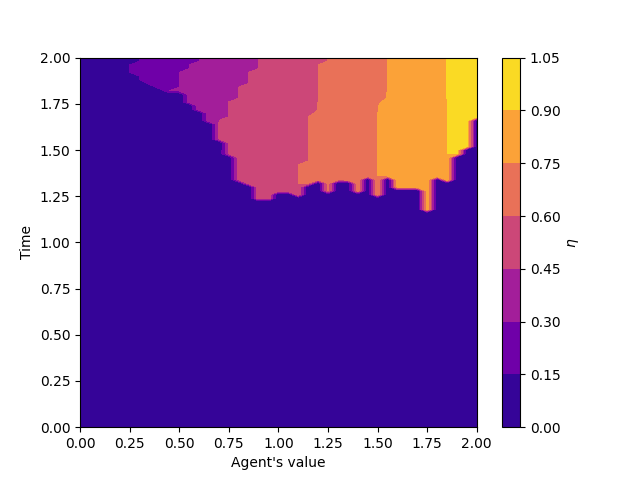}
         \caption{Bonus payment scheme $\eta^{1,*}(t,y)$.}
         \label{fig:scheme008025}
     \end{subfigure}
        \caption{Principal's value and bonus payment structure when $\delta = 0.32, \delta\gamma < 1$.}
        \label{fig:principalvi1008025}
\end{figure}

Figure \ref{fig:principalvi1008025} illustrates the case in which the principal is substantially more impatient than the agent. Subfigure \ref{fig:principalvi1v008025} shows the principal’s value function $V^{p,1}(t, y)$. We observe the presence of informational rent, meaning that at the beginning of the contract, the principal optimally offers the agent a utility level strictly above the reservation value. Subfigure \ref{fig:interventionvi1008025} depicts the intervention region, defined as
$$
\left\{(t, y) \in[0, T) \times(0, \infty): V^{p, 1}(t, y)=\mathcal{M} V^{p, 0}(t, y), \eta^{1,*}(t, y)>0\right\}.
$$
The black region takes the shape of a trapezoid, with the top base longer than the bottom base. This geometry arises from the coexistence of limited liability and a high discount rate $\rho$. Specifically, any downward jump is undesirable for the principal when the agent’s continuation utility is relatively low, as it increases the risk of hitting the zero boundary, at which point the agent ceases exerting effort—a situation detrimental to the principal. At the same time, given the large discount rate, the principal prefers to delay bonus payments to benefit from discounting and thereby reduce the effective cost.

Finally, Subfigure \ref{fig:scheme008025} presents the single-bonus payment scheme $\eta^{1,*}(t, y)$. As expected, it cross-validates the intervention region in Subfigure \ref{fig:interventionvi1008025}. Moreover, it shows that, for fixed $t$, the scheme $\eta^{1,*}(t, y)$ is monotone in $y$.
\begin{figure}[H]
     \centering
     \begin{subfigure}[b]{0.44\textwidth}
         \centering
         \includegraphics[width=\textwidth]{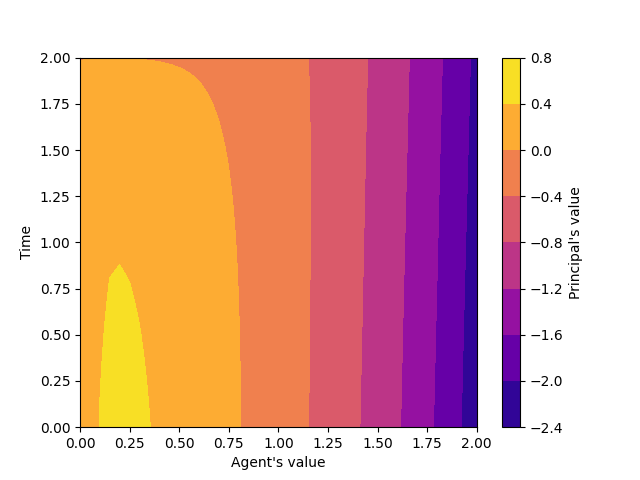}
         \caption{Principal's value $V^{p,1}$.}
         \label{fig:principalvi1v010025}
     \end{subfigure}
     \begin{subfigure}[b]{0.44\textwidth}
         \centering
         \includegraphics[width=\textwidth]{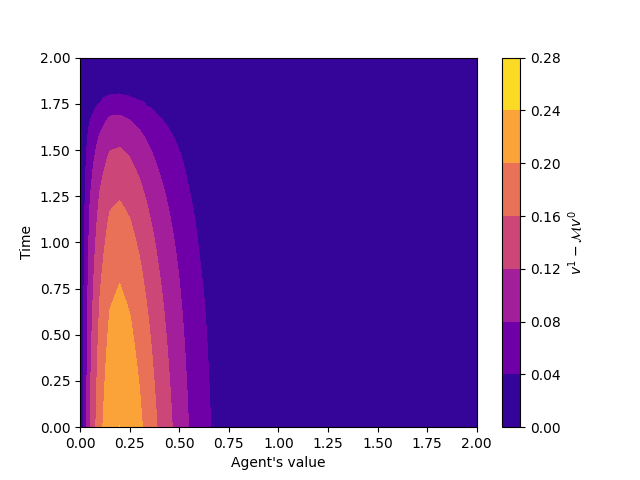}
         \caption{$V^{p,1} - \mathcal{M}V^{p,0}$.}
         \label{fig:principalvi1diff010025}
     \end{subfigure}
      \begin{subfigure}[b]{0.44\textwidth}
         \centering
         \includegraphics[width=\textwidth]{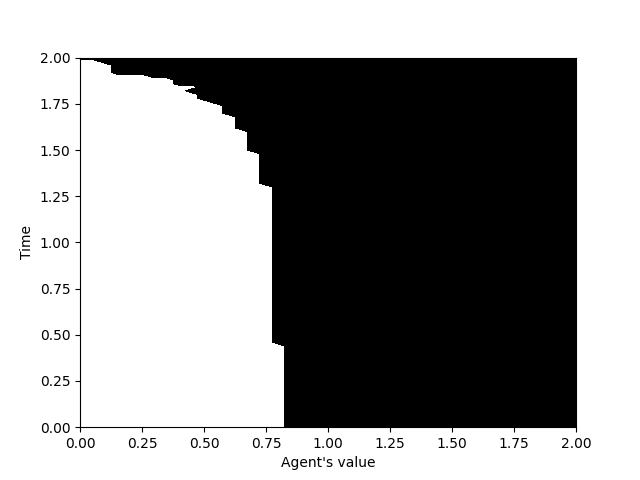}
         \caption{Bonus payment region.}
         \label{fig:interventionvi1010025}
     \end{subfigure}
     \begin{subfigure}[b]{0.44\textwidth}
         \centering
         \includegraphics[width=\textwidth]{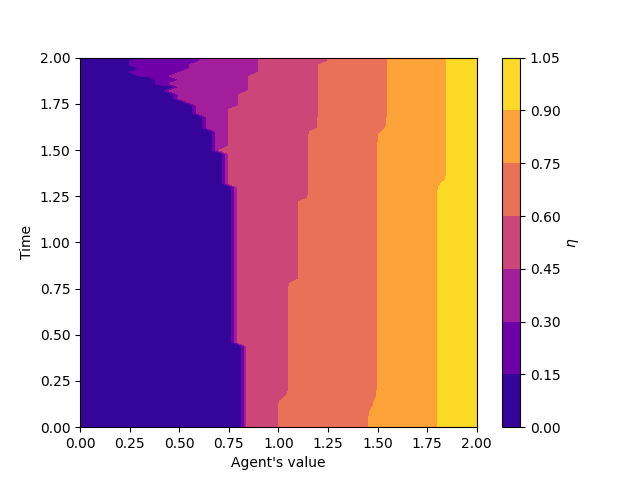}
         \caption{Bonus payment scheme $\eta^{1,*}(t,y)$.}
         \label{fig:scheme010025}
     \end{subfigure}
        \caption{Principal's value and bonus payment structure when $\delta = 0.4, 1 <\delta\gamma, \delta < 1$.}
        \label{fig:principalvi1010025}
\end{figure}

Figure \ref{fig:principalvi1010025} illustrates the case in which the principal is more impatient than the agent, but not excessively so. Similar to Subfigure \ref{fig:principalvi1v008025}, Subfigure \ref{fig:principalvi1v010025} depicts the principal’s value function $V^{p,1}(t, y)$. Informational rent still exists; however, compared with Subfigure \ref{fig:principalvi1v008025}, an increase in the agent’s discount rate is detrimental to the principal. The dark region in Subfigure \ref{fig:interventionvi1010025} represents the positive bonus-payment region. Compared with Subfigure \ref{fig:interventionvi1008025}, the key difference is the emergence of the “golden hello” phenomenon. This means that when the agent commits to the contract with a relatively high reservation utility, the principal optimally offers a positive payment at the start of the contract. This arises due to the increase in the relative impatience ratio, $\tfrac{r}{\rho}$. Moreover, as $r$ increases, the positive bonus-payment region expands.

Finally, Subfigure \ref{fig:scheme010025} shows the single-bonus payment scheme $\eta^{1,*}(t, y)$. Two properties are evident: (i) for fixed $t$, the scheme $\eta^{1,*}(t, y)$ is monotone in $y$, that is, increasing with respect to the agent’s utility; and (ii) for fixed $y$, the principal tends to offer a similar bonus amount across different times.

\begin{figure}[H]
     \centering
     \begin{subfigure}[b]{0.44\textwidth}
         \centering
         \includegraphics[width=\textwidth]{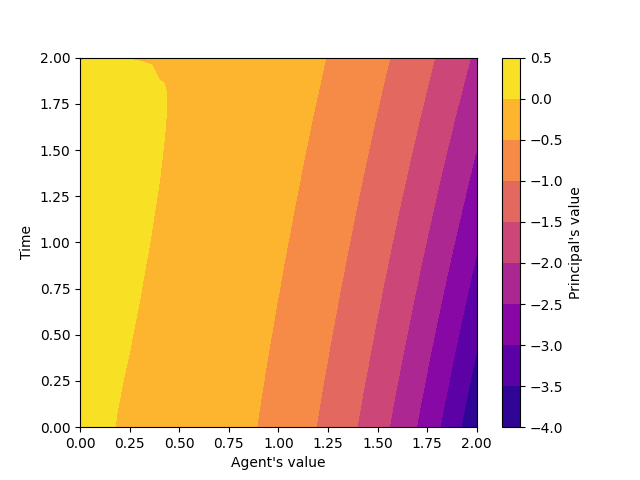}
         \caption{Principal's value $V^{p,1}$.}
         \label{fig:principalvi1v027025}
     \end{subfigure}
     \begin{subfigure}[b]{0.44\textwidth}
         \centering
         \includegraphics[width=\textwidth]{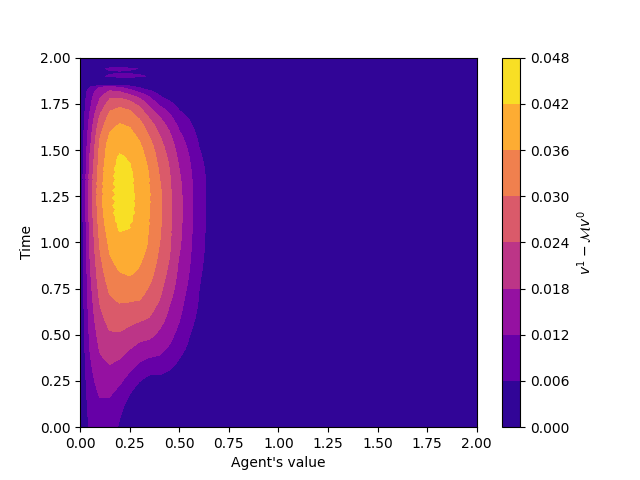}
         \caption{$V^{p,1} - \mathcal{M}V^{p,0}$.}
         \label{fig:principalvi1diff027025}
     \end{subfigure}
      \begin{subfigure}[b]{0.44\textwidth}
         \centering
         \includegraphics[width=\textwidth]{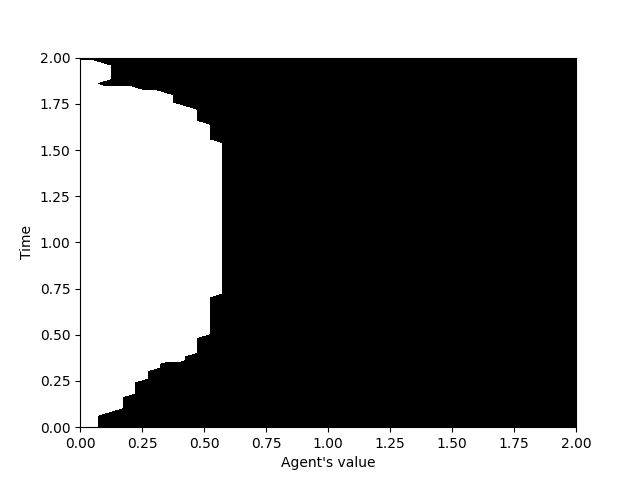}
         \caption{Bonus payment region.}
         \label{fig:interventionvi1027025}
     \end{subfigure}
     \begin{subfigure}[b]{0.44\textwidth}
         \centering
         \includegraphics[width=\textwidth]{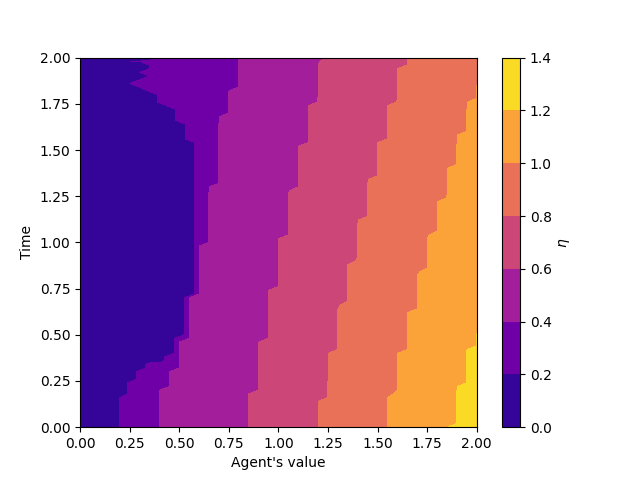}
         \caption{Bonus payment scheme $\eta^{1,*}(t,y)$.}
         \label{fig:scheme027025}
     \end{subfigure}
        \caption{Principal's value and bonus payment structure when $\delta = 1.08,\delta > 1$.}
        \label{fig:principalvi1027025}
\end{figure}

Figure \ref{fig:principalvi1027025} illustrates the case in which the agent is more impatient than the principal. Subfigure \ref{fig:principalvi1v027025} depicts the principal’s value function $V^{p,1}(t, y)$. Contrary to the conclusion in \cite{possamai2024there}, we observe that in the presence of bonus payments, informational rent persists even when the agent is more impatient than the principal. This indicates that informational rent is preserved through the flexibility of bonus-payment timing.

Next, Subfigure \ref{fig:principalvi1diff027025} presents the difference between the principal’s value function with bonus payments and the obstacle $\mathcal{M}V^{p,1}$, which cross-validates with Subfigure \ref{fig:interventionvi1027025} and confirms the robustness of our numerical results. Subfigure \ref{fig:interventionvi1027025} displays the region in which strictly positive bonus payments occur. We observe that, as the agent becomes more impatient, the principal tends to offer a signing bonus to agents with relatively low reservation utilities. However, the principal remains reluctant to provide bonuses to agents whose reservation utility is extremely low, since such agents are likely to exit the project prematurely due to the limited-liability constraint.

Finally, Subfigure \ref{fig:scheme027025} shows the bonus-payment scheme. In general, we find that $\eta^{1,*}(t_1, y) \leq \eta^{1,*}(t_2, y)$ for $t_1 \geq t_2$, reflecting that, as the agent’s impatience increases, the principal must pay more to achieve the same incentive effect.

\subsection{Multiple bonus payments}
In this section, we examine the optimal bonus payments when the principal has multiple rights. In particular, we consider the case in which the principal is allowed to offer $N=4$ bonus payments over the contracting horizon $T = 2$. The cost of effort is specified as $h(a) = \frac{1}{2}a^{2} + a$, the admissible effort set is $A = [0, 4]$, the volatility is fixed at $\sigma = 0.9$, and the inverse utility function is given by $u^{-1}(y) = y^{3}$.

\begin{figure}[H]
     \centering
     \begin{subfigure}[b]{0.44\textwidth}
         \centering
         \includegraphics[width=\textwidth]{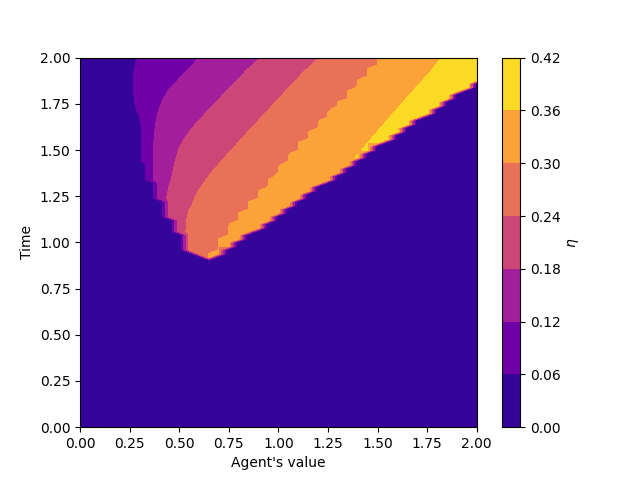}
         \caption{First bonus payment scheme $\eta^{4,*}(t,y)$.}
         \label{fig:scheme007025vi1vi0}
     \end{subfigure}
     \begin{subfigure}[b]{0.44\textwidth}
         \centering
         \includegraphics[width=\textwidth]{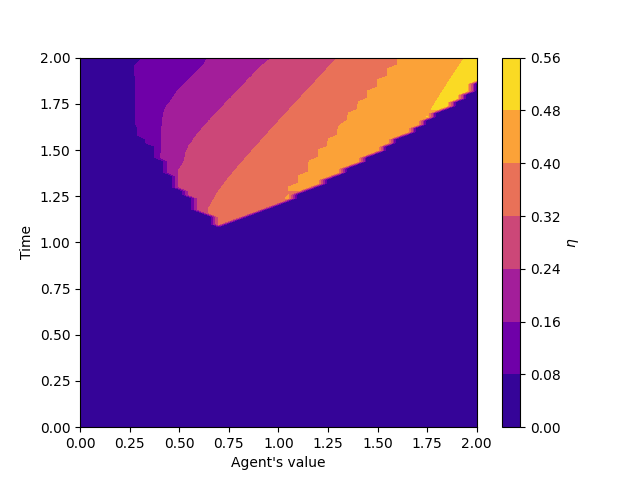}
         \caption{Second bonus payment scheme $\eta^{3,*}(t,y)$.}
         \label{fig:scheme007025vi2vi1}
     \end{subfigure}
      \begin{subfigure}[b]{0.44\textwidth}
         \centering
         \includegraphics[width=\textwidth]{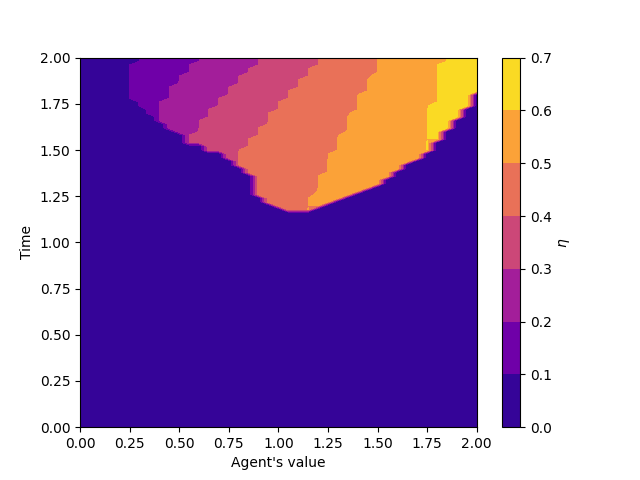}
         \caption{Third bonus payment scheme $\eta^{2,*}(t,y)$.}
         \label{fig:scheme007025vi3vi2}
     \end{subfigure}
     \begin{subfigure}[b]{0.44\textwidth}
         \centering
         \includegraphics[width=\textwidth]{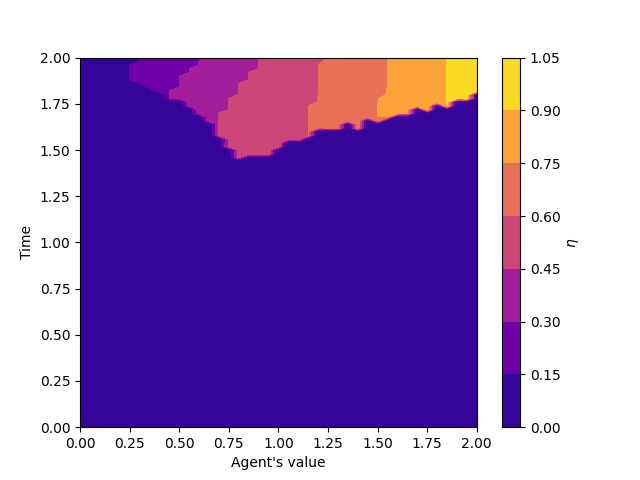}
         \caption{Last bonus payment scheme $\eta^{1,*}(t,y)$.}
         \label{fig:scheme007025vi4vi3}
     \end{subfigure}
        \caption{Principal's bonus payment scheme when $\delta = 0.28$.}
        \label{fig:ppschemes007025vivi-1}
\end{figure}

Figure \ref{fig:ppschemes007025vivi-1} illustrates the bonus-payment scheme when the principal is more impatient than the agent, i.e., when $\delta \gamma < 1$. Here, $\eta^{n,*}(t, y)$ denotes the optimal bonus-payment function characterized in Theorem \ref{hjbvi:pp:in}. A comparison across the subfigures shows that, as the number of bonuses increases, the principal generally provides smaller bonuses. This reflects the fact that having more opportunities to intervene enhances her flexibility in incentivizing the agent. Moreover, the intervention region expands with the number of payments, capturing the principal’s impatience: she tends to trigger the first bonus payment earlier to accommodate her own preferences.

\begin{figure}[H]
     \centering
     \begin{subfigure}[b]{0.44\textwidth}
         \centering
         \includegraphics[width=\textwidth]{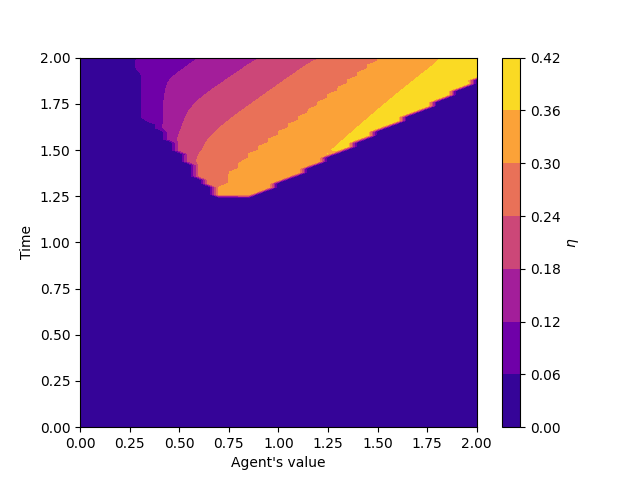}
         \caption{First bonus payment scheme $\eta^{4,*}(t,y)$.}
         \label{fig:scheme010025vi4vi3}
     \end{subfigure}
     \begin{subfigure}[b]{0.44\textwidth}
         \centering
         \includegraphics[width=\textwidth]{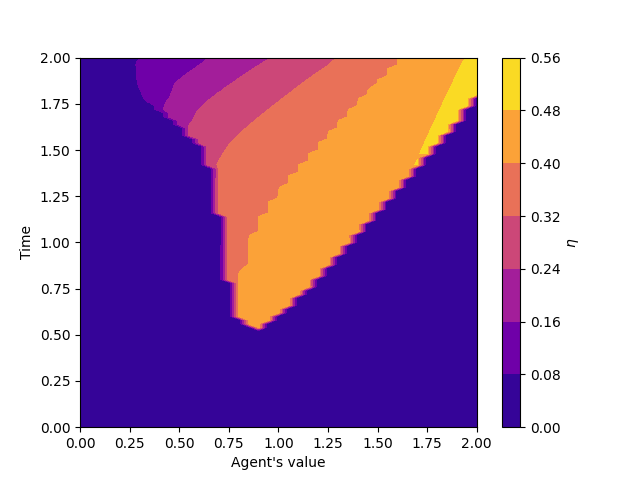}
         \caption{Second bonus payment scheme $\eta^{3,*}(t,y)$.}
         \label{fig:scheme010025vi3vi2}
     \end{subfigure}
      \begin{subfigure}[b]{0.44\textwidth}
         \centering
         \includegraphics[width=\textwidth]{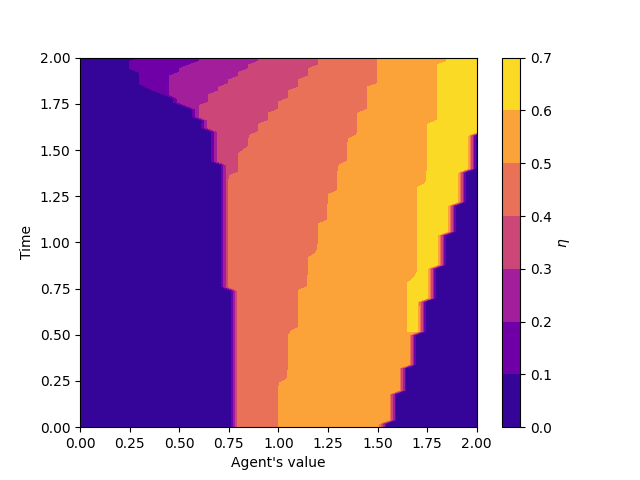}
         \caption{Third bonus payment scheme $\eta^{2,*}(t,y)$.}
         \label{fig:scheme010025vi2vi1}
     \end{subfigure}
     \begin{subfigure}[b]{0.44\textwidth}
         \centering
         \includegraphics[width=\textwidth]{figures/sb_principal_scheme_010_025_vi1.png}
         \caption{Last bonus payment scheme $\eta^{1,*}(t,y)$.}
         \label{fig:scheme010025vi1}
     \end{subfigure}
        \caption{Principal's bonus payment scheme when $\delta = 0.4$.}
        \label{fig:ppschemes010025vivi-1}
\end{figure}

Figure \ref{fig:ppschemes010025vivi-1} illustrates the bonus-payment schemes when the principal is more impatient than the agent, but not excessively so, i.e., $\delta \gamma > 1$ and $\delta < 1$. We first observe that, as the number of bonus payments increases, the “sign-on” bonus gradually disappears. This occurs because the principal gains greater flexibility in designing the payment scheme. Under the two-payment scheme, however, a sign-on bonus still arises when the agent’s reservation utility level is neither too high nor too low. Moreover, as the number of payments increases, the pattern of the first bonus payment in Figure \ref{fig:scheme010025vi4vi3} gradually aligns with that in Figure \ref{fig:scheme007025vi4vi3}, indicating that the additional flexibility from multiple payments can offset the distortions introduced by the choice of $\delta$.

\begin{figure}[H]
     \centering
     \begin{subfigure}[b]{0.44\textwidth}
         \centering
         \includegraphics[width=\textwidth]{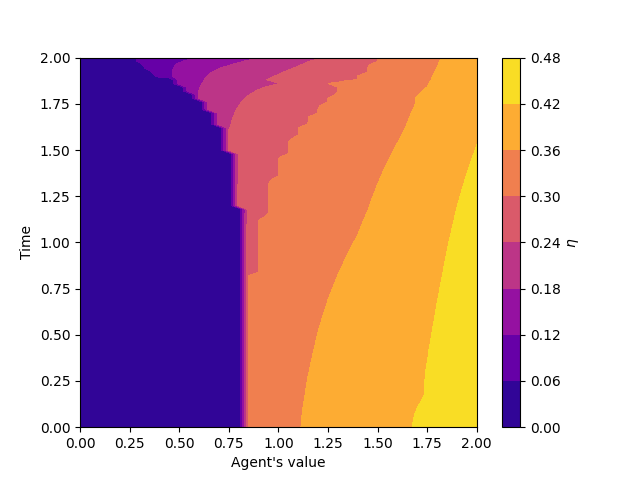}
         \caption{First bonus payment scheme $\eta^{4,*}(t,y)$.}
         \label{fig:scheme027025vi1vi0}
     \end{subfigure}
     \begin{subfigure}[b]{0.44\textwidth}
         \centering
         \includegraphics[width=\textwidth]{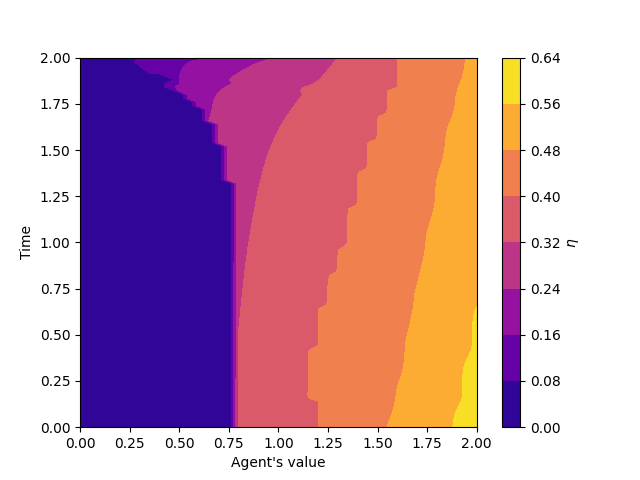}
         \caption{Second bonus payment scheme $\eta^{3,*}(t,y)$.}
         \label{fig:scheme027025vi2vi1}
     \end{subfigure}
      \begin{subfigure}[b]{0.44\textwidth}
         \centering
         \includegraphics[width=\textwidth]{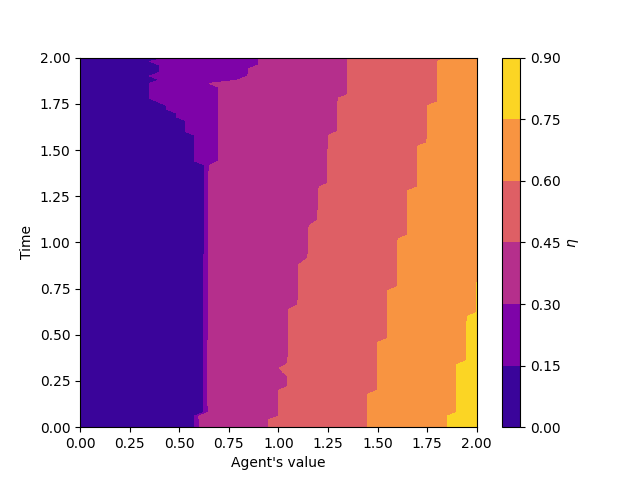}
         \caption{Third bonus payment scheme $\eta^{2,*}(t,y)$.}
         \label{fig:scheme027025vi3vi2}
     \end{subfigure}
     \begin{subfigure}[b]{0.44\textwidth}
         \centering
         \includegraphics[width=\textwidth]{figures/sb_principal_scheme_027_025_vi1.png}
         \caption{Last bonus payment scheme $\eta^{1,*}(t,y)$.}
         \label{fig:scheme027025vi4vi3}
     \end{subfigure}
        \caption{Principal's bonus payment scheme when $\delta = 1.08$.}
        \label{fig:ppschemes027025vivi-1}
\end{figure}
Figure \ref{fig:ppschemes027025vivi-1} depicts the bonus-payment schemes when the agent is more impatient than the principal, i.e., $\delta > 1$. We first observe that, as the number of bonus payments increases, the region of strictly positive payments gradually shrinks. Moreover, the maximum utility delivered to the agent from the first bonus payment is lower than that provided in later payments, reflecting the principal’s enhanced flexibility. This pattern is consistent with the results shown in Figures \ref{fig:ppschemes007025vivi-1} and \ref{fig:ppschemes010025vivi-1}.

Second, due to the agent’s high impatience, the region associated with the sign-on bonus diminishes. As the number of payments increases, the principal increasingly allocates such bonuses to agents with relatively high reservation utilities. However, this effect does not emerge immediately for small increases in the number of payments. A plausible conjecture is that, as the number of payments becomes sufficiently large, the sign-on bonus will eventually vanish.

\appendix
\section{Proof of Theorem \ref{thm::FB} }\label{appendix::FB}
\subsection{Proofs of Theorems \ref{thm::FB}.1, \ref{thm::FB}.2, \ref{thm::FB}.3}
    We address the first-best contracting problem by establishing upper and lower bounds on the principal’s value under varying levels of the agent’s reservation utility. By applying the standard Karush–Kuhn–Tucker (KKT) method, the first-best problem can be reformulated as
\begin{flalign*}
    &\quad v^{FB,P}(u(R))\\
    &:= \inf_{\lambda \leq 0}\sup_{\mathcal{C} \in \mathbb{C},\alpha \in \mathcal{A}} \{\lambda u(R) + J^{\mathrm{P}}(\mathbf{C}, \alpha) - \lambda J^{\mathrm{A}}(\mathbf{C}, \alpha) \}  \\
    &= \inf_{\lambda \leq 0} \left\{\lambda u(R)  + \sup_{(\alpha,\bar{\tau}_N, \bar{\xi}_N, \xi_T)} \mathbb{E}^{\mathbb{P}^{\alpha}} \Bigg[ - \sum_{i=1}^{N} 
    \mathrm{e}^{-\rho \tau_i}  \left(  u(\xi_i)\lambda \mathrm{e}^{(\rho-r)\tau_i} + (1 + k)\xi_i  \right)  \right. \\
    &\quad\quad\quad  \left.     - \mathrm{e}^{-\rho T} \left(u(\xi_T)\lambda \mathrm{e}^{(\rho-r)T} +  \xi_T  \right)+   \int_{0}^{T} \rho \mathrm{e}^{-\rho t}   \left(\alpha_t + h(\alpha_t) \delta \lambda \mathrm{e}^{(\rho-r)t} \right) \mathrm{d}t \Bigg] \right\}\\
    &= \inf_{\lambda \leq 0} \left\{\lambda u(R)  +\sum_{i=1}^{N} \sup_{(\xi_i,\tau_i)}
    \mathrm{e}^{-\rho \tau_i}  (-1 ) \cdot \left(  u(\xi_i)\lambda \mathrm{e}^{(\rho-r)\tau_i} + (1 + k)\xi_i  \right)   \right.\\
    &\quad\quad\quad\left.- \sup_{\xi_T} \mathrm{e}^{-\rho T} \left(u(\xi_T)\lambda \mathrm{e}^{(\rho-r)T} +  \xi_T  \right)+ \int_{0}^{T} \rho \mathrm{e}^{-\rho t} \sup_{\alpha}  \left(\alpha_t + h(\alpha_t) \delta \lambda \mathrm{e}^{(\rho-r)t} \right) \mathrm{d}t \right\}\\
    &=\inf_{\lambda \leq 0} \left\{\lambda u(R)  + \sum_{i=1}^{N} \sup_{(\tau_i)}
   \left\{- \mathrm{e}^{-\rho \tau_i} \left(  F^{*,k}(\lambda \mathrm{e}^{(\rho - r)\tau_i}) \right)  \right\} -  \mathrm{e}^{-\rho T} F^{*,0}(\lambda \mathrm{e}^{(\rho - r)T}) \right.\\
   &\left.\quad\quad\quad+ \int_{0}^{T}\rho \mathrm{e}^{-\rho t}G^{*}(\delta  \lambda \mathrm{e}^{(\rho-r)t} ) dt \right\}, 
\end{flalign*}
where the third equality comes from the fact that maximization over $(\alpha,\bar{\tau}_N, \bar{\xi}_N, \xi_T)$ reduces to the pointwise maximization inside the expectation and the integral. Then, the last inequality follows from the definition of $F^{*,k}$,$F^{*,0}$ and $G^{\star}$
\begin{align*}
    F^{*,k}(\lambda \mathrm{e}^{(\rho-r)\tau_i}) &:=  \inf_{y\geq 0}( y \lambda \mathrm{e}^{(\rho-r)\tau_i} - (1 + k)F(y))  =\inf_{\xi_i \geq 0} \left( u(\xi_i)\lambda \mathrm{e}^{(\rho-r)\tau_i} + (1 + k)\xi_i  \right) \;,\; F(y) = -y^{\gamma},\\
    G^*(\delta  \lambda \mathrm{e}^{(\rho-r)t} ) &:= \sup_{a \in A}(a + h(a) \delta  \lambda \mathrm{e}^{(\rho-r)t})  .
\end{align*}
Now we address the optimization on timings of the bonus payment $(\tau_1,\cdots,\tau_N)$, we consider the $ith$ bonus payment
\begin{equation*}
\frac{\partial - \mathrm{e}^{-\rho \tau_i} F^{*,k}(\lambda \mathrm{e}^{(\rho-r)\tau_i})}{ \partial \tau_i} = \rho \mathrm{e}^{-\rho \tau_i} (-  \lambda \mathrm{e}^{(\rho-r)\tau_i} )^{\frac{\gamma}{\gamma-1}}\left(   - \frac{1}{((1 + k) \gamma)^{\frac{1}{\gamma-1 }}}  \right)\left( \frac{\delta\gamma - 1 }{\gamma}\right).
\end{equation*}
It is clear that the monotonicity of $-\mathrm{e}^{-\rho \tau_i} F^{*,k}(\lambda \mathrm{e}^{(\rho-r)\tau_i})$ depends on the sign of $\delta\gamma - 1$. We now proceed to discuss the subcases based on the relationship between $\delta\gamma$ and 1. 
Firstly, when $\delta\gamma > 1$, based on nonpositivity of $\lambda$ and $\gamma > 1$, we have $\frac{\partial - \mathrm{e}^{-\rho \tau_i} F^{*,k}(\lambda \mathrm{e}^{(\rho-r)\tau_i})}{ \partial \tau_i} \leq 0$ and 
\begin{equation*}
    \begin{split}
    &\quad v^{FB,P}(u(R))\\
    &=\inf_{\lambda \leq 0} \left\{\lambda u(R)  + \sum_{i=1}^{N} \sup_{(\tau_i)}
   \left\{- \mathrm{e}^{-\rho \tau_i} \left(  F^{*}(\lambda \mathrm{e}^{(\rho - r)\tau_i}) \right)  \right\} -  \mathrm{e}^{-\rho T} F^{*}(\lambda \mathrm{e}^{(\rho - r)T}) + \int_{0}^{T}\rho \mathrm{e}^{-\rho t}G^{*}(\delta  \lambda \mathrm{e}^{(\rho-r)t} ) dt \right\}  \\ 
    &\leq \inf_{\lambda \leq 0} \left\{\lambda u(R) - N 
    F^{*}(\lambda) -  \mathrm{e}^{-\rho T} F^{*}(\lambda \mathrm{e}^{(\rho - r)T}) + \int_{0}^{T}\rho \mathrm{e}^{-\rho t}G^{*}(\delta  \lambda \mathrm{e}^{(\rho-r)t} ) dt \right\}  .
    \end{split}
\end{equation*}
Next, we consider one admissible bonus payment timing $\tau_1 = 0, \tau_{i+1} = \tau_i + \zeta, i \in \{0,\cdots,N-1\}$, where $0<\zeta < \frac{T}{N}$. Then we have the lower bound
\begin{equation*}
\begin{split}
  &\quad v^{FB,P}(u(R))\geq
        \inf_{\lambda \leq 0} \bigg\{\lambda u(R)  + \sum_{i=1}^{N} 
    \left\{- \mathrm{e}^{-\rho \tau_i} \left(  F^{*,k}(\lambda \mathrm{e}^{(\rho - r)\tau_i}) \right) \right\}\\
    &\quad\quad\quad\quad\quad\quad\quad\quad\quad\quad\quad\quad\quad\quad -  \mathrm{e}^{-\rho T} F^{*,k}(\lambda \mathrm{e}^{(\rho - r)T}) + \int_{0}^{T}\rho \mathrm{e}^{-\rho t}G^{*}(\delta  \lambda \mathrm{e}^{(\rho-r)t} ) dt \bigg\}.
\end{split}
\end{equation*}
\medskip 
Secondly, when $\delta\gamma < 1$, based on nonpositivity of $\lambda$, $\gamma > 1$, and $\frac{\partial - \mathrm{e}^{-\rho \tau_i} F^{*,k}(\lambda \mathrm{e}^{(\rho-r)\tau_i})}{ \partial \tau_i} \geq 0$, following similar argument, we have the upper bound
$$v^{P,FB}(u(R)) \leq \inf_{\lambda \leq 0} \left\{\lambda u(R) -  (N+1)\mathrm{e}^{-\rho T} F^{*}(\lambda \mathrm{e}^{(\rho - r)T}) + \int_{0}^{T}\rho \mathrm{e}^{-\rho t}G^{*}(\delta  \lambda \mathrm{e}^{(\rho-r)t} ) dt \right\}. $$
Considering admissible bonus payment timing $\tau_N = T - \zeta, \tau_i  = \tau_{i+1} - \zeta  , i \in \{0,\cdots,N-1\}$, where $0<\zeta < \frac{T}{N}$. Then we have the lower bound 
\begin{equation*}
    \begin{split}
    & v^{P,FB}(u(R))
    \geq  \inf_{\lambda \leq 0} \bigg\{\lambda u(R)  + \sum_{i=1}^{N}
   \left\{- \mathrm{e}^{-\rho \tau_i} \left(  F^{*}(\lambda \mathrm{e}^{(\rho - r)\tau_i}) \right) \right\} \\    &\quad\quad\quad\quad\quad\quad\quad\quad\quad\quad\quad\quad-  \mathrm{e}^{-\rho T} F^{*}(\lambda \mathrm{e}^{(\rho - r)T}) + \int_{0}^{T}\rho \mathrm{e}^{-\rho t}G^{*}(\delta  \lambda \mathrm{e}^{(\rho-r)t} ) dt \bigg\}.
    \end{split}
\end{equation*}
\medskip 
Finally, when $\delta\gamma = 1$, $\frac{\partial - \mathrm{e}^{-\rho \tau_i} F^{*,k}(\lambda \mathrm{e}^{(\rho-r)\tau_i})}{ \partial \tau_i} = 0$, we could get 
\begin{equation*}
   v^{FB,P}(u(R)) =  \inf_{\lambda \leq 0} \left\{\lambda u(R) -  (N+1)\mathrm{e}^{-\rho T} F^{*}(\lambda \mathrm{e}^{(\rho - r)T}) + \int_{0}^{T}\rho \mathrm{e}^{-\rho t}G^{*}(\delta  \lambda \mathrm{e}^{(\rho-r)t} ) dt \right\}.
\end{equation*}
The proof is complete.
\qed
\subsection{Proof of Theorem \ref{thm::FB}.4}
After establishing both upper and lower bounds on $v^{FB,P}(u(R))$, we proceed to show that the principal's value is uniformly bounded when the agent has the lowest requirement for precommitment to work,    namely
$$
0<v^{FB,P}(u(0))<\rho \bar{a}(1 - \mathrm{e}^{-\rho T}).
$$
\begin{proof} For brevity, we demonstrate the argument only in the case $\delta \gamma<1$; the other cases follow by a similar reasoning.
When $\delta \gamma<1$, Theorem \ref{thm::FB}.2 concludes that
\begin{small}
  \begin{equation*}
  \overline{v}^{FB,P}(u(R)) 
            =  \inf_{\lambda \leq 0} \left\{\lambda u(R) -  N\mathrm{e}^{-\rho T} F^{*,k}(\lambda \mathrm{e}^{(\rho - r)T}) - \mathrm{e}^{-\rho T} F^{*,0}(\lambda \mathrm{e}^{(\rho - r)T}) + \int_{0}^{T}\rho \mathrm{e}^{-\rho t}G^{*}(\delta  \lambda \mathrm{e}^{(\rho-r)t} ) \mathrm{d}t \right\}   .
\end{equation*}  
\end{small}
Next, we define the auxiliary function
\begin{equation*}
    FG(\lambda):=   - N\mathrm{e}^{-\rho T} F^{*,k}(\lambda \mathrm{e}^{(\rho - r)T}) - \mathrm{e}^{-\rho T} F^{*,0}(\lambda \mathrm{e}^{(\rho - r)T}) + \int_{0}^{T}\rho \mathrm{e}^{-\rho t}G^{*}(\delta  \lambda \mathrm{e}^{(\rho-r)t} ) \mathrm{d}t .
\end{equation*}
Due to the definition (\ref{defG*F*}), we know that $FG(0) =  \bar{a}(1 - e^{-\rho T})$,  $ - F^{*,k}(\cdot)$ is a convex decreasing function, and $- F^{*,k}(0) = 0$. Additionally, $G^{*}(\cdot)$ is a convex increasing function, $G^{*}(0) = \bar{a}$, and we have
\begin{align*}
    \quad{FG}^{\prime}(0)
    &= - N\mathrm{e}^{-r T} {F^{*,k}}'(\lambda \mathrm{e}^{(\rho - r)T}) - \mathrm{e}^{-r T} {F^{*,0}}'(\lambda \mathrm{e}^{(\rho - r)T}) + \int_{0}^{T} r \mathrm{e}^{-r t}{G^{*}}'(\delta  \lambda \mathrm{e}^{(\rho-r)t} ) \mathrm{d}t \\
   &=\int_{0}^{T} r \mathrm{e}^{-r t}{G^{*}}^{\prime}(0 ) dt=\int_{0}^{T} r \mathrm{e}^{-r t} h(\bar{a}) dt =h(\bar{a})  (1  -  e^{-rT}) > 0,
\end{align*}
where the third equality comes from the envelope theorem. Since $\lim _{\lambda \rightarrow-\infty} G F(\lambda)=\infty$ with polynomial growth of order comparable to $(-\lambda)^{\frac{\gamma}{\gamma-1}}$, together with the convexity and nonnegativity of $G F(\cdot)$. Hence, we conclude that $ 
 0 < V^{P,FB}(u(0)) < \bar{a}(1 - e^{-\rho T}) $. 
\end{proof}

\section{Proof of Theorem \ref{hjbvi:pp:in}}\label{appendix::SB}

To prove Theorem \ref{hjbvi:pp:in}, we proceed as follows. First, we construct an auxiliary function that serves as a smooth supersolution for the value function. Next, we address the well-posedness of the baseline problem, $V^{p,0}$, corresponding to the case where the principal has no right to offer bonus payments. After showing the well-posedness of $V^{p,0}$, we demonstrate by induction the well-posedness of $V^{p,n}$ for $n \geq 1$.

We now present a key auxiliary result. Let $\Delta:=(b,c,M) \in \RR^{+,3}$,  $a \in \RR^+$, $\gamma > 2$, and $n\in \{0,\ldots,N\}$. We introduce the function $\varphi^{A^{(n,d)},d,\Delta} : [0,T]\times [0,\infty)\mapsto \RR$ defined by
\begin{equation}
\varphi^{A^{(n,d)},d,\Delta}(t,y) := - A^{(n,d)}(t) \mathrm{e}^{d(t -T)} y^{\gamma} + Me^{(T-t)b}(1 - e^{-cy}),
\end{equation}
where
\begin{equation}
    A^{(n,d)}(t) = \begin{cases}
    A^{n} :=\left(\frac{(1+k)^{1 /(\gamma-1)}}{n+(1+k)^{1 /(\gamma-1)}}\right)^{\gamma-1}, &    \text{if $d = \rho - r \gamma>0$},\\
    A^{n}(t) := \left(\frac{(1+k)^{1 /(\gamma-1)}}{n \mathrm{e}^{\frac{d(t-T)}{\gamma - 1}}+(1+k)^{1 /(\gamma-1)}}\right)^{\gamma-1}, &  \text{if $d = \rho - r \gamma \leq 0$}.
    \end{cases}
\end{equation}
\begin{lemma}\label{supersolconstruction}
    For all $\gamma > 2$, $ N \in \mathbb{N}_{+}$, there exists $\Delta_N\in \RR^{+,3}$ such that the following inequalities hold for all $n\in \{0,\ldots,N\}$
    \begin{equation}\label{phi0pde}
    \begin{split}
     -\varphi^{A^{(n,d)},d,\Delta_N}_t + \rho \varphi^{A^{(n,d)},d,\Delta_N} - r\varphi^{A^{(n,d)},d,\Delta_N}_y y\quad\quad\quad\quad\quad\quad\quad\quad\quad\quad\quad\quad \\
     - \sup_{z \in \mathbb{R} }\sup_{\hat{a} \in \hat{\mathcal{A}}(z)}\{rh(\hat{a})\varphi^{A^{(n,d)},d,\Delta_N}_y + \rho \hat{a} + \frac{1}{2}\varphi^{A^{(n,d)},d,\Delta_N}_{yy}r^{2}\sigma^{2}z^{2}  \} &\geq 0,\\  
    \varphi^{A^{(n,d)},d,\Delta_N}(t,0) &\geq 0,\\
    \varphi^{A^{(n,d)},d,\Delta_N}(T,y) &\geq  - A^{n}  y^{\gamma}.
      \end{split}
    \end{equation}  
\end{lemma}

\begin{proof}
   As a preliminary remark, it is evident that for all $ \Delta \in \RR^{+,3}$, the following inequality holds regardless of the sign of $\rho - \gamma r$ 
    $$\varphi^{A^{(N,d)},d,\Delta}(t,0) \geq 0 \; , \;
    \varphi^{A^{(N,d)},d,\Delta}(T,y) \geq g^{n}(y) \;,\; n \in \{0,1,\cdots,N\}. $$
    Next, we justify the existence of $\Delta^{N}$ that ensures the first inequality in \eqref{phi0pde}, depending on the sign of $\rho - \gamma r$.

    \medskip
    Firstly, we discuss the case when $\rho - \gamma r > 0$. 
    To guarantee that the first inequality holds, it is sufficient to show that there exists a $\Delta^{N} \in \RR^{+,3}$ such that
    \begin{align}\label{d>0superfirst1}
    -\varphi^{A^{(N,d)},d,\Delta_N}_t + \rho \varphi^{A^{(N,d)},d,\Delta_N} - r\varphi^{A^{(N,d)},d,\Delta_N}_y y &\geq 0\\  \label{d>0superfirst2}
    - \sup_{z \in \mathbb{R} }\sup_{\hat{a} \in \hat{\mathcal{A}}(z)}\{rh(\hat{a})\varphi^{A^{(N,d)},d,\Delta_N}_y + \rho \hat{a} + \frac{1}{2}\varphi^{A^{(N,d)},d,\Delta_N}_{yy}r^{2}\sigma^{2}z^{2}  \} &\geq 0.
\end{align}
Next, we show the inequality (\ref{d>0superfirst1}) holds if
\begin{equation}\label{setbcdphi1}
  \Delta \in \Delta_1 := \left\{  b \geq rc, c > 0, M> 0 \right\}.
\end{equation} 
If $\Delta \in \Delta_1$, we have
\begin{align*}
    &\quad -\varphi^{A^{(N,d)},d,\Delta}_t + \rho \varphi^{A^{(N,d)},d,\Delta}  - r\varphi^{A^{(N,d)},d,\Delta}_y y  \\
    &=  A^{N}d e^{d(t -T)} y^{\gamma} + M b e^{(T-t)b}(1 - e^{-cy}) + \rho(- A^{N}e^{d(t -T)} y^{\gamma} + Me^{(T-t)b}(1 - e^{-cy}))   \\
    &\quad - r (- A^{N}e^{d(t -T) } \gamma y^{\gamma - 1} +  c Me^{(T-t)b} e^{-cy} ) y\\
    &=  A^{N}e^{d(t -T)} y^{\gamma}(d - \rho + r \gamma )  + M e^{(T-t)b} \left( ( b + \rho )(1 - e^{-cy}) - r c y e^{-cy}   \right).
  \end{align*}
Due to the choice of $d = \rho - r \gamma $, $A > 0$, and $b \geq rc$, we conclude that
\begin{equation*}
    A^{N}e^{d(t -T)} y^{\gamma}(d - \rho + r \gamma )  + M e^{(T-t)b} \left( ( b + \rho )(1 - e^{-cy}) - r c y e^{-cy}   \right) \geq 0.
\end{equation*}
Then, we show that if  
\small{
\begin{align}\label{setbcdphi2}
  \Delta \in \Delta_2 &:= \left\{  b \geq rc, c > \frac{rh(\bar{a}) + \sqrt{r^{2}h(\bar{a})^{2} + 2 r^{2}\sigma^{2}\beta^{2}}}{r^{2}\sigma^{2}\beta^{2}},\right.\\
  &\quad\quad\quad\quad\quad\left. M> \max\left\{\rho \bar{a} \mathrm{e}^{c\underline{y}} , \left(\frac{1}{2}\frac{\rho^{2} \mathrm{e}^{dT}}{r^{2}\sigma^{2}\alpha^{2} A^{N} \gamma(\gamma - 1) \underline{y}^{\gamma-2}} - \frac{\rho \beta}{\alpha} \right)\mathrm{e}^{r\bar{y}}\right\} \right\}\notag
\end{align}
}
where $\alpha = h''(0) , \beta = h'(0)$, the inequality (\ref{d>0superfirst2}) holds. In particular, we have
$$- \sup_{z \in \mathbb{R} }\sup_{\hat{a} \in \hat{\mathcal{A}}(z)}\{rh(\hat{a})\varphi^{A^{N},d,\Delta}_y + \rho \hat{a} + \frac{1}{2}\varphi^{A^{N},d,\Delta}_{yy}r^{2}\sigma^{2}z^{2}  \} \geq 0. $$
Due to the compactness of the feasible effort set, we know that
\begin{equation*}
\begin{split}
     \mathcal{J}(p,q) &= \sup_{z \in \mathbb{R} }\sup_{\hat{a} \in \hat{\mathcal{A}}(z)}\{rh(\hat{a})p + \rho \hat{a} + \frac{1}{2}qr^{2}\sigma^{2}z^{2} \}  \\
     &= \sup_{z \geq 0 }\sup_{\hat{a} \in \hat{\mathcal{A}}(z)}\{rh(\hat{a})p + \rho \hat{a} + \frac{1}{2}qr^{2}\sigma^{2}z^{2} \}  \\ 
     &=\max\left\{\sup_{z \geq h'(0) }\sup_{\hat{a} \in \hat{\mathcal{A}}(z)}\{rh(\hat{a})p + \rho \hat{a} + \frac{1}{2}qr^{2}\sigma^{2}z^{2} \},0 \right\} := \mathcal{J}^{+}(p,q).
\end{split}
\end{equation*}
Thus, to obtain 
$$- \sup_{z \in \mathbb{R} }\sup_{\hat{a} \in \hat{\mathcal{A}}(z)}\{rh(\hat{a})\varphi^{A^{(N,d)},d,\Delta}_y + \rho \hat{a} + \frac{1}{2}\varphi^{A^{(N,d)},d,\Delta}_{yy}r^{2}\sigma^{2}z^{2}  \}  = \mathcal{J}^{+}(\varphi^{A^{(N,d)},d,\Delta}_y,\varphi^{A^{(N,d)},d,\Delta}_{yy})  \geq 0  ,\Delta \in \Delta_2, $$
it is sufficient to show that  for all $ \Delta \in \Delta_2$, we have
$$\mathcal{J}^{+}(\varphi^{A^{(N,d)},d,\Delta}_y,\varphi^{A^{(N,d)},d,\Delta}_{yy}) = 0 \Longleftrightarrow \sup_{z \geq h'(0) }\sup_{\hat{a} \in \hat{\mathcal{A}}(z)}\{rh(\hat{a})\varphi^{A^{(N,d)},d,\Delta}_y + \rho \hat{a} + \frac{1}{2} \varphi^{A^{(N,d)},d,\Delta}_{yy} r^{2}\sigma^{2}z^{2} \} \leq 0 .$$
Next, we define the following function
\begin{equation}\label{JZ}
    \mathcal{JZ}(z,p,q) := \sup_{\hat{a} \in \hat{\mathcal{A}}(z)}\{rh(\hat{a})p + \rho \hat{a} + \frac{1}{2}qr^{2}\sigma^{2}z^{2} \} \; , \; z \geq h'(0).
\end{equation}
Due to the boundedness of $a \in [0,\bar{a}]$, we know that when $z \geq h'(\bar{a})$, $\mathcal{JZ}(z,p,q)$ is a decreasing function with respect $z$ when $q < 0$. When $h'(0) = \beta \leq z \leq h'(\bar{a}) $, we have
\begin{align*}
    &\quad \mathcal{JZ}(z,\varphi^{A^{(N,d)},d,\Delta}_y,\varphi^{A^{(N,d)},d,\Delta}_{yy}) \\
     &= \sup_{\hat{a} \in \hat{\mathcal{A}}(z)}\bigg\{rh(\hat{a}) \varphi^{A^{(N,d)},d,\Delta}_y + \rho \hat{a} + \frac{1}{2}\varphi^{A^{(N,d)},d,\Delta}_{yy} r^{2}\sigma^{2}z^{2} \bigg\}  \\ 
     &= \sup_{\hat{a} \in \hat{\mathcal{A}}(z)}\bigg\{rh(\hat{a}) \left(  - A^{N} e^{d(t -T)} \gamma y^{\gamma - 1}  + c Me^{(T-t)b} e^{-cy}  \right) + \rho \hat{a} \\
     &\quad\quad\quad\quad\quad\quad\quad\quad\quad\quad\quad- \frac{1}{2} r^{2}\sigma^{2}z^{2}  \left(  A^{N} e^{d(t -T )} \gamma (\gamma - 1) y^{\gamma - 2} + c^{2} M e^{(T-t)b} e^{-cy}   \right)  \bigg\}   \\
     &\leq \sup_{\hat{a} \in [0,\bar{a}]} \bigg\{  r h(\bar{a}) c M e^{(T-t)b}  e^{-cy}  + \rho \hat{a} - \frac{1}{2} r^{2}\sigma^{2} {{h'}(\hat{a})}^{2}  \left(  A^{N} e^{d(t -T )} \gamma (\gamma - 1) y^{\gamma - 2} + c^{2} M e^{(T-t)b} e^{-cy}   \right) \bigg\}\\
     &\leq M e^{(T-t)b}  e^{-cy} \left(  r  h(\bar{a})  c  -   \frac{1}{2}r^{2}\sigma^{2} {h'(0)}^{2}c^2 \right)  + \sup_{\hat{a} \in [0,\bar{a}]} \{\rho \hat{a} - \frac{1}{2} r^{2}\sigma^{2} {(\alpha \hat{a} + \beta)}^{2}   A^{N} \mathrm{e}^{-dT} \gamma (\gamma - 1) y^{\gamma - 2} \}\\
     &\leq M e^{-cy} \left(  r  h(\bar{a})  c  -   \frac{1}{2}r^{2}\sigma^{2} {h'(0)}^{2}c^2 \right)  + \sup_{\hat{a} \in [0,\bar{a}]} \{\rho \hat{a} - \frac{1}{2} r^{2}\sigma^{2} {(\alpha \hat{a} + \beta)}^{2}   A^{N} \mathrm{e}^{-dT} \gamma (\gamma - 1) y^{\gamma - 2} \},
\end{align*}
where $\alpha = h''(0) , \beta = h'(0)$.
To show the inequality 
\begin{equation*}
    M e^{-cy} \left(  r  h(\bar{a})  c  -   \frac{1}{2}r^{2}\sigma^{2} {h'(0)}^{2}c^2 \right)  + \sup_{\hat{a} \in [0,\bar{a}]} \{\rho \hat{a} - \frac{1}{2} r^{2}\sigma^{2} {(\alpha \hat{a} + \beta)}^{2}   A^{N} \mathrm{e}^{-dT} \gamma (\gamma - 1) y^{\gamma - 2} \} \leq 0,
\end{equation*}
it is enough to justify
\begin{equation*}
     \sup_{\hat{a} \in [0,\bar{a}]} \{\rho \hat{a} - \frac{1}{2} r^{2}\sigma^{2} {(\alpha \hat{a} + \beta)}^{2}   A^{N} \mathrm{e}^{-dT} \gamma (\gamma - 1) y^{\gamma - 2} \} \mathrm{e}^{cy} \leq M \; , \; y \in [0,\infty) .
\end{equation*}
Using that the boundedness of $\hat{a} \in [0,\bar{a}]$,
we choose $\underline{y}$ ,  $\bar{y}$ where $ 0 < \underline{y}  \leq \bar{y}$ to make the following equalities hold
\begin{equation*}
        \frac{1}{\alpha}\left( \frac{\rho \mathrm{e}^{dT}}{r^{2}\sigma^{2}\alpha A^{N} \gamma(\gamma - 1) \underline{y}^{\gamma-2}} - \beta \right) = \bar{a}\;,\; 
        \frac{1}{\alpha}\left( \frac{\rho \mathrm{e}^{dT}}{r^{2}\sigma^{2}\alpha A^{N} \gamma(\gamma - 1) \bar{y}^{\gamma-2}} - \beta \right) = 0.
\end{equation*}
Then, we separate three situations for discussion:\\
If $\frac{1}{\alpha}\left( \frac{\rho \mathrm{e}^{dT}}{r^{2}\sigma^{2}\alpha A^{N} \gamma(\gamma - 1) y^{\gamma-2}} - \beta \right) > \bar{a}$, we have
\begin{equation*}
     \sup_{\hat{a} \in [0,\bar{a}]} \{\rho \hat{a} - \frac{1}{2} r^{2}\sigma^{2} {(\alpha \hat{a} + \beta)}^{2}   A^{N} \mathrm{e}^{-dT} \gamma (\gamma - 1) y^{\gamma - 2} \} \mathrm{e}^{cy} 
    \leq \rho \bar{a} \mathrm{e}^{c\underline{y}} \leq M.
\end{equation*}
If $\frac{1}{\alpha}\left( \frac{\rho \mathrm{e}^{dT}}{r^{2}\sigma^{2}\alpha A^{N} \gamma(\gamma - 1) y^{\gamma-2}} - \beta \right) \leq 0$, the inequality is automatically satisfied.\\
If $  0 < \frac{1}{\alpha}\left( \frac{\rho \mathrm{e}^{dT}}{r^{2}\sigma^{2}\alpha A^{N} \gamma(\gamma - 1) y^{\gamma-2}} - \beta \right) < \bar{a}$, which means $ y \in (\underline{y},\bar{y}) $. Due to the choice of $M$ in $\Delta_2$ defined in (\ref{setbcdphi2}), we have
\begin{equation*}
\begin{split}
    &\quad \sup_{\hat{a} \in [0,\bar{a}]} \{\rho \hat{a} - \frac{1}{2} r^{2}\sigma^{2} {(\alpha \hat{a} + \beta)}^{2}   A^{N} \mathrm{e}^{-dT} \gamma (\gamma - 1) y^{\gamma - 2} \} \mathrm{e}^{cy}\\ 
    &\leq \frac{\rho}{\alpha}\left( \frac{\rho \mathrm{e}^{dT}}{r^{2}\sigma^{2}\alpha A^{N} \gamma(\gamma - 1) y^{\gamma-2}} - \beta \right) - \frac{1}{2} r^{2}\sigma^{2} A^{N} \mathrm{e}^{-dT} \gamma (\gamma - 1) y^{\gamma - 2} \left(\frac{\rho \mathrm{e}^{dT}}{r^{2}\sigma^{2}\alpha A^{N} \gamma(\gamma - 1) y^{\gamma-2}} \right)^{2} \\
    &\leq \left(\frac{1}{2}\frac{\rho^{2} \mathrm{e}^{dT}}{r^{2}\sigma^{2}\alpha^{2} A^{N} \gamma(\gamma - 1) \underline{y}^{\gamma-2}} - \frac{\rho \beta}{\alpha} \right)\mathrm{e}^{r\bar{y}}\\
    &\leq M.
\end{split}
\end{equation*}
Picking $\Delta_N \in \Delta_1 \cap \Delta_2 $, we  show the existence of $\Delta_N$ making the inequalities in  (\ref{phi0pde}) hold.\\
Based on the monotonicity and strict positivity of $A^n$, namely $0<A^{n+1} \leq A^n \leq 1$ for $n \in \mathbb{N}_{+}$, we deduce that $\Delta_N$ can be chosen such that $\varphi^{A^{(n,d)},d,\Delta_N}$ satisfies the inequalities in (\ref{phi0pde}).

\medskip
Secondly, we consider the case:  $\rho - \gamma r \leq 0$. To guarantee the first inequality in (\ref{phi0pde}) holds, it is sufficient to show that there exists a $\Delta^{N} \in \RR^{+,3}$ such that
\begin{align}\label{d<0superfirst1}
    -\varphi^{A^{(N,d)},d,\Delta_N}_t + \rho \varphi^{A^{(N,d)},d,\Delta_N} - r\varphi^{A^{(N,d)},d,\Delta_N}_y y &\geq 0,\\  \label{d<0superfirst2}
    - \sup_{z \in \mathbb{R} }\sup_{\hat{a} \in \hat{\mathcal{A}}(z)}\{rh(\hat{a})\varphi^{A^{(N,d)},d,\Delta_N}_y + \rho \hat{a} + \frac{1}{2}\varphi^{A^{(N,d)},d,\Delta_N}_{yy}r^{2}\sigma^{2}z^{2}  \} &\geq 0.
\end{align}
When $\Delta \in \Delta_1$, we have
\begin{align*}
    &\quad -\varphi^{A^{(N,d)},d,\Delta}_t + \rho \varphi^{A^{(N,d)},d,\Delta}  - r\varphi^{A^{(N,d)},d,\Delta}_y y  \\
    &=  A^{N}(t)e^{d(t -T)} y^{\gamma} \left(\frac{d}{(\gamma - 1) ( N (\frac{\mathrm{e}^{d(t-T)}}{1+k})^{\frac{1}{\gamma - 1}} + 1 )} - d \right)  + M e^{(T-t)b} \left( ( b + \rho )(1 - e^{-cy}) - r c y e^{-cy} \right).
  \end{align*}
Due to  $d = \rho - r \gamma < 0 $, $\gamma > 2$,  $b \geq rc$, and $M > 0$, we have
\begin{equation*}
   -\varphi^{A^{(N,d)},d,\Delta_N}_t + \rho \varphi^{A^{(N,d)},d,\Delta_N} - r\varphi^{A^{(N,d)},d,\Delta_N}_y y \geq 0\\
\end{equation*}
Next, we move to the discussion of the inequality (\ref{d<0superfirst2}). Due to the boundness of the function $\mathrm{e}^{d(t-T)}$, the same argument used in the case $d > 0$ could be applied. Then, we redefine the set $\Delta_2$ in the following form
\small{
\begin{equation}\label{setbcdphi2d>0}
\begin{split}
  \Delta_2 &:= \left\{  b \geq rc, c > \frac{rh(\bar{a}) + \sqrt{r^{2}h(\bar{a})^{2} + 2 r^{2}\sigma^{2}\beta^{2}}}{r^{2}\sigma^{2}\beta^{2}}, \right.\\
  &\quad \quad\quad\quad\quad\quad\left. M> \max\{\rho \bar{a} \mathrm{e}^{c\underline{y}} , \left(\frac{1}{2}\frac{\rho^{2} }{r^{2}\sigma^{2}\alpha^{2} A^{N} \gamma(\gamma - 1) \underline{y}^{\gamma-2}} - \frac{\rho \beta}{\alpha} \right)\mathrm{e}^{r\bar{y}}\} \right\}. 
\end{split}
\end{equation}
}
Therefore, if $\Delta \in  \Delta_2 $,  
where $\alpha = h''(0) , \beta = h'(0)$, we have
$$- \sup_{z \in \mathbb{R} }\sup_{\hat{a} \in \hat{\mathcal{A}}(z)}\{rh(\hat{a})\varphi^{A^{(N,d)},d,\Delta}_y + \rho \hat{a} + \frac{1}{2}\varphi^{A^{(N,d)},d,\Delta}_{yy}r^{2}\sigma^{2}z^{2}  \} \geq 0 .$$
Similarly, based on the strict positivity of $A^{n}(t)$ and the strict monotonicity of $A^n(t)$ with respect to $n$, namely $0<A^{n+1}(t) \leq A^{n}(t) \leq 1$ for $n \in \mathbb{N}_{+}$, we deduce that $\Delta_N$ can be chosen such that $\varphi^{A^{(n,d)}, d, \Delta_N}$ satisfies the inequalities in (\ref{phi0pde}).
The proof is complete.
\qed   
\end{proof}

\subsection{Proof of Theorem \ref{hjbvi:pp:in} on $V^{p,0}$ }
\subsubsection{Proof of Theorem \ref{hjbvi:pp:in}.1 on $V^{p,0}$ }\label{thm2.1::vp0}
In this section, we focus on the baseline problem: $V^{p,0}(t,y)$. We introduce the open state constraint problem as follows
\begin{equation}
\begin{split}
\hat{V}^{p,0}(t,y) &= \sup_{Z \in {\hat{\mathcal{U}}}^{0}(t,y),\hat{a}\in  \hat{\mathcal{A}}(Z)}J^{p,0}(t,y;Z,\hat{a}),  \\ 
J^{p,0}(t,y;Z,\hat{a})&= \mathbb{E}\left[ \int_t^T \rho \mathrm{e}^{-\rho(s-t)}\left(\hat{a}_s\right) \mathrm{d} s-\mathrm{e}^{-\rho(T-t)}u^{-1}\left(Y_T^{t, y, Z,\hat{a} }\right) \right],
\end{split}
\end{equation}
where
$$Y_{\hat{t}}^{t,y, Z,\hat{a}} = y + \int_t^{\hat{t}} r\left(Y_s^{t,y,Z,\hat{a}}+h\left(\hat{a}_s\right)\right) \mathrm{d} s+\int_t^{\hat{t}}r Z_s \sigma \mathrm{~d} B_s \hspace{1 mm} ,\hspace{1 mm} t\leq\hat{t} \leq T,$$
and $\hat{\mathcal{U}}^{0}(t,y):= \left\{Z \in \mathcal{U}^{0}(t,y) : \hspace{1 mm} Y_s^{t,y,Z,\hat{a}}>0, \hspace{1 mm} \PP -a.s \hspace{1 mm},\hspace{1 mm} \hat{a} \in \hat{\mathcal{A}}(Z) \hspace{1 mm} , \hspace{1 mm} t \leq s \leq T \right\}$, for all $(t,y) \in [0,T] \times (0,\infty)$.

\medskip
A direct proof of the dynamic programming principle (DPP) would require the application of a measurable selection theorem, a step that becomes technically intricate in the presence of the above control constraints. To circumvent this difficulty, we adopt the weak formulation approach developed in \cite{bouchard2011weak} and \cite{bouchard2012weak}. By reformulating the problem in the sense of Bolza and Lagrange, this method allows us to establish a weak version of the DPP without resorting to measurable selection arguments. Specifically, for any stopping time $\tau \in \mathcal{T}^{t}_{t,T}$, we obtain
\begin{enumerate}
    \item Consider a lower semicontinuous function $\psi:[0, T] \times (0,\infty) \rightarrow \mathbb{R}$, such that $\hat{V}^{p,0} \leq \psi$, and $\mathbb{E}\left[\psi\left(\tau, Y_\tau^{t, y, Z}\right)^{-}\right]<\infty$, for all $(t, y) \in[0, T] \times (0,\infty)$. Then
        \begin{equation}\label{eq::principalvopen::weakdpeq1}
            \hat{V}^{p,0}(t,y) \leq \sup _{Z \in \hat{\mathcal{U}}^{0}(t,y) , \hat{a}\in  \hat{\mathcal{A}}(Z)} \mathbb{E} \left[ \int_{t}^{\tau}\rho  \mathrm{e}^{-\rho s} \hat{a}_s  ds +    \mathrm{e}^{-\rho (\tau - t)}\psi\left(\tau, Y_\tau^{t, y, Z}\right)\right]. 
        \end{equation}
    \item Let $\psi:[0, T] \times (0,\infty) \rightarrow \mathbb{R}$ be a continuous function such that $\hat{V}^{p,0} \geq \psi$. Then\\  $\mathbb{E}\left[\psi\left(\tau, Y_\tau^{t, y, Z}\right)^{+}\right]<\infty$, for every $Z \in \hat{\mathcal{U}}^{0}(t,y)$, and
   \begin{equation}\label{eq::principalvopen::weakdpeq2}
        \hat{V}^{p,0}(t, y) \geq \sup _{Z \in \hat{\mathcal{U}}^{0}(t,y), \hat{a}\in  \hat{\mathcal{A}}(Z)} E\left[ \int_{t}^{\tau}\rho  \mathrm{e}^{-\rho s} \hat{a}_s  ds + \mathrm{e}^{-\rho (\tau - t)}\psi\left(\tau, Y_\tau^{t, y, Z}\right)\right].
   \end{equation}
\end{enumerate}
Next, we show that the value function $\hat{V}^{p,0}$ is a viscosity solution to the associated Hamilton-Jacobi-Bellman variational inequalities (\ref{hjbvi:pp:in}). For all $n\in \{0,\ldots,N\}$, we introduce 
\begin{equation}
     {{\hat{V}}}^{{p,n}^*}(t, y):=  \limsup _{\substack{\left(t^{\prime}, y^{\prime}\right) \rightarrow(t, y) \\ 
\left(t^{\prime}, y^{\prime}\right) \in[0, T) \times (0,\infty)}} {\hat{V}}^{p,n} \left(t^{\prime}, y^{\prime}\right), \;\;{\hat{V}^{p,n}}_{*}(t, y):= \liminf _{\substack{\left(t^{\prime}, y^{\prime}\right) \rightarrow(t, y) \\
\left(t^{\prime}, y^{\prime}\right) \in[0, T) \times (0,\infty)}} {\hat{V}}^{p,n}\left(t^{\prime}, y^{\prime}\right).
\end{equation}
\begin{lemma}\label{viscositycharacterizationv0}
  The value function ${\hat{V}^{p,0}}_{*}$ ($ {{\hat{V}}}^{{p,0}^*}$) is a lower semicontinuous (upper semicontinuous) viscosity supersolution (viscosity subsolution) of the following Hamilton-Jacobi-Bellman variational inequality on the domain $[0,T) \times (0,\infty)$
  \begin{equation}\label{vi::principalv0open}
     \min\{ -v^{p,0}_t + \mathcal{G}v^{p,0}  , - v^{p,0}_{yy}\} = 0.
  \end{equation}
  For the terminal condition, ${\hat{V}^{p,0}}_{*}(T,\cdot)$ ($ {{\hat{V}}}^{{p,0}^*}(T,\cdot)$) is a lower semicontinuous (upper semicontinuous) viscosity supersolution (viscosity subsolution) of the variational inequality on the domain $(0,\infty)$
  \begin{equation}\label{vi::principalv0openterminal}
    \min\{v^{p,0}(T-,\cdot) - F , - v^{p,0}_{yy}(T-,\cdot) \} = 0.
  \end{equation}
\end{lemma}
\begin{proof}

\medskip
Firstly, we show the ${\hat{V}^{p,0}}_{*}(t,y)$ is the viscosity supersolution to variational inequality (\ref{vi::principalv0open}). Let $(t_0,y_0) \in [0,T)\times (0,\infty)$.
We pick a smooth function $\varphi:[0,T) \times (0,\infty) \to \mathbb{R}$ so that
$$0 = ({\hat{V}^{p,0}}_{*} - \varphi) (t_0,y_0) <  ({\hat{V}^{p,0}}_{*} - \varphi)(t,y),\quad \forall (t,y) \in [0,T) \times (0,\infty),\quad (t,y)\neq (t_0,y_0).$$
We aim to show
$$\min\{ -\varphi_t(t_0,y_0) + \mathcal{G}\varphi(t_0,y_0) , - \varphi_{yy}(t_0,y_0)\}  \geq  0.$$
We prove the above inequality through contradiction. We assume that
at $(t_0,y_0)$
$$\min\{ -\varphi_t(t_0,y_0) + \mathcal{G}\varphi(t_0,y_0) , - \varphi_{yy}(t_0,y_0)\}  < 0.$$
Moreover, for the variational inequality above, there are three possible situations:
\begin{enumerate}
    \item $-\varphi_t(t_0,y_0) + \mathcal{G}\varphi(t_0,y_0) < 0\;,\;  - \varphi_{yy}(t_0,y_0)  < 0,$
    \item $-\varphi_t(t_0,y_0) + \mathcal{G}\varphi(t_0,y_0) <  0  \;,\; - \varphi_{yy}(t_0,y_0)  \geq  0,$
    \item $-\varphi_t(t_0,y_0) + \mathcal{G}\varphi(t_0,y_0) \geq  0  \;,\; - \varphi_{yy}(t_0,y_0)  < 0.$ \label{case3::exludes}
\end{enumerate}
From the definition of $\mathcal{G}$ in \eqref{operatorG}, we can exclude the subcase (\ref{case3::exludes}) above. \\Next, we introduce the function
$$\bar{\varphi}(t,y) :=\varphi(t,y) - (|t - t_0|^{2} + |y - y_0|^{4}).$$
By the exclusion of the subcase (\ref{case3::exludes}), combined with $\left(\varphi, \varphi_t, \varphi_y,  \varphi_{yy}\right)(t_0, y_0)=\left(\bar{\varphi},\bar{\varphi}_t, \bar{\varphi}_y,  \bar{\varphi}_{yy}\right)(t_0, y_0)$, there exists a constant control $\tilde{z} > h'(0) $, and $r > 0$, $r + t_0 < T$, such that
$$-\bar{\varphi}_t(t,y) + \mathcal{G}^{\tilde{z}}\bar{\varphi}(t,y) < 0 \;,\; \forall (t,y) \in B^{r}(t_0,y_0),$$
where $B^{r}\left(t_0, y_0\right)$ denotes the ball of radius $r$ and center $\left(t_0, y_0\right)$. Given this particular choice of $r$, we have
\begin{equation}\label{v0barvarphivarphi}
   \bar{\varphi}(t,y) \leq \varphi(t,y) -  2\eta, \;\text{for some}\; \eta > 0, (t,y) \in ( [0,T] \times (0,\infty)  ) \setminus B^{r}(t_0,y_0) . 
\end{equation}
Hence, rescaling the previous inequality, we obtain
\begin{equation}\label{v0barvarphivarphi.1}
   \bar{\varphi}(t,y) \leq \varphi(t,y) -  2\mathrm{e}^{\rho(t_0 + r)}\tilde{\eta},\;\tilde{\eta} > 0, (t,y) \in ( [0,T] \times (0,\infty)  ) \setminus B^{r}(t_0,y_0)  .
\end{equation}
Let $(t_m,y_m)_{m\geq 1}$ be a sequence in $B^{r}(t_0,y_0)$, such that 
$$(t_m,y_m) \to (t_0,y_0), \quad  \quad \hat{V}^{p,0}(t_m,y_m) \to {\hat{V}^{p,0}}_{*}(t_0,y_0). $$
For $m$ big enough, the following inequality holds
\begin{equation*}
  \hat{V}^{p,0}(t_m,y_m) \leq  \bar{\varphi}(t_m,y_m) + \tilde{\eta}  .
\end{equation*}
Then, we define the sequence of the stopping times $(\tau_{m})_{m\geq 1}$ by
$$\tau_m := \inf\{s \geq t_m : (s, Y^{t_m,y_m,\tilde{z}}_s) \notin B^{r}(t_0,y_0) \}.$$ Applying Ito's lemma to $\bar{\varphi}$, we have 
\begin{equation*}
\begin{split}
        &\bar{\varphi}(t_m,y_m)\\
      =& \mathbb{E}\left[ \mathrm{e}^{-\rho (\tau_m - t_m)}\bar{\varphi}(\tau_m, Y^{t_m,y_m,\tilde{z}}_{\tau_m}) +  \int_{t_m}^{\tau_m}\mathrm{e}^{-\rho (s - t_m)}( -\bar{\varphi}_t(s,Y^{t_m,y_m,\tilde{z}}_{\tau_s}) + \mathcal{L}^{\tilde{z}}\bar{\varphi}(s,Y^{t_m,y_m,\tilde{z}}_{\tau_s}) )  ds  \right] \\
     \leq&  \mathbb{E}\left[  \mathrm{e}^{-\rho (\tau_m - t_m)}\bar{\varphi}(\tau_m, Y^{t_m,y_m,\tilde{z}}_{\tau_m}) +\int_{t_m}^{\tau_m} \rho \hat{a}_s   \mathrm{e}^{-\rho (s-t_m)} ds   \right],\\ 
\end{split}
\end{equation*}
where
\begin{equation}\label{Loperatorz}
\mathcal{L}^{\tilde{z}}\bar{\varphi}(t,y): = \rho \bar{\varphi} - \sup_{\hat{a} \in \hat{\mathcal{A}}(\tilde{z})}\{rh(\hat{a})\bar{\varphi}_y + \frac{1}{2}\bar{\varphi}_{yy}r^{2}\sigma^{2}{\tilde{z}}^{2}  \} - r \bar{\varphi}_y y .
\end{equation}
As a result, we have 
\begin{equation*}
    \begin{split}
        \hat{V}^{p,0}(t_m,y_m) &\leq \bar{\varphi}(t_m,y_m) + \tilde{\eta
        }\\
        &  \leq  \mathbb{E}\left[ \mathrm{e}^{-\rho (\tau_m - t_m)}\bar{\varphi}(\tau_m, Y^{t_m,y_m,\tilde{z}}_{\tau_m}) + \int_{t_m}^{\tau_m} \rho \hat{a}_s   \mathrm{e}^{-\rho (s-t_m)} ds \right] + \tilde{\eta
        } \\
         & \leq  \mathbb{E}\left[ \mathrm{e}^{-\rho (\tau_m - t_m)}\varphi(\tau_m, Y^{t_m,y_m,\tilde{z}}_{\tau_m}) +\int_{t_m}^{\tau_m} \rho \hat{a}_s   \mathrm{e}^{-\rho (s-t_m)} ds  \right] - \tilde{\eta
        },\\
    \end{split}
\end{equation*}
which contradicts the WDPP inequality (\ref{eq::principalvopen::weakdpeq2}).

\medskip
Secondly, we show that ${{\hat{V}}}^{{p,0}^*}$ is a viscosity subsolution to the variational inequality (\ref{vi::principalv0open}). We proceed by contradiction. Assume there exist $\left(t_0, y_0\right) \in[0, T) \times (0,\infty)$ and a smooth test function $\varphi:[0, T) \times (0,\infty) \to \mathbb{R}$ satisfying
$$0 = ({{\hat{V}}}^{{p,0}^*} - \varphi) (t_0,y_0) >  ( {{\hat{V}}}^{{p,0}^*} - \varphi)(t,y) \;,\; \forall (t,y) \in [0,T) \times (0,\infty) , (t,y)\neq (t_0,y_0),$$
and
$$\min\{ -\varphi_t(t_0,y_0) + \mathcal{G}\varphi(t_0,y_0) , - \varphi_{yy}(t_0,y_0)\}  > 0.$$\\
Next, we define the function
$$\bar{\varphi}(t,y) :=\varphi(t,y) + (|t - t_0|^{2} + |y - y_0|^{4}).$$
Based on $\left(\varphi, \varphi_t, \varphi_y,  \varphi_{yy}\right)(t_0, y_0)=\left(\bar{\varphi},\bar{\varphi}_t, \bar{\varphi}_y,  \bar{\varphi}_{yy}\right)(t_0, y_0)$, and the continuity of $\mathcal{G}$, we obtain
\begin{equation*}
    -\bar{\varphi}_t(t,y) + \mathcal{G}^{\tilde{z}}\bar{\varphi}(t,y) > 0,  \forall \tilde{z} \in [0,\infty) , \quad  \forall (t,y) \in B^{r_\e}(t_0,y_0).
\end{equation*}
With choice of $r$, through the similar scaling argument as inequality (\ref{v0barvarphivarphi.1}), we have
\begin{equation}\label{v0barvarphivarphisub}
   \bar{\varphi}(t,y) \geq \varphi(t,y) + 2 e^{\rho(t_0+r)}  \tilde{\eta} \quad, \tilde{\eta} > 0, (t,y) \in ( [0,T] \times (0,\infty)  ) \setminus B^{r}(t_0,y_0).  
\end{equation}
Let $(t_m,y_m)_{m\geq 1}$  be a sequence in $B^{r}(t_0,y_0)$ such that 
$$(t_m,y_m) \to (t_0,y_0) \; , \; \hat{V}^{p,0}(t_m,y_m) \to {{\hat{V}}}^{{p,0}^*}(t_0,y_0).$$
Let $(\tau_{m})_{m\geq 1}$ be a sequence of stopping times defined through picking arbitrary controls $\tilde{z}^{m} \in \hat{\mathcal{U}}^{0}(t_m,y_m)$
$$\tau_m := \inf\{s \geq t_m : (s, Y^{t_m,y_m,\tilde{z}^{m}}_s) \notin B^{r}(t_0,y_0) \}.$$ Applying Ito's lemma to $\bar{\varphi}$, we have
\begin{small}
\begin{equation*}
\begin{split}
     \bar{\varphi}(t_m,y_m) &= \mathbb{E}\left[ \mathrm{e}^{-\rho (\tau_m - t_m)}\bar{\varphi}(\tau_m, Y^{t_m,y_m,\tilde{z}^{m}}_{\tau_m}) +  \int_{t_m}^{\tau_m}\mathrm{e}^{-\rho (s - t_m)}( -\bar{\varphi}_t(s,Y^{t_m,y_m,\tilde{z}^{m}}_{\tau_s}) + \mathcal{G}^{\tilde{z}^{m}}\bar{\varphi}(s,Y^{t_m,y_m,\tilde{z}^{m}}_{\tau_s}) )  ds   \right] \\
     &\geq  \mathbb{E}\left[ \mathrm{e}^{-\rho (\tau_m - t_m)}\bar{\varphi}(\tau_m, Y^{t_m,y_m,\tilde{z}^{m}}_{\tau_m}) + \int_{t_m}^{\tau_m} \rho \mathrm{e}^{-\rho (s - t_m)} \hat{a}^{\tilde{z}^{m}}_s ds \right]. \\ 
\end{split}
\end{equation*}
\end{small}
The above inequality follows directly from
\begin{equation*}
    \begin{split}
       &\quad -\bar{\varphi}_t(t,y)  + \mathcal{G}^{\tilde{z}^{m}}\bar{\varphi}(t,y)  \\
       &=  -\bar{\varphi}_t(t,y)  +  \rho \bar{\varphi} - \sup_{\hat{a} \in \hat{\mathcal{A}}(\tilde{z}^{m})}\{rh(\hat{a})\bar{\varphi}_y + \rho \hat{a} + \frac{1}{2}\bar{\varphi}_{yy}r^{2}\sigma^{2}\tilde{z}^{{m}^{2}}  \} - r\bar{\varphi}_y y\\
       & > 0 \; ,  \; \forall (t,y) \in B^{r}(t_0,y_0).
    \end{split}
\end{equation*}
By (\ref{v0barvarphivarphisub}), we have
\begin{equation*}
    \bar{\varphi}(t_m,y_m) \geq \mathbb{E}\Bigg[  \mathrm{e}^{-\rho (\tau_m - t_m)}\varphi(\tau_m, Y^{t_m,y_m,\tilde{z}^{m}}_{\tau_m}) + \int_{t_m}^{\tau_m} \rho \mathrm{e}^{-\rho (s - t_m)} \hat{a}^{\tilde{z}^{m}}_s ds\Bigg]  + 2\tilde{\eta}.
\end{equation*}
In addition, for sufficiently large $m$,  we have $  \hat{V}^{p,0}(t_m,y_m) +  \tilde{\eta} \geq \bar{\varphi}(t_m,y_m)$. As a result, we obtain
\begin{equation*}
    \hat{V}^{p,0}(t_m,y_m) \geq \mathbb{E}\Bigg[  \mathrm{e}^{-\rho (\tau_m - t_m)}\varphi(\tau_m, Y^{t_m,y_m,\tilde{z}^{m}}_{\tau_m}) + \int_{t_m}^{\tau_m} \rho \mathrm{e}^{-\rho (s - t_m)} \hat{a}^{\tilde{z}^{m}}_s ds\Bigg]  + \tilde{\eta},
\end{equation*}
for any arbitrary admissible control $\tilde{z}^{m}$. Since $\varphi\geq \hat{V}^{p,0}$, we reach a contradiction to the inequality (\ref{eq::principalvopen::weakdpeq1}).\\
Next, we consider the viscosity characterization of
$\hat{V}^{p,0}$ at the terminal time $T$.\\
Firstly, we show that ${\hat{V}^{p,0}}_{*}$ is the viscosity supersolution to the following variational inequality:
$$ \min\{v^{p,0}(T-,y) - F(y) , - v^{p,0}_{yy}(T-,y) \} = 0.$$
Applying \cite[Lemma 4.3.2]{pham2009continuous}, we could deduce that ${\hat{V}^{p,0}}_{*}$ is the viscosity supersolution.\\
Secondly, we show that $\hat{V}^{{p,0}^{*}}(T,y)$  is a viscosity subsolution to variational inequality (\ref{vi::principalv0openterminal}). With help of inequality (\ref{eq::principalvopen::weakdpeq1}), we could deduce that \cite[Lemma 4.3.3]{pham2009continuous} is true. Through applying  \cite[Lemma 4.3.4]{pham2009continuous}, we conclude that  $\hat{V}^{{p,0}^{*}}(T,y)$  is a viscosity subsolution to variational inequality (\ref{vi::principalv0openterminal}).
The proof is complete.\qed
\end{proof}
\medskip
Following Lemma \ref{viscositycharacterizationv0} above, we have the conclusion that ${\hat{V}^{p,0}}_{*}(t,y)$ ($ {{\hat{V}}}^{{p,0}^*}(t,y)$) is a lower semicontinuous (upper semicontinuous) viscosity supersolution (viscosity subsolution) of the Hamilton-Jacobi-Bellman variational inequalities (\ref{vi::principalv0open}),(\ref{vi::principalv0openterminal}).

\medskip
Next, we focus on the derivation of the desired upper bound on $\hat{V}^{p,0}$.\\
Consider $(t,y) \in [0,T] \times (0,\infty)$, a control process $Z \in \hat{\mathcal{U}}^{0}(t,y)$, and an $\FF$-stopping time $\theta$. Applying Itô's Lemma, we obtain
\begin{align} \nonumber
         &\quad \mathrm{e}^{-\rho(\theta - t)}\varphi^{A^{(0,d)},d,\Delta_N}\left(\theta, Y_{\theta}^{t, y, Z}\right) - \varphi^{A^{(0,d)},d,\Delta_N}(t, y) + \int_{t}^{\theta}\rho  \mathrm{e}^{-\rho(s-t)}\hat{a}^{z}_s ds\\  \nonumber
         & =  - \int_{t}^{\theta} \mathrm{e}^{-\rho(s-t)} \left( -\varphi^{A^{(0,d)},d,\Delta_N}_t(s,Y^{t,y,Z}_{s}) + \mathcal{L}^{z}\varphi^{A^{(0,d)},d,\Delta_N}(s,Y^{t,y,Z}_{s}) - \rho\hat{a}_s \right) ds \\ \label{super.sol.ineq::v0}
         &\quad + \sigma\int_{t}^{\theta} \mathrm{e}^{-\rho(s-t)} \varphi_y^{A^{(0,d)},d,\Delta_N}(s,Y^{t,y,Z}_{s}) Z_s dB_s,
\end{align}
where $A^{(0,d)}(t) = 1$, $\Delta_N$ is constructed in Lemma \ref{supersolconstruction}, and the operator $\mathcal{L}^z$ is defined in  (\ref{Loperatorz}).
Then we introduce a localizing sequence of $\FF$-stopping times $(\theta_n)_{n\in \NN} $ defined by
$$\theta_n:=\inf \left\{u \geq t:\left| \mathrm{e}^{-\rho(s-t)} \varphi_y^{A^{n},d,\Delta_N}(s,Y^{t,y,Z}_{s})Z_s \right| \geq n\right\} \wedge T.$$
\medskip
Fixing $\theta := \theta_n$, and taking expectations in \eqref{super.sol.ineq::v0}, we obtain from Lemma \ref{supersolconstruction} 
\begin{equation*}
   \varphi^{A^{(0,d)},d,\Delta_N}(t, y) \geq \EE\left[ \mathrm{e}^{-\rho(\theta_n - t)}\varphi^{A^{(0,d)},d,\Delta_N}\left(\theta_n, Y_{\theta_n}^{t, y, Z}\right)  + \int_{t}^{\theta_n}\rho  \mathrm{e}^{-\rho(s-t)}\hat{a}^{z}_s ds \right].
\end{equation*}
Next, by the set of admissible controls $\hat{\mathcal{U}}^{0}(t,y)$ (\ref{condition::interability::YZ::reduction2}) and standard SDE estimates, we have
\begin{align*}
\EE\left[ \sup_{ t \leq \theta \leq T} \left| \int_{t}^{\theta}\rho  \mathrm{e}^{-\rho(\theta -t)}\hat{a}^{z}_s ds  + \varphi^{A^{(0,d)},d,\Delta_N}\left(\theta, Y_{\theta}^{t, y, Z}\right) \right|\right] < \infty.
\end{align*}
 Hence, using that $\lim_{n\rightarrow\infty} \theta_n = T, \PP-a.s.$ and the continuity of $\varphi^{A^{(0,d)},d,\Delta_N}$, we obtain by the dominated convergence theorem
\begin{equation}\label{ineq.value}
     \varphi^{A^{(0,d)},d,\Delta_N}(t, y)  \geq \EE\left[  \mathrm{e}^{-\rho(T - t)}\varphi^{A^{(0,d)},d,\Delta_N}\left(T, Y_{T}^{t, y, Z}\right) + \int_{t}^{T}\rho  \mathrm{e}^{-\rho(s-t)}\hat{a}^{Z}_s ds  \right].
\end{equation}
Taking the supremum over $Z \in \hat{\mathcal{U}}^{0}(t,y)$ in \eqref{ineq.value}, and using that $\varphi^{A^{(0,d)},d,\Delta_N}(T, y) \geq F(y)$ for all $y\geq0$, we obtain
\begin{align}\label{limit.lower.envelope}
     \varphi^{A^{(0,d)},d,\Delta_N}(t, y) \geq \hat{V}^{p,0}(t,y).
\end{align}
Moreover, setting $Z=0$, we get the following inequalities
\begin{equation}\label{upperlowerboundvp0}
   - \mathrm{e}^{(T-t)(r\gamma - \rho)}y^{\gamma}  \leq \hat{V}^{p,0}(t,y) \leq  \varphi^{A^{(0,d)},d,\Delta_N}(t, y).
\end{equation}
The previous implies
$$ \lim_{(t',y') \to (t,0),y' > 0} - \mathrm{e}^{(T-t')(r\gamma - \rho)}{y'}^{\gamma}  \leq {\hat{V}^{p,0}}_{*}(t,0) \leq {{\hat{V}}}^{{p,0}^*}(t,0) \leq  \lim_{(t',y') \to (t,0),y' > 0}  \varphi^{A^{(0,d)},d,\Delta_N}(t', y').$$
Finally, we discuss the terminal condition. From (\ref{upperlowerboundvp0}), we know that $\hat{V}^{p,0}$ has polynomial growth in $y$. Moreover, the variational inequality (\ref{vi::principalv0openterminal}) falls within the framework of \cite[Section V.1]{ishii1990viscosity}. Consequently, the comparison principle holds, and we conclude that
$$
\hat{V}^{p, 0^*}(T, y)=\hat{V}_*^{p, 0}(T, y)=-y^\gamma \;,\;  y \in (0,\infty).
$$
After discussing the behaviour of $\hat{V}^{p,0}$ at the terminal and boundary regions, we apply Lemma \ref{cmp::v0hjbvi} (comparison principle) to deduce that $\hat{V}^{p,0}$ is the unique continuous viscosity solution of the following Dirichlet problem
\begin{equation}\label{hjbvi:pp:dirichletv0}
\begin{split}
 \min\{ -v^{p,0}_t + \mathcal{G}v^{p,0}  , - v^{p,0}_{yy}(t,y)\} &= 0 \;,\;    [0,T) \times (0,\infty) ,\\
v^{p,0}(T,y) &= -y^{\gamma} \;,\; y\in [0,\infty),\\
v^{p,0}(t,0) &= 0 \;,\; t \in [0,T]. \\
\end{split}   
\end{equation}
Following a similar argument to Lemma \ref{viscositycharacterizationv0}, we conclude that ${V^{p,0}}^{*}$ is the viscosity subsolution to Dirichlet Problem (\ref{hjbvi:pp:dirichletv0}). Finally, we have
$${\hat{V}^{p,0}}_{*}(t,y) \leq {V^{p,0}_{*}}(t,y)  \leq {V}^{{p,0}^*}(t,y) \leq {\hat{V}^{p,0}}_{*}(t,y)\leq {\hat{V}}^{{p,0}^{*}}(t,y), (t,y) \in [0,T) \times (0,\infty),$$
where we used Lemma \ref{cmp::v0hjbvi} in the third inequality. In conclusion, $V^{p,0}$ is the unique continuous viscosity solution to the Dirichlet Problem (\ref{hjbvi:pp:dirichletv0}).

\subsubsection{Proof of Theorem \ref{hjbvi:pp:in}.2 on $V^{p,0}$}
Since $V^{p,0}$ is a viscosity supersolution of the Dirichlet Problem (\ref{hjbvi:pp:dirichletv0}), we must have $- V^{p,0}_{yy}(t,y) \geq 0$ for all $(t,y)\in [0,T]\times [0,\infty) $ in the viscosity sense; otherwise, $\mathcal{J}^{+}\left(V^{p,0}_{y}(t,y), V^{p,0}_{yy}(t,y)\right)=\infty$. This implies that $V^{p,0}(t,\cdot)$ is concave for all $t\in [0,T]$.

\subsection{Proof of Theorem \ref{hjbvi:pp:in} for $V^{p,n}$}
After showing the baseline case, we prove Theorem \ref{hjbvi:pp:in} by induction. In particular, our induction hypothesis assumes that Theorem \ref{hjbvi:pp:in} holds for $V^{p, n-1}$ with $n \in \{1,\ldots,N\}$, and our objective is to show that it also holds for $V^{p, n}$. We provide a detailed discussion for the case $n=1$. The case $n\geq 2$ will follow an identical procedure. We begin by introducing the open-state-constraint value function
\begin{equation}
\begin{split}
\hat{V}^{p,1}(t,y) &:= \sup_{(\nu^{1}) \in {\hat{\mathcal{U}}}^{1}(t,y),\hat{a}\in  \hat{\mathcal{A}}(\nu^{1})}J^{p,1}(t,y;\nu^{1},\hat{a}),  \\ 
J^{p,1}(t,y;\nu^{1},\hat{a})&:= \mathbb{E}\left[ \int_t^T \rho \mathrm{e}^{-\rho(s-t)}\left(\hat{a}_s\right) \mathrm{d} s-   \mathrm{e}^{-\rho (\tau_1 - t)} (1+k)\eta^{\gamma}_1-\mathrm{e}^{-\rho(T-t)}u^{-1}\left(Y_T^{t, y, \nu^{1} }\right) \right],
\end{split}
\end{equation}
where $\nu^{1} := (Z,\tau_1,\eta_1)$, and 
$$Y_{\hat{t}}^{t,y, \nu^{1},\hat{a}} = y + \int_t^{\hat{t}} r\left(Y_s^{t,y,\nu^{1},\hat{a}}+h\left(\hat{a}_s\right)\right) \mathrm{d} s+\int_t^{\hat{t}}r Z_s \sigma \mathrm{~d} B_s - \eta_1\mathbbm{1}_{[\tau_1,T]}(\hat{t})  \hspace{1 mm} ,\hspace{1 mm} t\leq\hat{t} \leq T.$$
Correspondingly, we define the continuous process of $\bar{Y}$ without impulse control
$$\bar{Y}_{\hat{t}}^{t,y,Z,\hat{a}} := y + \int_t^{\hat{t}} r\left(\bar{Y}_s^{t,y,Z,\hat{a}}+h\left(\hat{a}_s\right)\right) \mathrm{d} s+\int_t^{\hat{t}}r Z_s \sigma \mathrm{~d} B_s   \hspace{1 mm} ,\hspace{1 mm} t\leq\hat{t} \leq T.$$
And we further set $Y_{t-}:= y$,  $\hat{\mathcal{U}}^{1}(t,y):= \left\{\nu^{1} \in \mathcal{U}^{1}(t,y) : \hspace{1 mm} Y_s^{t,y,\nu^{1},\hat{a}}>0 \hspace{1 mm} \PP -a.s \hspace{1 mm},\hspace{1 mm} \hat{a} \in \hat{\mathcal{A}}(\nu^{1}) \hspace{1 mm} , \hspace{1 mm} t \leq s \leq T \right\}$, for all $(t,y) \in [0,T] \times (0,\infty)$.

\medskip
The following lemma establishes a crucial link between $\hat{V}^0$ and $\hat{V}^1$. 
\begin{lemma}\label{lm::hatVp1Vp0}The following representation holds
\begin{equation}\label{vpkpk-1}
    \hat{V}^{p,1}(t,y) = \sup_{(Z,\tau_1) \in \hat{\mathcal{U}}^{1}(t,y) , \hat{a} \in \hat{\mathcal{A}}(Z)} \mathbb{E} \left[\int_{t}^{\tau_1}\rho  \mathrm{e}^{-\rho(s-t)}\hat{a}_s ds + \mathrm{e}^{-\rho(\tau_1 - t)}\mathcal{M}\hat{V}^{p,0}(\tau_1,\bar{Y}^{t,y,Z}_{\tau_1})\right].
\end{equation}
\end{lemma}
\begin{proof}
Firstly, we show the easy direction. 
Starting from the bonus payment time valued at $\tau_1 \in [t,T)$ and applying the tower property of conditional expectation, we obtain
\begin{equation*}
\begin{split}
 J^{p,1}(t,y;\nu^{1},\hat{a})
     &=  \mathbb{E}\left[ \int_t^{T} \rho \mathrm{e}^{-\rho (s -t) }\left(\hat{a}_s \right) \mathrm{d} s   - \mathrm{e}^{-\rho (\tau_1 - t)} (1+k)\eta^{\gamma}_1   - \mathrm{e}^{\rho(T-t)}u^{-1}(Y^{t,y, Z,\tau_1,\eta_1}_{T})    \right]    \\
     &=\mathbb{E}\left[ \int_t^{\tau_1} \rho \mathrm{e}^{-\rho (s -t) }\left(\hat{a}_s \right) \mathrm{d} s   + \mathrm{e}^{-\rho (\tau_1 - t)} J^{p,0}(\tau_1,\bar{Y}^{t,y,Z,\tau_1}_{\tau_1} -\eta_1 ) -\mathrm{e}^{-\rho (\tau_1 - t)} (1+k)\eta^{\gamma}_1   \right] \\
     &\leq \mathbb{E}\left[ \int_t^{\tau_1} \rho \mathrm{e}^{-\rho (s -t) }\left(\hat{a}_s \right) \mathrm{d} s   + \mathrm{e}^{-\rho (\tau_1 - t)} \hat{V}^{p,0}(\tau_1,\bar{Y}^{t,y,Z,\tau_1}_{\tau_1} -\eta_1)  - \mathrm{e}^{-\rho (\tau_1 - t)} (1+k)\eta^{\gamma}_1   \right] \\
     &\leq \mathbb{E}\left[ \int_t^{\tau_1} \rho \mathrm{e}^{-\rho (s -t) }\left(\hat{a}_s \right) \mathrm{d} s   + \mathrm{e}^{-\rho (\tau_1 - t)} \mathcal{M}\hat{V}^{p,0}(\tau_1,\bar{Y}^{t,y,Z}_{\tau_1})  \right] \\
     &\leq \sup_{(Z,\tau_1) \in \mathcal{\hat{U}}^{1}(t,y) , \hat{a} \in \hat{\mathcal{A}}(Z) } \mathbb{E}\left[ \int_t^{\tau_1} \rho \mathrm{e}^{-\rho (s -t) }\left(\hat{a}_s \right) \mathrm{d} s   + \mathrm{e}^{-\rho (\tau_1 - t)} \mathcal{M}\hat{V}^{p,0}(\tau_1,\bar{Y}^{t,y,Z}_{\tau_1})   \right]. 
\end{split}
\end{equation*}
By taking the supremum over admissible control $(Z,\tau_1,\eta_1) \in \mathcal{\hat{U}}^{1}(t,y)$, we conclude
\begin{equation}\label{vp1Mv0<=}
  \hat{V}^{p,1}(t,y) \leq \sup_{(Z,\tau_1) \in \mathcal{\hat{U}}^{1}(t,y)  , \hat{a} \in \hat{\mathcal{A}}(Z) } \mathbb{E}\left[ \int_t^{\tau_1} \rho \mathrm{e}^{-\rho (s -t) }\left(\hat{a}_s \right) \mathrm{d} s   + \mathrm{e}^{-\rho (\tau_1 - t)} \mathcal{M}\hat{V}^{p,0}(\tau_1,\bar{Y}^{t,y,Z}_{\tau_1})   \right]  
\end{equation}
Next, we show the hard direction. 
Fix some arbitrary control $Z \in {\hat{\mathcal{U}}}^{1}(t, y)$ and $ \tau_1 \in \mathcal{T}_{[t, T)}$, $\eta_1(\omega) \in [0, Y^{t,y,Z}_{\tau_1-}(\omega)) $. By the continuity of $\hat{V}^{p,0}(t,y)$ established in \ref{thm2.1::vp0}, noting the impulse time $\tau_1$ valued in $[t,T)$ which is predictable, for any $\mathrm{e}^{\rho T} \varepsilon>0$ and $\omega \in \Omega$, 
the set of $\mathrm{e}^{\rho T} \varepsilon$-optimal control for $\hat{V}^{p,0}\left(\tau_1(\omega), Y_{\tau_1(\omega)}^{t, y,Z,\eta_1}(\omega)\right)$ defined as
$$\left\{ Z \in {\hat{\mathcal{U}}}^{0}\left(\tau_1(\omega), Y_{\tau_1(\omega)}^{t, y, Z, \eta_1}(\omega)\right):\hat{V}^{p,0}\left(\tau_1(\omega), Y_{\tau_1(\omega)}^{t, y,Z,\eta_1}(\omega)\right) -\mathrm{e}^{\rho T} \varepsilon \leq J^{p,0}\left( \tau_1(\omega), Y_{\tau_1(\omega)}^{t, y , Z ,\eta_1}(\omega); Z\right)\right\}$$
is not empty. By a measurable selection argument, there exists $Z^{\e,\omega}$ in this set so that the triplet $\left(\hat{Z}, \tau_1, \eta_1\right)$ belongs to $\hat{\mathcal{U}}^1(t, y)$, where 
$$
\hat{Z}_s(\omega)=\left\{\begin{array}{c}
Z_s(\omega)\;,\;  s \in[0, \tau_1(\omega)], \\
Z_s^{\varepsilon, \omega}(\omega)\;,\; s \in(\tau_1(\omega), T].
\end{array}\right.
$$
 Then, we have
\begin{equation*}
\begin{split}
     &\hat{V}^{p,1}(t,y)\\
     \geq& J^{p,1}(t,y;(\hat{Z},\tau_1,\eta_1) )\\
     =  &\mathbb{E}\left[ \int_t^{T} \rho \mathrm{e}^{-\rho (s -t) }\left(\hat{a}_s \right) \mathrm{d} s   - \mathrm{e}^{-\rho (\tau_1 - t)} (1+k)\eta^{\gamma}_1   - \mathrm{e}^{\rho(T-t)}u^{-1}(Y^{t,y, Z,\bar{\tau}_N,\bar{\eta}_n}_{T})    \right]    \\
     =&\mathbb{E}\left[ \int_t^{\tau_1} \rho \mathrm{e}^{-\rho (s -t) }\left(\hat{a}_s \right) \mathrm{d} s   + \mathrm{e}^{-\rho (\tau_1 - t)} J^{p,0}(\tau_1,Y^{t,y,\hat{Z},\tau_1,\eta_1}_{\tau_1}) -\mathrm{e}^{-\rho (\tau_1 - t)} (1+k)\eta^{\gamma}_1   \right] \\
     =&\mathbb{E}\left[ \int_t^{\tau_1} \rho \mathrm{e}^{-\rho (s -t) }\left(\hat{a}_s \right) \mathrm{d} s   + \mathrm{e}^{-\rho (\tau_1 - t)} J^{p,0}(\tau_1,\bar{Y}^{t,y,\hat{Z},\tau_1}_{\tau_1} - \eta_1) -\mathrm{e}^{-\rho (\tau_1 - t)} (1+k)\eta^{\gamma}_1   \right] \\
     \geq& \mathbb{E}\left[ \int_t^{\tau_1} \rho \mathrm{e}^{-\rho (s -t) }\left(\hat{a}_s \right) \mathrm{d} s   + \mathrm{e}^{-\rho (\tau_1 - t)} \hat{V}^{p,0}(\tau_1,\bar{Y}^{t,y,Z,\tau_1}_{\tau_1} - \eta_1)  - \mathrm{e}^{-\rho (\tau_1 - t)} (1+k)\eta^{\gamma}_1   \right] - \e .\\
\end{split}
\end{equation*}
From the arbitraness of $(Z,\tau_1,\eta_1) \in  \mathcal{\hat{U}}^{1}(t,y)$, we deduce
\begin{equation*}
\begin{split}
        &\hat{V}^{p,1}(t,y) \\
         \geq& \sup_{(Z,\tau_1,\eta_1) \in  \mathcal{\hat{U}}^{1}(t,y) , \hat{a} \in \hat{\mathcal{A}}(Z)}\mathbb{E}\left[ \int_t^{\tau_1} \rho \mathrm{e}^{-\rho (s -t) }\left(\hat{a}_s \right) \mathrm{d} s   + \mathrm{e}^{-\rho (\tau_1 - t)} \hat{V}^{p,0}(\tau_1,\bar{Y}^{t,y,Z,\tau_1}_{\tau_1} - \eta_1) \right. \\
         &\quad\quad \quad\quad\quad\quad\quad\quad\quad\left. - \mathrm{e}^{-\rho (\tau_1 - t)} (1+k)\eta^{\gamma}_1   \right] - \e \\
        =&\sup_{(Z,\tau_1 ) \in  \mathcal{\hat{U}}^{1}(t,y) , \hat{a} \in \hat{\mathcal{A}}(Z)}\mathbb{E}\left[ \int_t^{\tau_1} \rho \mathrm{e}^{-\rho (s -t) }\left(\hat{a}_s \right) \mathrm{d} s   + \mathrm{e}^{-\rho (\tau_1 - t)} \mathcal{M}\hat{V}^{p,0}(\tau_1,\bar{Y}^{t,y,Z}_{\tau_1})   \right] -\e.
\end{split}
\end{equation*}
As a result, we obtain the desired inequality
\begin{equation}\label{vp1Mv0>=}
  \hat{V}^{p,1}(t,y) \geq \sup_{(Z,\tau_1) \in \mathcal{\hat{U}}^{1}(t,y) , \hat{a} \in \hat{\mathcal{A}}(Z) } \mathbb{E}\left[ \int_t^{\tau_1} \rho \mathrm{e}^{-\rho (s -t) }\left(\hat{a}_s \right) \mathrm{d} s   + \mathrm{e}^{-\rho (\tau_1 - t)} \mathcal{M}\hat{V}^{p,0}(\tau_1,\bar{Y}^{t,y,Z}_{\tau_1})   \right].  
\end{equation}
Combining the inequalities (\ref{vp1Mv0<=}) and (\ref{vp1Mv0>=}), we prove the lemma. 
\qed
\end{proof}

\begin{lemma}\label{supersolvpk}
\medskip

Let $\Delta_N$ be the three-tuple introduced in Lemma \ref{supersolconstruction}. Then, for every $(t, y) \in[0, T] \times (0,\infty), \nu^{1}\in \hat{\mathcal{U}}^{1}(t, y)$, and  $\theta \in \mathcal{T}_{[t,T)}$, the random variable $\varphi^{A^{(1,d)},d,\Delta_N}\left(\theta, \bar{Y}_{\theta}^{t,y ,\nu^{1},\hat{a}}\right)$ is integrable and satisfies
$$
\mathbb{E}\left[\int_{t}^{\theta}\rho  \mathrm{e}^{-\rho(s-t)}\hat{a}_s \mathrm{d} s + \mathrm{e^{-\rho(\theta - t)}}\varphi^{A^{(1,d)},d,\Delta_N}\left(\theta, \bar{Y}_{\theta}^{t, y,\nu^{1},\hat{a}}\right)\right] \leq \varphi^{A^{(1,d)},d,\Delta_N}(t, y) .
$$

Moreover,
$$
\hat{V}^{p,1}(t, y) \leq \varphi^{A^{(1,d)},d,\Delta_N}(t, y)\; , \; \forall(t, y) \in[0, T] \times (0,\infty) .
$$
\end{lemma}

\begin{proof}
     Let $\theta$ be an $\mathbb{F}$-stopping time. Applying Itô's lemma to  $\varphi^{A^{(1,d)},d,\Delta_N}$ yields
     \begin{align}\nonumber
         &\quad \mathrm{e}^{-\rho(\theta_n - t)}\varphi^{A^{(1,d)},d,\Delta_N}\left(\theta_n, \bar{Y}_{\theta_n}^{t, y,\nu^{1},\hat{a}}\right) - \varphi^{A^{(1,d)},d,\Delta_N}(t, y) + \int_{t}^{\theta_n}\rho  \mathrm{e}^{-\rho(s-t)}\hat{a}_s \mathrm{d} s\\ \nonumber
         & =  - \int_{t}^{\theta_n} \mathrm{e}^{-\rho(s-t)} \left( -\varphi^{A^{(1,d)},d,\Delta_N}_t(s,\bar{Y}^{t, y,\nu^{1},\hat{a}}_{s}) + \mathcal{L}^{\nu}\varphi^{A^{(1,d)},d,\Delta_N}(s,\bar{Y}^{t, y,\nu^{1},\hat{a}}_{s}) - \rho\hat{a}_s \right) \mathrm{d} s \\ \label{super.sol.ineq::v1}
         &\quad + \sigma\int_{t}^{\theta_n} \mathrm{e}^{-\rho(s-t)} \varphi_y^{A^{(1,d)},d,\Delta_N}(s,\bar{Y}^{t, y,\nu^{1},\hat{a}}_{s}) Z_s dB_s,
     \end{align}
     where $\theta_n:=\inf \left\{u \geq t:\left| \mathrm{e}^{-\rho(s-t)} \varphi_y^{A^{(1,d)},d,\Delta_N}(s,\bar{Y}^{t, y,\nu^{1},\hat{a}}_{s})Z_s \right| \geq n\right\} \wedge \theta$ such that $\theta_n \nearrow \theta$. Then, the stochastic integral $\sigma\int_{t}^{\theta_n} \mathrm{e}^{-\rho(s-t)} \varphi_y^{A^{(1,d)},d,\Delta_N}(s,\bar{Y}^{t, y,\nu^{1},\hat{a}}_{s}) Z_s dB_s$ is a martingale.
     Taking expectation in (\ref{super.sol.ineq::v1}) and applying Lemma \ref{supersolconstruction}, we obtain 
     \begin{equation*}
         \mathbb{E}\left[\int_{t}^{\theta_n}\rho  \mathrm{e}^{-\rho(s-t)}\hat{a}_s \mathrm{d} s + \mathrm{e}^{-\rho(\theta_n - t)}\varphi^{A^{(1,d)},d,\Delta_N}\left(\theta_n, \bar{Y}_{\theta_n}^{t, y,\nu^{1},\hat{a}}\right)\right] \leq \varphi^{A^{(1,d)},d,\Delta_N}(t, y) .
     \end{equation*}
     Taking $\theta = \tau_1$, we have $\lim_{n \to \infty} \theta_n = \tau_1 $. Based on inequalities \eqref{upperlowerboundvp0}, we obtain the following inequality holds $\PP$-a.s.
    \begin{equation*}
    \begin{split}
       \mathcal{M}\hat{V}^{p,0}(\theta_n,\bar{Y}^{t,y,Z}_{\theta_n})&= \sup_{ 0 \leq \eta \leq \bar{Y}^{t, y,\nu^{1},\hat{a}}_{\theta_n}}\left(\hat{V}^{p,0}(\theta_n,Y^{t, y,\nu^{1},\hat{a}}_{\theta_n} -\eta) - (1 + k)\eta^{\gamma} \right)   \\
       &\leq   \sup_{ 0 \leq \eta \leq \bar{Y}^{t, y,\nu^{1},\hat{a}}_{\theta_n}}\left(\varphi^{A^{0},d,\Delta_0}(\theta_n,\bar{Y}^{t, y,\nu^{1},\hat{a}}_{\theta_n} -\eta) - (1 + k)\eta^{\gamma} \right)  \\
       &\leq  \varphi^{A^{(1,d)},d,\Delta_N}(\theta_n,\bar{Y}^{t, y,\nu^{1},\hat{a}}_{\theta_n}).  
    \end{split}
    \end{equation*}
Next, applying Lemma \ref{lm::hatVp1Vp0}, we have
\begin{equation*}
    \begin{split}
      \hat{V}^{p,1}(t,y) &=  \sup_{(Z,\tau_1) \in \hat{\mathcal{U}}^{1}(t,y) } \mathbb{E} \left[\int_{t}^{\tau_1}\rho  \mathrm{e}^{-\rho(s-t)}\hat{a}_s \mathrm{d} s + \mathrm{e}^{-\rho(\tau_1 - t)}\mathcal{M}\hat{V}^{p,0}(\tau_1,\bar{Y}^{t,y,Z}_{\tau_1})\right]   \\ 
      &=  \sup_{(Z,\tau_1)  \in \hat{\mathcal{U}}^{1}(t,y) } \lim_{n \to \infty}\mathbb{E} \left[\int_{t}^{\theta_n}\rho  \mathrm{e}^{-\rho(s-t)}\hat{a}_s \mathrm{d} s + \mathrm{e}^{-\rho(\tau_1 - t)}\mathcal{M}\hat{V}^{p,0}(\theta_n,\bar{Y}^{t,y,Z}_{\theta_n})\right]   \\ 
      &\leq  \sup_{(Z,\tau_1)  \in \hat{\mathcal{U}}^{1}(t,y) } \lim_{n \to \infty}\mathbb{E} \left[\int_{t}^{\theta_n}\rho  \mathrm{e}^{-\rho(s-t)}\hat{a}_s \mathrm{d} s + \mathrm{e}^{-\rho(\tau_1 - t)}\varphi^{A^{(1,d)},d,\Delta_N}(\theta_n,\bar{Y}^{t,y,Z}_{\theta_n})\right]   \\ 
      &\leq \varphi^{A^{(1,d)},d,\Delta_N}(t, y).
    \end{split}
\end{equation*}
In the second equality, we used the continuity of the mapping $(t,y) \mapsto \mathcal{M}\hat{V}^{p,0}(t,y)$, and the dominated convergence theorem. The previous completes the proof. \qed
\end{proof}

Having established the necessary auxiliary results, we now turn to the weak dynamic programming principle for the mixed control–stopping problem defined in (\ref{vpkpk-1}). Following the framework of \cite{bouchard2012weak}, we obtain a weak form of the dynamic programming principle that avoids the use of measurable selection arguments. The result is an immediate consequence of the general weak DPP under generalized state constraints established in \cite[Section 4]{bouchard2012weak}. Although we do not reproduce the proof here, the structural assumptions of our model fall within their setting, and the theorem applies directly. Specifically, for any stopping time $\tau \in \mathcal{T}^{t}_{[t,T]}$, we have 
  \begin{enumerate}
        
        \item \label{weakdpp::f1} Consider a measurable function $\psi:[0, T] \times (0,\infty) \rightarrow \mathbb{R}$ such that $\hat{V}^{p,1} \leq \psi$ and $\mathbb{E}\left[\psi\left(\tau, \bar{Y}_{\tau}^{t,y,Z}\right)^{-}\right]<\infty$ for all $(t, y) \in[0, T] \times (0,\infty)$. Then
        \begin{equation}\label{eq::principalvnopen::weakdpeq1}
        \begin{split}
            \hat{V}^{p,1}(t,y)\leq &\sup_{(Z,\tau_1) \in \hat{\mathcal{U}}^{1}(t,y) , \hat{a} \in \hat{\mathcal{A}}(Z)}\mathbb{E}\Bigg[ \int_{t}^{\tau \wedge  \tau_1 } \rho \mathrm{e}^{-\rho (s-t)}\left(\hat{a}_s \right) \mathrm{d} s + \mathbbm{1}_{\{\tau \leq \tau_1 \}}\mathrm{e}^{-\rho(\tau - t)}\psi(\tau,\bar{Y}_{\tau}^{t,y,Z}) \\
            & + \mathbbm{1}_{\{\tau > \tau_1 \}}  \mathrm{e}^{-\rho(\tau_1  - t)} \mathcal{M}\hat{V}^{p,0}(\tau_1,\bar{Y}^{t,y,Z}_{\tau_1})  \Bigg].
        \end{split}
    \end{equation}
    \item \label{weakdpp::f2} Let $\psi:[0, T] \times (0,\infty) \rightarrow \mathbb{R}$ be a upper-semicontinuous (u.s.c) function such that $\hat{V}^{p,1} \geq \psi$. Then,  $\mathbb{E}\left[\psi\left(\tau, \bar{Y}_{\tau}^{t,y,Z}\right)^{+}\right]<\infty$ for every $(Z,\tau_1) \in \hat{\mathcal{U}}^{1}(t,y)$ and
    \begin{equation}\label{eq::principalvnopen::weakdpeq2}
        \begin{split}
            \hat{V}^{p,1}(t,y)\geq &\sup_{(Z,\tau_1) \in \hat{\mathcal{U}}^{1}(t,y) , \hat{a} \in \hat{\mathcal{A}}(Z)}\mathbb{E}\Bigg[ \int_{t}^{\tau \wedge  \tau_1 } \rho \mathrm{e}^{-\rho (s-t)}\left(\hat{a}_s \right) \mathrm{d} s + \mathbbm{1}_{\{\tau \leq \tau_1 \}}\mathrm{e}^{-\rho(\tau - t)}\psi(\tau,\bar{Y}_{\tau}^{t,y,Z}) \\
            & + \mathbbm{1}_{\{\tau > \tau_1 \}}  \mathrm{e}^{-\rho(\tau_1  - t)} \mathcal{M}\hat{V}^{p,0}(\tau_1,\bar{Y}^{t,y,Z}_{\tau_1})  \Bigg]. 
        \end{split}
    \end{equation}
    \end{enumerate}

\subsubsection{Proof of Theorem \ref{hjbvi:pp:in}.1 on $V^{p,n}$ }

In this subsection, we begin by showing that $\hat{V}^{p,1}$ is a viscosity solution to the HJB variational inequalities \eqref{vi::principalvnopen} and \eqref{vi::principalvnopenterminal}. From section \ref{thm2.1::vp0}, we know that Theorem \ref{hjbvi:pp:in} holds for $V^{p,0}$.

\begin{lemma}\label{viscositycharacterizationvigeq1}
  The value function ${\hat{V}^{p,1}}_{*}(t,y)$ ($ {{\hat{V}}}^{{p,1}^*}(t,y)$) is a lower semicontinuous (upper semicontinuous) viscosity supersolution (viscosity subsolution) of the Hamilton-Jacobi-Bellman variational inequality defined on the domain $[0,T)\times (0,\infty)$
  \begin{equation}\label{vi::principalvnopen}
     \min\{ -v^{p,1}_t + \mathcal{G}v^{p,1}  , v^{p,1} - \mathcal{M} \hat{V}^{p,0} , - v^{p,1}_{yy}\} = 0 \;.
  \end{equation}
  For the terminal condition, ${\hat{V}^{p,1}}_{*}(t,y)$ ($ {{\hat{V}}}^{{p,1}^*}(t,y)$) is a lower semicontinuous viscosity supersolution (viscosity subsolution) of the Hamilton-Jacobi-Bellman variational inequality defined on the domain $(0,\infty)$
  \begin{equation}\label{vi::principalvnopenterminal}
    \min\{v^{p,1}(T-,\cdot) - g^{1} , - v^{p,1}_{yy}(T-,\cdot) \} = 0, 
  \end{equation}
  where  $ g^{1} = \mathcal{M}^{1}F$.
\end{lemma}
\begin{proof}
    To begin with, we consider the situation when $t < T$. Let $\left(t_m, y_m\right)_{m \geq 1} \subset[0, T) \times(0, \infty)$ be an arbitrary sequence converging to $(t, y) \in[0, T) \times (0, \infty)$. For each $m\in \mathbb{N}$, we choose the admissible control $\tau_1:=t_m$. By Lemma \ref{lm::hatVp1Vp0}, we have
        ${{\hat{V}}}^{{p,1}}(t_m,y_m) \geq \mathcal{M}\hat{V}^{p,0}(t_m,y_m).$
Using that the mapping $(t,y)\mapsto \mathcal{M}\hat{V}^{p,0}(t,y)$ is continuous, we obtain
$${\hat{V}^{p,1}}_{*}(t,y) = \liminf_{(t_m,y_m) \to (t,y)} {{\hat{V}}}^{{p,1}}(t_m,y_m) \geq \lim_{(t_m,y_m) \to (t,y)}\mathcal{M}\hat{V}^{p,0}(t_m,y_m) =\mathcal{M} \hat{V}^{p,0}(t,y).$$
Moreover,
\begin{equation*}
    {{\hat{V}}}^{{p,1}^*}(t,y) \geq {\hat{V}^{p,1}}_{*}(t,y) \geq \mathcal{M} \hat{V}^{p,0}(t,y) \;,\; \forall (t,y) \in [0,T) \times (0,\infty). 
\end{equation*}
As a result, rather than proving directly that ${\hat{V}^{p,1}}_{*}(t,y)$ is a viscosity supersolution of equation  \eqref{vi::principalvnopen}, it suffices to verify that ${\hat{V}^{p,1}}_{*}(t,y)$ is a viscosity supersolution of the following variational inequality
\begin{equation}\label{vi::principalvopen2}
     \min\{ -v^{p,1}_t + \mathcal{G}v^{p,1} , - v^{p,1}_{yy}\} = 0 \;,\;  [0,T) \times (0,\infty).  
\end{equation}
Follow the similar viscosity supersolution charaterization argument in lemma \ref{viscositycharacterizationv0}, we get that  ${\hat{V}^{p,1}}_{*}(t,y)$ is a supersolution to (\ref{vi::principalvnopen}).

\medskip
Secondly, we show that $\hat{V}^{{p,1}^{*}}(t,y)$  is a viscosity subsolution to the variational inequality (\ref{vi::principalvnopen}).
Assume the contrary. That is, there exist $\left(t_0, y_0\right) \in[0, T) \times(0, \infty)$ and a smooth test function $\varphi:[0, T) \times(0, \infty) \rightarrow \mathbb{R}$, that satisfies
$$0 = ({{\hat{V}}}^{{p,1}^*} - \varphi) (t_0,y_0) >  ( {{\hat{V}}}^{{p,1}^*} - \varphi)(t,y),\;\; \forall (t,y) \in [0,T) \times (0,\infty)\;,\; (t,y)\neq (t_0,y_0),
$$
such that $$(\varphi - \mathcal{M}{\hat{V}}^{p,0}) (t_0,y_0) > 0\;,\;\min\{ -\varphi_t(t_0,y_0) + \mathcal{G}\varphi(t_0,y_0) , - \varphi_{yy}(t_0,y_0)\}  > 0.$$
Due to the inequality $(\varphi - \mathcal{M}{\hat{V}}^{p,0}) (t_0,y_0) > 0$, there exist $\e>0, r_\e>0$, such that $$\varphi(t, y) \geq \mathcal{M}{\hat{V}}^{p,0} (t,y) + 
\mathrm{e}^{\rho(t_0+r_\e)}\e ,$$
for all $(t, y)$ satisfying $t_0 \leq t \leq$ $t_0+r_\e<T$, and $\left|y-y_0\right| \leq r_\e$.\\
Next, we introduce the function
$$\bar{\varphi}(t,y) :=\varphi(t,y) + \mathrm{e}^{\rho t}(|t - t_0|^{2} + |y - y_0|^{4}).$$
Based on $\left(\varphi, \varphi_t, \varphi_y,  \varphi_{yy}\right)(t_0, y_0)=\left(\bar{\varphi},\bar{\varphi}_t, \bar{\varphi}_y,  \bar{\varphi}_{yy}\right)(t_0, y_0)$, and the continuity of $\mathcal{G}$, we obtain
\begin{equation*}
    -\bar{\varphi}_t(t,y) + \mathcal{G}^{\tilde{z}}\bar{\varphi}(t,y) > 0 \;,\;  \forall \tilde{z} \in [0,\infty) \;,\;  \forall (t,y) \in B^{r_\e}(t_0,y_0).
\end{equation*}
With choice of $r_\e$, through the similar scaling argument as inequality (\ref{v0barvarphivarphi.1}), we have
\begin{equation}\label{barvarphivarphisub}
   \bar{\varphi}(t,y) \geq \varphi(t,y) + e^{\rho(t_0+r_\e)}  \eta_\e \;,\;\text{for some}\; \eta_\e > 0,\; (t,y) \in ( [0,T] \times (0,\infty)  ) \setminus B^{r_\e}(t_0,y_0).  
\end{equation}
Let $(t_m,y_m)_{m}$ be a sequence in $B^{r_\e}(t_0,y_0)$ such that 
$$(t_m,y_m) \to (t_0,y_0) \;,\; \hat{V}^{p,1}(t_m,y_m) \to {{\hat{V}}}^{{p,1}^*}(t_0,y_0).$$
Fix $m \in \mathbb{N}$. Let $(\tilde{z}^{m},\tau^m_1)$ be an arbitrary control of $\mathcal{\hat{U}}^{1}(t_m,y_m)$ and $\bar{Y}^{t_m,y_m,\tilde{z}^{m}}$ the associated state process.
We define the stopping time $\tau_m$ as
$$\tau_m = \inf\{s \geq t_m : (s, \bar{Y}^{t_m,y_m,\tilde{z}^{m}}_s) \notin B^{r_\e}(t_0,y_0) \}.$$ 
Applying ito's lemma to $\bar{\varphi}$, we have
\begin{equation*}
\begin{split}
       & \bar{\varphi}(t_m,y_m) \\
      =& \mathbb{E}\bigg[  \mathrm{e}^{-\rho (\tau_m \wedge \tau^m_1 - t_m)}\bar{\varphi}(\tau_m \wedge \tau^m_1, \bar{Y}^{t_m,y_m,\tilde{z}^{m}}_{\tau_m \wedge \tau^m_1})  \\
      &\quad\quad\quad\quad \quad\quad + \int_{t_m}^{\tau_m \wedge \tau^m_1}\mathrm{e}^{-\rho (s - t_m)}( -\bar{\varphi}_t(s,\bar{Y}^{t_m,y_m,\tilde{z}^{m}}_{\tau_s}) + \mathcal{G}^{\tilde{z}^{m}}\bar{\varphi}(s,\bar{Y}^{t_m,y_m,\tilde{z}^{m}}_{\tau_s}) )  \mathrm{d} s   \bigg] \\
     \geq & \mathbb{E}\bigg[\mathrm{e}^{-\rho (\tau_m \wedge \tau^m_1 - t_m)}\bar{\varphi}(\tau_m \wedge \tau^m_1, \bar{Y}^{t_m,y_m,\tilde{z}^{m}}_{\tau_m \wedge \tau^m_1}) +  \int_{t_m}^{\tau_m \wedge \tau^m_1} \rho \mathrm{e}^{-\rho (s - t_m)} \hat{a}^{\tilde{z}^{m}}_s \mathrm{d} s   \bigg]  \\ 
     = &\mathbb{E}\bigg[ \mathrm{e}^{-\rho (\tau_m  - t_m)}\bar{\varphi}(\tau_m , \bar{Y}^{t_m,y_m,\tilde{z}^{m}}_{\tau_m}) \mathbbm{1}_{\{\tau_m \leq \tau^m_1\}}  + \mathrm{e}^{-\rho (\tau^m_1  - t_m)}\bar{\varphi}(\tau^m_1 , \bar{Y}^{t_m,y_m,\tilde{z}^{m}}_{\tau^m_1}) \mathbbm{1}_{\{\tau^m_1 < \tau_m\}} \\
      &\quad\quad\quad\quad \quad\quad + \int_{t_m}^{\tau_m \wedge \tau^m_1} \rho \mathrm{e}^{-\rho (s - t_m)} \hat{a}^{\tilde{z}^{m}}_s \mathrm{d} s \bigg]  \\
     \geq &\mathbb{E}\bigg[ \mathrm{e}^{-\rho (\tau_m  - t_m)} \varphi(\tau_m , \bar{Y}^{t_m,y_m,\tilde{z}^{m}}_{\tau_m}) \mathbbm{1}_{\{\tau_m \leq \tau^m_1\}} + \mathbbm{1}_{\{\tau_m \leq \tau^m_1\}}\eta_\e  + \int_{t_m}^{\tau_m \wedge \tau^m_1} \rho \mathrm{e}^{-\rho (s - t_m)} \hat{a}^{\tilde{z}^{m}}_s \mathrm{d} s \\
      &\quad\quad\quad\quad\quad\quad  + \mathrm{e}^{-\rho (\tau^m_1  - t_m)}\mathcal{M}{\hat{V}}^{p,0}(\tau^m_1 , \bar{Y}^{t_m,y_m,\tilde{z}^{m}}_{\tau^m_1}) \mathbbm{1}_{\{\tau^m_1 < \tau_m\}} + \mathbbm{1}_{\{\tau^m_1 < \tau_m\}} \e \bigg]  \\
    \geq & \mathbb{E}\bigg[ \mathrm{e}^{-\rho (\tau_m  - t_m)} \varphi(\tau_m , \bar{Y}^{t_m,y_m,\tilde{z}^{m}}_{\tau_m}) \mathbbm{1}_{\{\tau_m \leq \tau^m_1\}} + \mathcal{M}{\hat{V}}^{p,0}(\tau^m_1 , \bar{Y}^{t_m,y_m,\tilde{z}^{m}}_{\tau^m_1}) \mathbbm{1}_{\{\tau^m_1 < \tau_m\}} \\
      &\quad\quad\quad\quad \quad\quad  +   \int_{t_m}^{\tau_m \wedge \tau^m_1} \rho \mathrm{e}^{-\rho (s - t_m)} \hat{a}^{\tilde{z}^{m}}_s \mathrm{d} s\bigg]  + \e \wedge \eta_\e. \\
\end{split}
\end{equation*}
The first inequality comes from
\begin{equation*}
    \begin{split}
       &\quad -\bar{\varphi}_t(t,y)  + \mathcal{G}^{\tilde{z}^{m}}\bar{\varphi}(t,y)  \\
       &=  -\bar{\varphi}_t(t,y)  +  \rho \bar{\varphi} - \sup_{\hat{a} \in \hat{\mathcal{A}}(\tilde{z}^{m})}\{rh(\hat{a})\bar{\varphi}_y + \rho \hat{a} + \frac{1}{2}\bar{\varphi}_{yy}r^{2}\sigma^{2}\tilde{z}^{{m}^{2}}  \} - r\bar{\varphi}_y y\\
       & > 0 \quad ,  \quad \forall (t,y) \in B^{r_\e}(t_0,y_0).
    \end{split}
\end{equation*}
As a conclusion, we have
\begin{equation*}
\begin{split}
   &\bar{\varphi}(t_m,y_m) \\
   \geq& \mathbb{E}\bigg[ \mathrm{e}^{-\rho (\tau_m  - t_m)} \varphi(\tau_m ,\bar{Y}^{t_m,y_m,\tilde{z}^{m}}_{\tau_m}) \mathbbm{1}_{\{\tau_m \leq \tau^m_1\}} + \mathcal{M}{\hat{V}}^{p,0}(\tau^m_1 ,\bar{Y}^{t_m,y_m,\tilde{z}^{m}}_{\tau^m_1}) \mathbbm{1}_{\{\tau^m_1 < \tau_m\}} \\
   &\quad\quad\quad\quad\quad + \int_{t_m}^{\tau_m \wedge \tau^m_1} \rho \mathrm{e}^{-\rho (s - t_m)} \hat{a}^{\tilde{z}^{m}}_s \mathrm{d} s\bigg]  + \e \wedge \eta_\e . 
\end{split}
\end{equation*}
For sufficiently large $m$, there exists $\frac{\eta_\e \wedge \e}{2}$ such that $\hat{V}^{p,1}(t_m,y_m) +  \frac{\eta_\e \wedge \e}{2} \geq \bar{\varphi}(t_m,y_m)$.
Then we have
\begin{equation*}
\begin{split}
    &\hat{V}^{p,1}(t_m,y_m) \\
    \geq& \mathbb{E}\Bigg[ \mathrm{e}^{-\rho (\tau_m  - t_m)} \varphi(\tau_m ,\bar{Y}^{t_m,y_m,\tilde{z}^{m}}_{\tau_m}) \mathbbm{1}_{\{\tau_m \leq \tau^m_1\}} + \mathcal{M}{\hat{V}}^{p,0}(\tau^m_1 ,\bar{Y}^{t_m,y_m,\tilde{z}^{m}}_{\tau^m_1}) \mathbbm{1}_{\{\tau^m_1 < \tau_m\}}\\
    &\quad\quad\quad\quad\quad + \int_{t_m}^{\tau_m \wedge \tau^m_1} \rho \mathrm{e}^{-\rho (s - t_m)} \hat{a}^{\tilde{z}^{m}}_s \mathrm{d} s\Bigg]  + \frac{\eta_\e \wedge \e}{2}
\end{split} 
\end{equation*}
for any arbitrary admissible control $(\tilde{z}^{m},\tau_{1}^{m}) \in \mathcal{\hat{U}}^{1}(t_m,y_m)$. Since $\varphi \geq \hat{V}^{p,1}$, the above contradicts  (\ref{eq::principalvnopen::weakdpeq1}).

\medskip
Next, we turn to the viscosity characterization of the value function $\hat{V}^{p,1}$ at the terminal time $T$.

\medskip
Firstly, we show that ${\hat{V}^{p,1}}_{*}$ is the viscosity supersolution to the following variational inequality:
$$ \min\{v^{p,1}(T-,y) - g^{1}(y) , - v^{p,1}_{yy}(T-,y) \} = 0 $$
Applying lemma \cite[Lemma 4.3.2]{pham2009continuous}, we get that ${\hat{V}^{p,1}}_{*}$ is the viscosity supersolution.

\medskip
Secondly, we show that $\hat{V}^{{p,1}^{*}}(T,y)$  is a viscosity subsolution of the variational inequality (\ref{vi::principalvnopenterminal}).\\
Let $(T,y_0) \in\{ T \}\times (0,\infty)$ and $\varphi: [0,T] \times (0,\infty) \to \mathbb{R}$ be a smooth test function such that
$$0 =  \hat{V}^{{p,1}^{*}}(T,y_0) - \varphi (y_0) >   \hat{V}^{{p,1}^{*}}(T,y) - \varphi(y) \; , \; \forall y \in  (0,\infty) \;,\; y\neq y_0.$$
Assume, on the contrary, that
$$\min\{ \hat{V}^{{p,1}^{*}}(T,y_0) - g^{1}(y_0) , - \varphi_{yy}(y_0)\}  >  0.$$\\
Then, without loss of generality, we have
\begin{equation}\label{Vstarsubeps}
     \hat{V}^{{p,1}^{*}}(T,y_0) - g^{1}(y_0)  >  2 \mathrm{e}^{\rho T}\e,\;\text{for some}\; \e > 0.
\end{equation}
We set 
$$\bar{\varphi}(t,y) =\varphi(y) + \mathrm{e}^{\rho T}(c\sqrt{T - t} + |y - y_0|^{4}) ,\; c > 0,\;$$
where $c$ is a positive fixed constant.
Then, for $t$ sufficiently close to $T$, we have
$$-\bar{\varphi}_t(t,y_0)  + \mathcal{G}\bar{\varphi}(t,y_0)  > 0.$$
By the continuity of $\mathcal{G}$, for all $c>0$, there exists a pair $(r^c_t,r^c_y)\in \RR^{2,+}$  such that 
\begin{equation}\label{ballsubterminal}
    -\bar{\varphi}_t(t,y) + \mathcal{G}^{\tilde{z}}\bar{\varphi}(t,y) > 0 \;,\;  \forall \tilde{z} \in [0,\infty) \;,\;   \forall (t,y) \in [T - r_t^c,T) \times B(y_0,r^c_y),
\end{equation}
where $B(y_0,r^c_y)$ denotes the ball of radius $r^c_y$ and center $y_0$.\\
Using the upper bound found in Lemma \ref{supersolvpk}, we obtain that the function $\hat{V}^{{p,1}^{*}}- \bar{\varphi}$ has a finite maximum on any compact set. 
Hence, taking $c>0$ large enough, there exists a constant $\eta>0$ such that
\begin{equation}\label{barvarphivarphiterminal2}
   \bar{\varphi}(t,y) \geq \hat{V}^{{p,1}^{*}}(t,y) +  2 e^{\rho T}  \eta \;,\; (t,y) \in  [T - r^c_t,T] \times \partial B(y_0,r^c_y),
\end{equation}
In addition, based on continuity of $\mathcal{M}\hat{V}^{{p,0}}$ and $\bar{\varphi}$, with $c$ large enough, there exists a constant $\tilde{\eta} > 0$ satisfying
\begin{equation}\label{barvarphiMvterminal3}
   \bar{\varphi}(t,y) \geq \mathcal{M}\hat{V}^{p,0}(t,y) +  2 e^{\rho T}  \tilde{\eta} \;,\; (t,y) \in  [T - r^c_t,T] \times B(y_0,r^c_y).
\end{equation}

Let $(t_m,y_m)_{m\geq 1}$ be a sequence in $ [T - r^c_t,T] \times B(y_0,r^c_y) $ such that 
$$(t_m,y_m) \to (T,y_0) \; , \; \hat{V}^{p,1}(t_m,y_m) \to \hat{V}^{{p,1}^{*}}(T,y_0).$$
Fix $m \in \mathbb{N}$. Let $(\tilde{z}^{m},\tau^m_1)$ be an arbitrary control of $\mathcal{\hat{U}}^{1}(t_m,y_m)$, and $\bar{Y}^{t_m,y_m,\tilde{z}^{m}}$ be its associated state process.
We introduce the following sequence of stopping times
\begin{equation*}
    \begin{split}
        \theta^m&:=\inf \left\{s \geq t_m,\left(s, \bar{Y}_s^{t_m,y_m,\tilde{z}^{m}}\right) \in ([T - r^c_t,T] \times  B(y_0,r^c_y))^{c}\right\} \wedge T.
    \end{split}
\end{equation*}
By applying Ito's lemma to $\bar{\varphi}$, we obtain
\begin{equation*}
\begin{split}
    &\bar{\varphi}(t_m,y_m) \\
      =& \mathbb{E}\Bigg[ \int_{t_m}^{\theta^m \wedge \tau^m_1}\mathrm{e}^{-\rho (s - t_m)}( -\bar{\varphi}_t(s,\bar{Y}^{t_m,y_m,\tilde{z}^{m}}_{s}) + \mathcal{L}^{\tilde{z}^{m}}\bar{\varphi}(s,\bar{Y}^{t_m,y_m,\tilde{z}^{m}}_{s}) )  \mathrm{d} s\\
      &\quad\quad\quad\quad \quad\quad + \mathrm{e}^{-\rho (\theta^m \wedge \tau^m_1 - t_m)}\bar{\varphi}(\theta^m \wedge \tau^m_1, \bar{Y}^{t_m,y_m,\tilde{z}^{m}}_{\theta^m \wedge \tau^m_1})  \Bigg] \\
     \geq&  \mathbb{E}\Bigg[ \int_{t_m}^{\theta^m \wedge \tau^m_1} \rho \hat{a}^{\tilde{z}^{m}}_s  \mathrm{e}^{-\rho (s-t_m)} \mathrm{d} s + \mathrm{e}^{-\rho (\theta^m \wedge \tau^m_1 - t_m)}\bar{\varphi}(\theta^m \wedge \tau^m_1, \bar{Y}^{t_m,y_m,\tilde{z}^{m}}_{\theta^m \wedge \tau^m_1})      \Bigg]  \\ 
     =& \mathbb{E}\Bigg[ \int_{t_m}^{\theta^m \wedge \tau^m_1} \rho \hat{a}^{\tilde{z}^{m}}_s  \mathrm{e}^{-\rho (s-t_m)} \mathrm{d} s  + \mathbbm{1}_{ \{\theta^m \leq \tau^m_1\} } \mathrm{e}^{-\rho (\theta^m  - t_m)}\bar{\varphi}(\theta^m , \bar{Y}^{t_m,y_m,\tilde{z}^{m}}_{\theta^m }) \\
     &\quad\quad\quad\quad \quad\quad  +  \mathbbm{1}_{ \{ \tau^m_1 <  \theta^m \} } \mathrm{e}^{-\rho (\tau^m_1  - t_m)}\bar{\varphi}(\tau^m_1 , \bar{Y}^{t_m,y_m,\tilde{z}^{m}}_{\tau^m_1 }) \Bigg]  \\
     \geq&  \mathbb{E}\Bigg[ \int_{t_m}^{\theta^m \wedge \tau^m_1} \rho \hat{a}_s \mathrm{e}^{-\rho (s-t_m)} \mathrm{d} s +  \mathbbm{1}_{ \{\theta^m \leq \tau^m_1\} } \mathrm{e}^{-\rho (\theta^m  - t_m)}\left(\hat{V}^{{p,1}^{*}}(\theta^m , \bar{Y}^{t_m,y_m,\tilde{z}^{m}}_{\theta^m }) + 2\mathrm{e}^{\rho T} \eta \right) \\
     & \quad\quad\quad\quad \quad\quad  +  \mathbbm{1}_{ \{ \tau^m_1 <  \theta^m \} } \mathrm{e}^{-\rho (\tau^m_1  - t_m)}\left(\mathcal{M}\hat{V}^{p,0}(\tau^m_1 , \bar{Y}^{t_m,y_m,\tilde{z}^{m}}_{\tau^m_1 }) + 2\mathrm{e}^{\rho T} \tilde{\eta} \right)  \Bigg]  \\
    > & \mathbb{E}\Bigg[ \int_{t_m}^{\theta^m \wedge \tau^m_1} \rho \hat{a}_s \mathrm{e}^{-\rho (s-t_m)} \mathrm{d} s +  \mathbbm{1}_{ \{\theta^m \leq \tau^m_1\} } \mathrm{e}^{-\rho (\theta^m  - t_m)}\varphi(\theta^m , \bar{Y}^{t_m,y_m,\tilde{z}^{m}}_{\theta^m })\\
    &\quad\quad\quad\quad \quad\quad +  \mathbbm{1}_{ \{ \tau^m_1 <  \theta^m \} } \mathrm{e}^{-\rho (\tau^m_1  - t_m)}\mathcal{M}\hat{V}^{p,0}(\tau^m_1 , \bar{Y}^{t_m,y_m,\tilde{z}^{m}}_{\tau^m_1 })   \Bigg] + 2  (\tilde{\eta}\wedge \eta).  \\
\end{split}
\end{equation*}
The first inequality follows from (\ref{ballsubterminal}). The second and third inequalities follow from  (\ref{barvarphivarphiterminal2}) and (\ref{barvarphiMvterminal3}).
In addition, we know that when $m$ is large enough, we have 
$$\hat{V}^{p,1}(t_m,y_m) + \tilde{\eta} \wedge \eta \geq \bar{\varphi}(t_m,y_m).$$
Then, we obtain
\begin{equation*}
    \begin{split}
        & \hat{V}^{p,1}(t_m,y_m)\\
        &\geq \mathbb{E}\Bigg[ \int_{t_m}^{\theta^m \wedge \tau^m_1} \rho \hat{a}_s \mathrm{e}^{-\rho (s-t_m)} \mathrm{d} s +   \mathbbm{1}_{ \{\theta^m \leq \tau^m_1\} } \mathrm{e}^{-\rho (\theta^m  - t_m)}\hat{V}^{{p,1}^{*}}(\theta^m , Y^{t_m,y_m,\tilde{z}^{m}}_{\theta^m }) \\
        & \quad\quad\quad\quad \quad\quad +  \mathbbm{1}_{ \{ \tau^m_1 <  \theta^m \} } \mathrm{e}^{-\rho (\tau^m_1  - t_m)}\mathcal{M}\hat{V}^{p,0}(\tau^m_1 , Y^{t_m,y_m,\tilde{z}^{m}}_{\tau^m_1 })   \Bigg] +   \tilde{\eta}\wedge \eta  \\
    \end{split}
\end{equation*}
for any arbitrary admissible control $(\tilde{z}^{m},\tau_{1}^{m}) \in \mathcal{\hat{U}}^{1}(t_m,y_m)$. Since $\hat{V}^{{p,1}^{*}} \geq \hat{V}^{p,1}$, the above contradicts  (\ref{eq::principalvnopen::weakdpeq1}).
\qed
\end{proof}
\medskip
Next, we turn to the discussion of the terminal condition. By Lemma \ref{lm::hatVp1Vp0}, we know that $\hat{V}^{p,1}$ has polynomial growth. Moreover, the variational inequality (\ref{vi::principalvnopenterminal}) falls within the framework of \cite[Section V.1]{ishii1990viscosity}. Consequently, the comparison principle holds, and we conclude that
$$
\hat{V}^{p, 1^*}(T, y)=\hat{V}_*^{p, 1}(T, y)= g^{1}(y)\;,\; y \in (0,\infty).
$$
Then, we justify that $V^{p,1}(T,y) = \hat{V}(T,y) = g^{1}(y) \;,\; y \in(0,\infty)$. Indeed, picking any admissible control $(Z,\tau_1,\eta_1) \in \mathcal{U}^{1}(t,y)$, we obtain
\begin{equation*}
    \begin{split}
     &J^{p,1}(t,y;(Z,\tau_1,\eta_1) )\\
     =&  \mathbb{E}\left[ \int_t^{T} \rho \mathrm{e}^{-\rho (s -t) }\left(\hat{a}_s \right) \mathrm{d} s   - \mathrm{e}^{-\rho (\tau_1 - t)} (1+k)\eta^{\gamma}_1   - \mathrm{e}^{\rho(T-t)}u^{-1}(Y^{t,y, Z,\tau_1,\eta_1}_{T})    \right]    \\
     =&\mathbb{E}\left[ \int_t^{\tau_1} \rho \mathrm{e}^{-\rho (s -t) }\left(\hat{a}_s \right) \mathrm{d} s   + \mathrm{e}^{-\rho (\tau_1 - t)} J^{p,0}(\tau_1,\bar{Y}^{t,y,Z,\tau_1}_{\tau_1} -\eta_1 ) -\mathrm{e}^{-\rho (\tau_1 - t)} (1+k)\eta^{\gamma}_1   \right] \\
     \leq& \mathbb{E}\left[ \int_t^{\tau_1} \rho \mathrm{e}^{-\rho (s -t) }\left(\hat{a}_s \right) \mathrm{d} s   + \mathrm{e}^{-\rho (\tau_1 - t)} \hat{V}^{p,0}(\tau_1,\bar{Y}^{t,y,Z,\tau_1}_{\tau_1} -\eta_1)  - \mathrm{e}^{-\rho (\tau_1 - t)} (1+k)\eta^{\gamma}_1   \right] \\
     \leq& \mathbb{E}\left[ \int_t^{\tau_1} \rho \mathrm{e}^{-\rho (s -t) }\left(\hat{a}_s \right) \mathrm{d} s   + \mathrm{e}^{-\rho (\tau_1 - t)} \mathcal{M}V^{p,0}(\tau_1,\bar{Y}^{t,y,Z}_{\tau_1})  \right] \\
     \leq& \sup_{(Z,\tau_1) \in \mathcal{U}^{1}(t,y)  , \hat{a} \in \hat{\mathcal{A}}(Z)} \mathbb{E}\left[ \int_t^{\tau_1} \rho \mathrm{e}^{-\rho (s -t) }\left(\hat{a}_s \right) \mathrm{d} s   + \mathrm{e}^{-\rho (\tau_1 - t)} \mathcal{M}V^{p,0}(\tau_1,\bar{Y}^{t,y,Z}_{\tau_1})\right]. 
      \end{split}
\end{equation*}
By taking the supremum over admissible control in  $\mathcal{U}^{1}(t,y)$ on the left side, we have
$$V^{p,1}(t,y) \leq  \sup_{(Z,\tau_1) \in \mathcal{U}^{1}(t,y)  , \hat{a} \in \hat{\mathcal{A}}(Z)} \mathbb{E}\left[ \int_t^{\tau_1} \rho \mathrm{e}^{-\rho (s -t) }\left(\hat{a}_s \right) \mathrm{d} s   + \mathrm{e}^{-\rho (\tau_1 - t)} \mathcal{M}V^{p,0}(\tau_1,\bar{Y}^{t,y,Z}_{\tau_1})   \right].$$
By letting $t \rightarrow T$, we obtain
$$V^{p,1}(T-,y) \leq \mathcal{M}V^{p,0}(T,y) = g^{1}(y) \; , \; y \in (0,\infty).$$
Next, we proceed to show that $\hat{V}^{p,1}$ admits a continuous extension to the boundary $[0,T)\times \{0\}$. Indeed, due to the inclusion $\mathcal{\hat{U}}^{1}(t,y) \subset \mathcal{U}^{1}(t,y)$, we have 
$$V^{p,1}(t,y) \geq \hat{V}^{p,1}(t,y) \geq \mathcal{M}\hat{V}^{p,0}(t,y) =\mathcal{M}V^{p,0}(t,y) , (t,y) \in \{T\} \times (0,\infty).$$

\medskip
As a result, we conclude that $V^{p,1}(T-,y) = \hat{V}^{p,1}(T-,y) = g^{1}(y), y\in (0,\infty)$.

\medskip
By applying Lemma \ref{supersolvpk}, we know that $\hat{V}^{p, 1}(t, y) \leq \varphi^{A^{(1,d)},d,\Delta_N}(t, y),\; \forall(t, y) \in[0, T] \times(0, \infty)$.
Moreover, by Lemma {\ref{lm::hatVp1Vp0}} and considering the admissible controls in the subcases  $d > 0 : (Z = 0, \tau_1 = T-) $ and $d \leq 0 : (Z = 0, \tau_1 = t)$, we have
\begin{equation*}
   - A^{(1,d)}(t) e^{d(t-T)} y^{\gamma} \leq \hat{V}^{p,1}(t,y) \leq  \varphi^{A^{(1,d)},d,\Delta_N}(t, y). 
\end{equation*}
The previous implies
$$ \lim_{(t',y') \to (t,0),y' > 0} -   A^{(1,d)}(t') e^{d(t'-T)} {y'}^{\gamma} \leq {\hat{V}^{p,1}}_{*}(t,0) \leq {{\hat{V}}}^{{p,1}^*}(t,0) \leq  \lim_{(t',y') \to (t,0),y' > 0}  \varphi^{A^{(1,d)},d,\Delta_N}(t', y').$$
After discussing the terminal and boundary  behavior of $\hat{V}^{p,1}$, we can apply Lemma \ref{cmp::vnhjbvi} to deduce that $\hat{V}^{p,1}$ is the unique viscosity solution to the following Dirichlet problem
\begin{equation}\label{hjbvi:pp:dirichletvk}
\begin{split}
 \min\{ -v^{p,1}_t + \mathcal{G}v^{p,1}  , v^{p,1} - \mathcal{M} v^{p,0} , - v^{p,1}_{yy}\} &= 0, \quad [0,T) \times (0,\infty),    \\
v^{p,1}(T,y) &= g^{1}(y) , \quad (0,\infty),\\
v^{p,1}(t,0) &= 0, \quad t \in [0,T]. \\
\end{split}   
\end{equation}
Next, it is evident that the inequality (\ref{eq::principalvnopen::weakdpeq1})  also holds for $V^{p,1}$. Following an argument similar to the one presented in Lemma \ref{viscositycharacterizationvigeq1}, we conclude that ${V^{p,1}}^{*}$ is a viscosity subsolution to the Dirichlet Problem (\ref{hjbvi:pp:dirichletvk}). Finally, applying Lemma \ref{cmp::vnhjbvi}, we have
$${\hat{V}^{p,1}}_{*}(t,y) \leq {V_{*}^{p,1}}(t,y)  \leq {V}^{{p,1}^*}(t,y) \leq {\hat{V}^{p,1}}_{*}(t,y),\quad (t,y) \in [0,T) \times (0,\infty).$$
In conclusion, $V^{p,1}$ is the unique continuous viscosity solution to the Dirichlet Problem (\ref{hjbvi:pp:dirichletvk}).

Having established the result for $V^{p, 1}$, we note that the same reasoning extends to all subsequent indices $n \geq 2$, since the recursive structure of the problem implies that the relation between $V^{p, n}$ and $V^{p, n-1}$ is identical to that between $V^{p, 1}$ and $V^{p, 0}$. Therefore, by induction, we have that Theorem \ref{hjbvi:pp:in}.1 holds for all $V^{p, n}$ with $n \in \{0,\ldots,N\}$.

\subsubsection{Proof of Theorem \ref{hjbvi:pp:in}.2 on $V^{p,n}$}
We first consider the case $n=1$. Since $V^{p, 1}(t, y)$ is a viscosity supersolution of Equation \eqref{hjbvi:pp:dirichletvk}, it must satisfy $-V_{y y}^{p, 1}(t, y) \geq 0$ in the viscosity sense; otherwise, $\mathcal{J}^{+}\left(V_y^{p, 1}(t, y), V_{y y}^{p, 1}(t, y)\right)=\infty$.
This implies that $V^{p, 1}(t, y)$ is partially concave in $y$. By induction, the same argument applies to $V^{p, n}$ for all $n \geq 2$.
Hence, we conclude that Theorem \ref{hjbvi:pp:in}.2 holds for all $V^{p, n}$ with $n \geq 1$.

\subsubsection{Proof of Theorem \ref{hjbvi:pp:in}.3 on $V^{p,n}$}
To prove the general conclusion of Theorem \ref{hjbvi:pp:in}.3 for $V^{p,n}, n \in \{1,\cdots,N\}$, it suffices to consider the case $V^{p, 1}$, since only the partial concavity of $V^{p, n}$ with respect to the agent's utility is used in the argument.\\
    To begin with, we show $$\eta^{1,*}(t,y) =  \inf \{\eta \in [0,y] : \partial^{+}_{y}V^{p,0}(t,y - \eta) > -\gamma (1 + k)\eta^{\gamma - 1} \}  \wedge y .$$
        Based on the Karush–Kuhn–Tucker (KKT) method, we have the necessary characterization for $\eta^{1,*}(t,y)$:
        \begin{itemize}
            \item Stationarity: $\quad 0= (1 + k) \gamma  \eta^{\gamma - 1} +  v  +a_{2}-a_{1}$
        for some $v \in\left[\partial_y^{+} V^{p,0}, \partial_y^{-} V^{p,0}\right]\left(t, y-\eta \right)$,
            \item Primal feasibility: $0 \leq \eta \leq y$,
            \item Dual feasibility: $\quad a_1 ,a_2\geq 0$,
            \item Complementary Slackness: $\quad a_2\left(y - \eta \right)=0, \quad a_1 \eta =0$,
        \end{itemize}
        where $\partial_y^{-} V$ and $\partial_y^{+} V$ are the left and right derivative of $V$ with respect to $y$, which exist everywhere by concavity.\\
        We define 
        \begin{equation*}
            \begin{split}
                &\overline{\eta}^{1} :=\inf \mathbb{O}   = \inf \{\eta \in [0,y] : \partial^{+}_{y}V^{p,0}(t,y - \eta) > - (1+k)\gamma \eta^{\gamma - 1} \}, 
            \end{split}
        \end{equation*}
        where  $\mathbb{O} := \{\eta \in [0,y] : \partial^{+}_{y}V^{p,0}(t,y - \eta) > - (1+k)\gamma \eta^{\gamma - 1}  \}$.
    
        First, we show  $\eta^{1,*}(t,y) \leq \overline{\eta}^{1}$. \\
        Assume that $\eta^{1,*}(t,y) > \overline{\eta}^{1}$, due to the primal feasibility, we have $\eta^{1,*} \leq y$, so  $\overline{\eta}^1 = \inf \mathbb{O}$, and there exists $\overline{\eta}^1_{0} \in \mathbb{O}$ such that $\eta^{1,*} > \overline{\eta}^1_{0} \geq 0$.  
        From the concavity of the value function $V^{p,0}(t,y)$ with respect to $y$, we have
        $$\partial^{+}_y V^{p,0}(t, y - \eta^{1,*}) \geq \partial^{+}_y V^{p,0}(t, y - \overline{\eta}^1_{0} ) > -  (1+k)\gamma {\overline{\eta}^1_{0}}^{\gamma - 1}   >  -  (1+k)\gamma { \eta^{1,*} }^{\gamma - 1}.$$
        Then, we obtain
        $$ (1+k)\gamma { \eta^{1,*} }^{\gamma - 1}  +  \partial^{+}_y V^{p,0}(t, y - \eta^{1,*})  >0 \Longleftrightarrow a_1 = a_2 +  (1+k)\gamma { \eta^{1,*} }^{\gamma - 1} + v > 0   $$
        Based on the Complementary Slackness condition we have  $\eta^{1,*} = 0$, which contradicts $\eta^{1,*} > \overline{\eta}^1_{0} \geq 0  $.

        \medskip
        Second, we show that 
        $\eta^{1,*} \geq \overline{\eta}^{1}$.\\
        If $\eta^{1,*} = y$, we directly have  $\eta^{1,*}(t,y) = y \geq  \overline{\eta}^{1} $ .\\
        Now we assume that $\eta^{1,*} < y$.
        By the Complementary Slackness condition, we have $a_2 = 0$. \\
        Firstly, we show that one inequality holds, which would be used later
        \begin{equation}\label{sub-Vy}
            -(1+k)\gamma { \eta^{1,*} }^{\gamma - 1} \leq  \partial^{-}_y V^{p,0}(t, y - \eta^{1,*}).
        \end{equation}
        Assume $ -(1+k)\gamma { \eta^{1,*} }^{\gamma - 1} >  \partial^{-}_y V^{p,0}(t, y - \eta^{1,*})$, then we have 
        $$   -(1+k)\gamma { \eta^{1,*} }^{\gamma - 1}  - v \geq   -(1+k)\gamma { \eta^{1,*} }^{\gamma - 1} -  \partial^{-}_y V^{p,0}(t, y - \eta^{1,*}) > 0 .$$
        Then based on above inequality, we have $a_2 = a_1 -(1+k)\gamma { \eta^{1,*} }^{\gamma - 1}  - v> 0$ which contradicts with $a_2 = 0$.\\
        Now we assume $\eta^1_{1} \in (\eta^{1,*},y)$.
        Given the concavity of $V^{p,0}(t,y)$ with respect to $y$, we have
        $$- (1+k)\gamma { \eta^1_{1} }^{\gamma - 1} < - (1+k)\gamma { \eta^{1,*} }^{\gamma - 1} \leq \partial^{-}_y V^{p,0}(t, y - \eta^{1,*}) \leq  \partial^{+}_y V^{p,0}(t, y - \eta^1_1).$$
        Then we obtain  $(1+k)\gamma { \eta^1_{1} }^{\gamma - 1} +  \partial^{+}_y V^{p,0}(t, y - \eta^1_1) > 0$. Hence, $\eta^1_1 \in \mathbb{O}$. Due to the arbitrariness of the choice of $\eta^1_1 \in (\eta^{1,*} , y)$, we conclude that $\eta^{1,*} \geq \overline{\eta}^1$. Combined with the conclusion $\eta^{1,*} \leq \overline{\eta}^1$, we have $$\eta^{1,*} = \overline{\eta}^1 = \inf \{\eta \in [0,y] : \partial^{+}_{y}V^{p,0}(t,y - \eta) > - (1+k)\gamma \eta^{\gamma - 1} \}.$$ 
        Finally, we show that
        for each fixed pair of $(t,y) \in [0,T) \times (0,\infty)$, we have 
        $$\eta^{1,*}(t,y) < y \;, \; (t,y) \in [0,T) \times (0,\infty).$$
         Assume that $\eta^{1,*}(t,y) = y$, for fixed $y > 0$, building on the Complementary Slackness condition, we have that $a_{2} \geq 0$, and $a_1 = 0$, implying $$\quad 0= (1+k)\gamma y^{\gamma - 1}  +  v  +a_{2}\; , \;\text{for some } v \in\left[\partial_y^{+} V^{p,0}, \partial_y^{-} V^{p,0}\right]\left(t, 0 \right). $$
         Then we have
         $$\partial_y^{+} V^{p,0}(t, 0) \leq v  = - (1+k)\gamma y^{\gamma - 1}  - a_2 \leq - (1+k)\gamma y^{\gamma - 1} < 0 ,$$ which contradicts with $\partial_y^{+} V^{p,0}(t, 0) \geq 0.$\\
          As a result, we conclude that
         $\eta^{1,*}(t,y) < y\;,\; (t,y) \in [0,T) \times (0,\infty).$
\subsection{Comparison principle}
\begin{lemma}\label{cmp::v0hjbvi}
 (Comparison) Let $u$ and $v$ be respectively an upper-semicontinuous viscosity subsolution and a lower-semicontinuous viscosity supersolution of (\ref{hjbvi:pp:dirichletv0}). Assume further that for $\varphi \in\{u, v\}$, $ - A^{(0,d)}(t) e^{d(t-T)} y^{\gamma} \leq \varphi(t,y) \leq \varphi^{A^{(0,d)},d,\Delta_N}(t, y), ( t, y) \in [0,T)  \times(0,\infty)$ . Then, if $u(t,y) \leq v(t,y) , (t,y) \in [0,T) \times \{ 0\} \cup \{ T\} \times [0,\infty)$, we have $u \leq v$ on $[0,T) \times (0,\infty)$. 
\end{lemma} 

\begin{proof}
\begin{remark}
    In the finite-horizon case, Theorem \ref{thm::FB}, which characterizes the first-best problem, shows that the corresponding value does not provide the desired upper bound. The main difficulty arises from the need to handle the state constraint carefully. In stochastic control problems with exit times, the unique continuous extension of the value function may fail to coincide with the Dirichlet boundary condition (see \cite{barles1998strong} for details).
\end{remark}
 To begin with,  we know that 
    $$ - A^{(0,d)}(t) e^{d(t-T)} y^{\gamma} \leq \varphi(t,y) \leq \varphi^{A^{(0,d)},d,\Delta_N}(t, y)  ,\quad (t,y) \in [0,T)\times (0,\infty).$$
    We note that there is no continuity issue in extending the value function from the interior of the domain $[0, T] \times(0, \infty)$ to the closure $[0, T] \times[0, \infty)$. Consequently, the classical notion of the Dirichlet boundary condition applies.
    Next, we adapt the argument of \cite[Theorem 5.3.3]{pham2009continuous} to our setting.
    
    \medskip
    Firstly, we know that for a relative large $\lambda_1 > 0,\lambda_2 >0 $, the function $\bar{\psi}(t,y) := e^{\lambda_1(T-t)} \lambda_2 \log(1 + y)$ is a strict viscosity supersolution to (\ref{vnpdevi}).
    Then, we define the function
$$v^{\varepsilon} := (1 - \varepsilon)v+\varepsilon \bar{\psi}\;,\; \varepsilon \in(0, 1).$$
    Due to the convexity of the generator $\mathcal{J}^{+}(\cdot,\cdot)$, we verify that $v^{\e}$ is a strict supersolution of \eqref{hjbvi:pp:dirichletv0}.
    Next, following a doubling technique \cite[Theorem $4.4.4$]{pham2009continuous}, for any $\alpha>0$, and $\varepsilon \in(0, 1)$, we introduce the maps
    \begin{equation*}
    \begin{split}
        \Phi_{\alpha, \varepsilon}(t,s,x,y)&:=u(t,x)-v^{\varepsilon}(s,y)- \phi_{\alpha}(t, s, x, y) , (t,s ,x, y) \in [0,T]^{2} \times  \mathbb{R}_{+}^2, \\
        \phi_{\alpha}(t, s, x, y)&:=\frac{1}{2\alpha}\left[|t-s|^2+|x-y|^\nu\right] \;,\; \nu = (\gamma+1) \vee 3, \\
    \end{split}
    \end{equation*}
    where
    \begin{equation*}
      M_{\alpha, \varepsilon}:=\sup _{(t,s ,x, y) \in [0,T]^2 \times \mathbb{R}_{+}^2   }  \Phi_{\alpha, \varepsilon}(t,s,x,y)= \Phi_{\alpha, \varepsilon}(t_{\alpha, \varepsilon},s_{\alpha, \varepsilon},x_{\alpha, \varepsilon},y_{\alpha, \varepsilon}).    
    \end{equation*}
        By standard viscosity–solution arguments \cite[Proposition 3.7]{Crandall:1992:UserGuideViscosity}, there exists a subsequence $(t_{\alpha_n}^{\varepsilon},s_{\alpha_n}^{\varepsilon},x_{\alpha_n}^{\varepsilon},y_{\alpha_n}^{\varepsilon})$, such that
    \begin{align}\label{viscositysubsequence}
    &\lim_{n \to \infty} (t_{\alpha_n}^{\varepsilon},s_{\alpha_n}^{\varepsilon},x_{\alpha_n}^{\varepsilon},y_{\alpha_n}^{\varepsilon}) = (\bar{t}^{\varepsilon} ,\bar{t}^{\varepsilon},\bar{x}^{\varepsilon},\bar{x}^{\varepsilon}) \quad , 
 \lim_{n \to \infty} \phi_{\alpha_n}\left(  t_{\alpha_n}^{\varepsilon},s_{\alpha_n}^{\varepsilon},x_{\alpha_n}^{\varepsilon},y_{\alpha_n}^{\varepsilon}\right) =  0, \\ \nonumber
     &M_{\varepsilon} := \lim_{n \to \infty} M_{\alpha_n,\varepsilon} = u(\bar{t}^{\varepsilon},\bar{x}^{\varepsilon})-v^{\varepsilon}(\bar{t}^{\varepsilon},\bar{x}^{\varepsilon}). 
    \end{align}
    To establish the comparison principle, it is sufficient to verify that, for every $\varepsilon \in (0,1)$, we have $\sup(u-v^{\varepsilon}) \leq 0$. We proceed by contradiction and assume that there exists $\varepsilon \in (0,1)$ such that at $(t_o,y_o) \in [0,T) \times \mathbb{R}^+$, $\eta:= u(t_o,y_o)-v^{\varepsilon}(t_o,y_o)>0$. Then $\forall n \in \mathbb{N}$, we have $\eta \leq M_{\alpha_n, \varepsilon}.$
    
    By the standard Ishii's lemma, there exists an $\mathbb{R}^2$-valued sequence $\left(M_n^{\varepsilon}, N_n^{\varepsilon}\right)_{n \in \mathbb{N}}$, such that
    \begin{equation}\label{parabolic_semijets}
        \begin{split}
            & \left(\frac{1}{\alpha_n}\left(t_{\alpha_n}^{\varepsilon}-s_{\alpha_n}^{\varepsilon}\right), \frac{\nu}{2\alpha_n}\left(x_{\alpha_n}^{\varepsilon}-y_{\alpha_n}^{\varepsilon}\right)^{\nu - 1}, M_n^{\varepsilon}\right) \in \overline{\mathcal{P}}^{2,+} u\left(t_{\alpha_n}^{\varepsilon}, x_{\alpha_n}^{\varepsilon}\right), \\
            & \left(\frac{1}{\alpha_n}\left(t_{\alpha_n}^{\varepsilon}-s_{\alpha_n}^{\varepsilon}\right) , \frac{\nu}{2\alpha_n}\left(x_{\alpha_n}^{\varepsilon}-y_{\alpha_n}^{\varepsilon}\right)^{\nu - 1}, N_n^{\varepsilon}\right)   \in \overline{\mathcal{P}}^{2,-} v^{\varepsilon}\left(s_{\alpha_n}^{\varepsilon}, y_{\alpha_n}^{\varepsilon}\right),
        \end{split}
    \end{equation}
and
\begin{align}\label{boundsecondorder}
        \left(\begin{array}{cc}
M_n^{\varepsilon} & 0 \\
0 & -N_n^{\varepsilon}
\end{array}\right) \leq C^{\varepsilon}_{n} +  \lambda_3 (C^{\varepsilon}_{n})^{2}, 
\end{align}
    where, $\lambda_3 > 0$, 
    $C^{\varepsilon}_n =  a^{\varepsilon}_{n} A ,a^{\varepsilon}_n := \frac{\nu(\nu - 1)}{2\alpha_n}|x_{\alpha_n}^{\varepsilon}-y_{\alpha_n}^{\varepsilon}|^{\nu - 2}  , A:=\left(\begin{array}{cc}
1 & -1 \\
-1 & 1
\end{array}\right).$\\
We fix $\lambda_3 =\frac{1}{||C^{\varepsilon}_n||}$, where $||\cdot||$ is the spectral norm. Next, we multiply the inequality (\ref{boundsecondorder}) by $(1,1)$ to the left and $(1,1)^{\top}$ to the right, and obtain $M_n^{\varepsilon}-N_n^{\varepsilon} \leq 0, \text { for all } n \in \mathbb{N}$.\\
Defining $G(y,a,p,q) := - a -  \mathcal{J}^{+}(p,q) - rpy$, we get
\begin{equation*}
\begin{split}
    &\rho u\left(t_{\alpha_n}^{\varepsilon}, x_{\alpha_n}^{\varepsilon}\right) + G\left(x_{\alpha_n}^{\varepsilon},\frac{1}{\alpha_n}\left(t_{\alpha_n}^{\varepsilon}-s_{\alpha_n}^{\varepsilon}\right), \frac{\nu}{2\alpha_n}\left(x_{\alpha_n}^{\varepsilon}-y_{\alpha_n}^{\varepsilon}\right)^{\nu - 1}, M_n^{\varepsilon}\right) \\
    \leq& 0 
    \leq \rho v^{\varepsilon}\left(s_{\alpha_n}^{\varepsilon}, y_{\alpha_n}^{\varepsilon}\right) + G\left(y_{\alpha_n}^{\varepsilon}, \frac{1}{\alpha_n}\left(t_{\alpha_n}^{\varepsilon}-s_{\alpha_n}^{\varepsilon}\right) , \frac{\nu}{2\alpha_n}\left(x_{\alpha_n}^{\varepsilon}-y_{\alpha_n}^{\varepsilon}\right)^{\nu - 1} , N_n^{\varepsilon}\right).  \\
\end{split}  
\end{equation*}
Therefore, we obtain
\begin{align*} 
     &\rho \eta \\
     \leq& \rho\left
    (u\left(t_{\alpha_n}^{\varepsilon}, x_{\alpha_n}^{\varepsilon}\right) -v^{\varepsilon}\left(s_{\alpha_n}^{\varepsilon}, y_{\alpha_n}^{\varepsilon}\right) \right) \\
    \leq&  G\left(y_{\alpha_n}^{\varepsilon}, \frac{1}{\alpha_n}\left(t_{\alpha_n}^{\varepsilon}-s_{\alpha_n}^{\varepsilon}\right)  , \frac{\nu}{2\alpha_n}\left(x_{\alpha_n}^{\varepsilon}-y_{\alpha_n}^{\varepsilon}\right)^{\nu - 1} , N_n^{\varepsilon}\right) \\
    &- G\left(x_{\alpha_n}^{\varepsilon},\frac{1}{\alpha_n}\left(t_{\alpha_n}^{\varepsilon}-s_{\alpha_n}^{\varepsilon}\right), \frac{\nu}{2\alpha_n}\left(x_{\alpha_n}^{\varepsilon}-y_{\alpha_n}^{\varepsilon}\right)^{\nu - 1}, M_n^{\varepsilon}\right) \\
    = & \mathcal{J}^{+}\left( \frac{\nu}{2\alpha_n}\left(x_{\alpha_n}^{\varepsilon}-y_{\alpha_n}^{\varepsilon}\right)^{\nu - 1},M_n^{\varepsilon} \right) - \mathcal{J}^{+}\left( \frac{\nu}{2\alpha_n}\left(x_{\alpha_n}^{\varepsilon}-y_{\alpha_n}^{\varepsilon}\right)^{\nu - 1}  , N_n^{\varepsilon}\right) +   \frac{r\nu}{2\alpha_n}\left(x_{\alpha_n}^{\varepsilon}-y_{\alpha_n}^{\varepsilon}\right)^{\nu }\\
    \leq &  \mathcal{J}^{+}\left( \frac{\nu}{2\alpha_n}\left(x_{\alpha_n}^{\varepsilon}-y_{\alpha_n}^{\varepsilon}\right)^{\nu - 1},N_n^{\varepsilon} \right) - \mathcal{J}^{+}\left( \frac{\nu}{2\alpha_n}\left(x_{\alpha_n}^{\varepsilon}-y_{\alpha_n}^{\varepsilon}\right)^{\nu - 1}  , N_n^{\varepsilon}\right) +   \frac{r\nu}{2\alpha_n}\left(x_{\alpha_n}^{\varepsilon}-y_{\alpha_n}^{\varepsilon}\right)^{\nu }\\
     \leq&  \frac{r\nu}{2\alpha_n}\left(x_{\alpha_n}^{\varepsilon}-y_{\alpha_n}^{\varepsilon}\right)^{\nu }.
\end{align*}
The third inequality follows from the fact that $\mathcal{J}^{+}$ is Lipschitz continuous and nondecreasing with respect to its second variable when $\{M_n^{\varepsilon}\}_n \leq 0 $, same as in \cite[Lemma B.1]{possamai2024there}, \cite[Theorem 5.7 (ii)]{possamai2025golden}. Letting $n \to  \infty$,
we obtain the desired contradiction for all $  \varepsilon \in (0,1)$, which gives the required result by sending $\varepsilon$ to zero. The latter completes the proof.
\qed
\end{proof}

\begin{lemma}\label{cmp::vnhjbvi}
(Comparison) Let $u$ and $v$ be respectively an upper-semicontinuous viscosity subsolution and a lower-semicontinuous viscosity supersolution of $(\ref{hjbvi:pp:dirichlephik})$, respectively.
\begin{equation}\label{hjbvi:pp:dirichlephik}
\begin{split}
 \min\{v^{p,n} - \mathcal{M} V^{p,n-1} , -v^{p,n}_t + \mathcal{G}v^{p,n} ,  - v^{p,n}_{yy}\} &= 0, \quad  [0,T) \times (0,\infty),   \\
v^{p,n}(T,y) &= g^{n}(y), \quad y\in [0,\infty), \\
v^{p,n}(t,0) &= 0, \quad\quad t \in [0,T]. 
\end{split}   
\end{equation}
Assume that for each $\varphi \in \{u,v\}$

$$
- A^{(n,d)}(t) e^{d(t-T)} y^{\gamma} \leq \varphi(t,y) \leq \varphi^{A^{(n,d)},d,\Delta_N}(t, y) ,( t, y) \in [0,T)  \times(0,\infty) \;, \; n \in \{1,\cdots,N\}.
$$
 Then, if $u(t,y) \leq v(t,y) , (t,y) \in [0,T) \times \{ 0\} \cup \{ T\} \times [0,\infty)$, we have $u \leq v$ on $[0,T) \times (0,\infty)$. 
\end{lemma}
\begin{proof}
We adapt the arguments presented in \cite[Theorem $4.4.4$]{pham2009continuous}, and  \cite[ Theorem $5.2.1$]{pham2009continuous}. Similar to the proof of Lemma \ref{cmp::v0hjbvi}, we introduce the following strict supersolution 
$$
v^{\varepsilon}:=(1 - \varepsilon)v+\varepsilon \bar{\psi}\;,\; \varepsilon \in(0, 1).
$$
Next, we follow the doubling technique introduced in \cite[Theorem $4.4.4$]{pham2009continuous}. For any $\alpha>0$ and $\varepsilon \in(0, 1)$ we define the following maps
$$
\begin{aligned}
\Phi_{\alpha, \varepsilon}(t, s, x, y) & :=u(t, x)-v^{\varepsilon}(s, y)-\phi_\alpha(t, s, x, y),(t, s, x, y) \in[0, T]^2 \times \mathbb{R}_{+}^2 \\
\phi_\alpha(t, s, x, y) & :=\frac{1}{2 \alpha}\left[|t-s|^2+|x-y|^\nu\right] \quad, \quad \nu=(\gamma+1) \vee 3,
\end{aligned}
$$
where
$$
M_{\alpha, \varepsilon}:=\sup _{(t, s, x, y) \in[0, T]^2 \times \mathbb{R}_{+}^2} \Phi_{\alpha, \varepsilon}(t, s, x, y)=\Phi_{\alpha, \varepsilon}\left(t_{\alpha, \varepsilon}, s_{\alpha, \varepsilon}, x_{\alpha, \varepsilon}, y_{\alpha, \varepsilon}\right).
$$
Using standard viscosity arguments \cite[Proposition 3.7]{Crandall:1992:UserGuideViscosity}, there exists a subsequence $\left(t_{\alpha_n}^{\varepsilon}, s_{\alpha_n}^{\varepsilon}, x_{\alpha_n}^{\varepsilon}, y_{\alpha_n}^{\varepsilon}\right)_{n\geq 1}$ such that
$$
\begin{aligned}
& \lim _{n \rightarrow \infty}\left(t_{\alpha_n}^{\varepsilon}, s_{\alpha_n}^{\varepsilon}, x_{\alpha_n}^{\varepsilon}, y_{\alpha_n}^{\varepsilon}\right)=\left(\bar{t}^{\varepsilon}, \bar{t}^{\varepsilon}, \bar{x}^{\varepsilon}, \bar{x}^{\varepsilon}\right) \;,\; \lim _{n \rightarrow \infty} \phi_{\alpha_n}\left(t_{\alpha_n}^{\varepsilon}, s_{\alpha_n}^{\varepsilon}, x_{\alpha_n}^{\varepsilon}, y_{\alpha_n}^{\varepsilon}\right)=0,\\
& M_{\varepsilon}:=\lim _{n \rightarrow \infty} M_{\alpha_n, \varepsilon}=u\left(\bar{t}^{\varepsilon}, \bar{x}^{\varepsilon}\right)-v^{\varepsilon}\left(\bar{t}^{\varepsilon}, \bar{x}^{\varepsilon}\right).
\end{aligned}
$$
Suppose that there exists $\varepsilon \in (0,1)$ such that at $(t_o,y_o) \in [0,T) \times \mathbb{R}^+$, we have $\eta:= u(t_o,y_o)-v^{\varepsilon}(t_o,y_o)>0$. Then, for all $ n \geq 1$, we have $\eta \leq M_{\alpha_n, \varepsilon}.$

Next, we distinguish cases according to the asymptotic behavior of the subsequence $(u(t_{\alpha_n}^{\varepsilon}, x_{\alpha_n}^{\varepsilon}) -  \mathcal{M}V^{p,n-1}(t_{\alpha_n}^{\varepsilon}, x_{\alpha_n}^{\varepsilon}))_{n\geq 1}$.
\begin{enumerate}[label=\roman*)]
    \item If $u(t_{\alpha_n}^{\varepsilon}, x_{\alpha_n}^{\varepsilon}) - \mathcal{M}V^{p,k-1}(t_{\alpha_n}^{\varepsilon}, x_{\alpha_n}^{\varepsilon}) > 0$ for sufficiently large $n$, then we fall into the case already treated in Lemma \ref{cmp::v0hjbvi}. Hence, the same argument applies, and the desired inequality follows.
    \item Otherwise, up to a subsequence, $u(t_{\alpha_n}^{\varepsilon}, x_{\alpha_n}^{\varepsilon}) - \mathcal{M}V^{p,n-1}(t_{\alpha_n}^{\varepsilon}, x_{\alpha_n}^{\varepsilon}) \leq 0$ for all $n\geq 1$. Since $ v^{\varepsilon}(t_{\alpha_n}^{\varepsilon}, y_{\alpha_n}^{\varepsilon}) - \mathcal{M}V^{p,n-1}(t_{\alpha_n}^{\varepsilon}, y_{\alpha_n}^{\varepsilon}) > 0$, by the strict supersolution property, it follows that

$$
 \eta \leq u\left(t_{\alpha_n}^{\varepsilon}, x_{\alpha_n}^{\varepsilon}\right) -v^{\varepsilon}\left(s_{\alpha_n}^{\varepsilon}, y_{\alpha_n}^{\varepsilon}\right) \leq \mathcal{M}V^{p,n-1}(t_{\alpha_n}^{\varepsilon}, x_{\alpha_n}^{\varepsilon})  - \mathcal{M}V^{p,n-1}(t_{\alpha_n}^{\varepsilon}, y_{\alpha_n}^{\varepsilon})  .
$$

By letting $n \rightarrow \infty$ and using the continuity of $\mathcal{M}V^{p,n-1}$, we obtain the required contradiction $\eta \leq 0$.
\end{enumerate}
In conclusion, we have the desired contradiction for all $\epsilon \in (0,1)$. Finally, we get the desired result by sending $\varepsilon \to 0$.
The proof is complete.\qed

\end{proof}

\bibliographystyle{apalike}
\bibliography{main}

\end{document}